\documentclass[final,12pt]{article}

\usepackage{verbatim}
\usepackage{latexsym}
\usepackage{amsmath,amsthm,amsfonts}
\usepackage{eucal}
\usepackage{cases}
\usepackage[notref,notcite]{showkeys}
\usepackage{ulem}

\usepackage{mathtext}
\usepackage[cp1251]{inputenc}
\usepackage[T2A]{fontenc}
\usepackage{graphicx}
\usepackage{amsmath}
\usepackage{amssymb}
\usepackage{amsxtra}
\usepackage{latexsym}
\usepackage{ifthen}

\usepackage{fullpage}
\usepackage[normal,footnotesize]{caption}

\newtheorem{theorem}{Theorem}
\newtheorem{definition}{Definition}[section]
\newtheorem{lemma}{Lemma}[section]
\newtheorem{proposition}{Proposition}[section]

\numberwithin{equation}{section}

\newcommand{\Int}{\int\limits_{-\infty}^{\infty}}

\newcommand{\R}{\mathbb{R}}
\newcommand{\eps}{\varepsilon}

\begin{document}

\title {\bf Research Announcement: Finite--time Blow Up and Long--wave Unstable Thin Film Equations}

\author{{\normalsize\bf Marina Chugunova, M.C. Pugh, Roman M. Taranets}
\smallskip}

\def\lhead{R.M. Taranets}

\date{\today}

\maketitle

\setcounter{section}{0}

\begin{abstract}
We study short--time existence, long--time existence, finite speed
of propagation, and finite--time blow--up of nonnegative solutions
for long-wave unstable thin film equations $h_t = -a_0(h^n
h_{xxx})_x - a_1(h^m h_x)_x$ with $n>0$, $a_0 > 0$, and $a_1 >0$.
The existence and finite speed of propagation results extend those
of [Comm Pure Appl Math 51:625--661, 1998].  For $0<n<2$ we prove
the existence of a nonnegative, compactly--supported, strong
solution on the line that blows up in finite time.  The
construction requires that the initial data be nonnegative,
compactly supported in $\R^1$, be in $H^1(\R^1)$, and have
negative energy.  The blow-up is proven for a large range of
$(n,m)$ exponents and extends the results of [Indiana Univ Math J
49:1323--1366, 2000].
\end{abstract}


\textbf{keywords:} {fourth-order degenerate parabolic equations,
thin liquid films, long--wave unstable, finite--time blow--up,
finite speed of propagation}

\section{Introduction} \label{A}

Numerous articles have been published since the early sixties
concerning the problem of developing finite-time singularities by
solutions of nonlinear parabolic equations (see the survey papers of
Levine \cite{Levine} and of Bandle and Brunner \cite{BunBru}). Such
problems arise in various applied fields such as combustion theory,
the theory of phase separation in binary alloys,
population dynamics and incompressible fluid flow.

Whether or not there is a finite--time singularity, such as $\|
u(\cdot,t) \|_\infty \to \infty$ as $t \to T^* < \infty$, is
strongly affected by the nonlinearity in the PDE.  For example,
consider the
semilinear heat equation on the line:
\begin{equation}
\label{A.heat} u_t = u_{xx} + u^p,
\end{equation}
where $u$ is real--valued.  If $p \leqslant 1$ then a solution of an
initial value problem exists for all time. If $1 < p \leqslant 3$, then any
non-trivial solution blows up in finite time.  If $p >3$ then some
initial data yield solutions that exist for all time and other initial
data result in solutions that have finite--time singularities.  The
manner in which solutions blow up is well understood computationally
and analytically.  The blow up is of a focussing type: there are
isolated points in space around which the graph of the solution
narrows and becomes taller as $t \uparrow T^*$, converging to delta
functions centered at the blow--up points.  As $t \uparrow T^*$, the
behaviour of the solution near the blow--up point(s) becomes more and
more self--similar.  Proving this convergence to a self--similar
solution uses the maximum principle, which doesn't hold for
fourth--order equations like the one we study in this article.

We study the longwave-unstable generalized thin film equation,
\begin{equation}
\label{A.gen.thin-film}
h_t = - a_0\,(|h|^n\,h_{xxx})_x -
a_1\,(|h|^m\,h_x)_x,
\end{equation}
where $a_0 > 0$, $a_1 > 0$, and $h$ is real valued.
Perturbing a constant steady state slightly,
$h_0(x) = \bar{h} + \epsilon h_1(x,0) =
\bar{h} + \epsilon \cos(\xi x + \phi)$, and linearizing
the equation about $\bar{h}$, the small perturbation $h_1(x,t)$
will (approximately)
satisfy
$h_t = -
a_0 |\bar{h}|^n h_{xxxx} - a_1 |\bar{h}|^m h_{xx}$.  Hence the
constant steady state is linearly unstable to long wave perturbations:
\begin{equation} \label{unstable_band}
\xi^2 < |\bar{h}|^{m-n} a_1/a_0 \quad \Longrightarrow \quad
h_1(x,t) \sim e^{-a_0 \xi^2 |\bar{h}|^n \left( \xi^2 - \tfrac{a_1}{a_0} |\bar{h}|^{m-n} \right) t} \cos(\xi x + \phi)
\quad \mbox{grows}.
\end{equation}

%

Such equations arise in the modelling of fluids and materials.  For
example, equation (\ref{A.gen.thin-film}) with $n=m=1$ describes a
thin jet in a Hele-Shaw cell \cite{Constantin_et_al} where $h$
represents the thickness of the jet and $x$ is the axial direction; if
$n = 3$ and $m = -1$ equation (\ref{A.gen.thin-film}) describes Van
der Waals driven rupture of thin films \cite{wit1} where $h$
represents the thickness of the film; if $n= m = 3$ the equation models
fluid droplets hanging from a ceiling \cite{Enrhard} with $h$
representing the thickness of the film, and finally if $n = 0$ and $m =
1$ the equation is a modified Kuramoto-Sivashinsky equation and
describes the solidification of a hyper-cooled melt \cite{Bernoff}
where $h$ desribes the deviation from flatness of a near planar
interface.  We note that in the first three cases the solution $h$
must be nonnegative for the model to make physical sense.

Hocherman and Rosenau \cite{HoRo} considered whether or not
equation (\ref{A.gen.thin-film}) could have solutions that blow up
in finite time.  They conjectured that if $n>m$ then solutions might
blow up in finite time but if $n<m$ they would exist for all time.
Indeed,
this conjecture is natural if one considers the linear stability
of a constant steady state $\bar{h}$: if $n > m$ the unstable
band (\ref{unstable_band}) grows as $\bar{h} \to \infty$ suggesting
that $n=m$ is critical.

Hocherman and Rosenau considered general, real--valued solutions.
However, if $n>0$ equation (\ref{A.gen.thin-film}) may have solutions
that are nonnegative for all time.  Bertozzi and Pugh \cite{BP1}
proposed that\footnote{In fact, the article considers
(\ref{A.gen.thin-film}) with general coefficients: $f(h)$ and $g(h)$
instead of $|h|^n$ and $|h|^m$ respectively. In the following, for
simplicity, their results are discussed for the power--law case.}  if
the boundary conditions are such that the mass, $\int h(x,t) \; dx$,
is conserved then mass conservation, combined with the nonnegativity
results in a different balance: $m = n+2$ instead of $m=n$.
For such cases they introduced the regimes
$$
\begin{cases}
m < n + 2 \Longrightarrow \mbox{subcritical regime} \\
m = n + 2 \Longrightarrow \mbox{critical regime} \\
m > n + 2 \Longrightarrow \mbox{supercritical regime}
\end{cases}
$$

In \cite{BP1}, Bertozzi and Pugh considered equation
(\ref{A.gen.thin-film}) on a finite interval with periodic boundary
conditions.  For a subset of the subcritical ($m < n+2$) regime they
proved some global-in-time results.  Specifically, they proved that
given positive initial data, $h_0 > 0$, there is a nonnegative weak
solution of (\ref{A.gen.thin-film}) that exists for all time if  $0 < n \leqslant m < n+2$.
By restricting $n$ further to $0 < n < 3$
they can consider nonnegative initial data, $h_0 \geqslant 0$.  They prove
that there a nonnegative weak solution that exists for all
time and also prove the local entropy bound needed for the finite
speed of propagation proof for $0<n<2$.
For the critical
($m=n+2$) regime, they prove that the above results will hold if the
mass $\int h_0 \; dx$ is sufficiently small.  Also, they provided
numerical simulations suggesting that other initial data can result in
solutions that blow up in finite time and conjectured that this is
also true for the supercritical ($m>n+2$) regime.


In \cite{BP2}, they considered equation
(\ref{A.gen.thin-film}) on the line and found some analytical results
for the critical and supercritical ($m \geqslant n+2$) regimes in the
special case of $n=1$.
They introduced a large class of ``negative
energy'' initial data and proved that given initial data $h_0$ with
compact support and negative energy there is a nonnegative weak
solution $h$ that blows up in finite time: there
is a time $T^* < \infty$ such that the weak solution
$h(\cdot,t)$
exists
on $[0,T^*)$ and
$$
\limsup_{t \to T^*}
\|h(\cdot,t)\|_{L^\infty} = \infty
\qquad \mbox{and} \qquad
\limsup_{t \to T^*}
\|h(\cdot,t)\|_{H^1} = \infty.
$$
%
%
The blow--up time $T^*$
depends only on $m$ and $H^1$-norm of the initial data.  We note that
uniqueness of nonnegative weak solutions of (\ref{A.gen.thin-film}) is
an open problem.  Indeed, there are simple counterexamples to
uniqueness for the simplest equation $h_t = - (h^n h_{xxx})_x$ (see, e.g. \cite{B2})
although it
is hypothesized that solutions are unique if one considers
the question within a sufficiently restrictive class of weak
solutions.  For this reason, one cannot exclude the possibility that
the initial data $h_0$ might also be achieved by a different weak
solution, one that exists for all time.
%

Their proof relied on a second moment argument, found
formally by Bernoff \cite{Bern98}: if $h$ is a smooth
compactly-supported solution of (\ref{A.gen.thin-film}) on $[0,T^*)$
then the second moment of $h$ satisfies
\begin{equation} \label{second_moment_n1}
\int\limits_{-\infty}^{\infty} x^2 h(x,t) \; dx
\leqslant \int\limits_{-\infty}^{\infty} x^2 h_0(x) \; dx + 6 \mathcal{E}(0) \, t
\end{equation}
holds for all $t \in [0,T^*)$.  Here, $\mathcal{E}(0)$ is the energy
of the initial data:
\begin{equation} \label{energy_is}
\mathcal{E}(0) :=  \Int \left\{ \tfrac{a_0}{2} {h_0}_x^2(x)
- \tfrac{a_1}{m(m+1)} h_0^{m+1}(x) \right\} \; dx.
\end{equation}
As a result, there could never be a global-in-time nonnegative smooth
solution with negative-energy initial data: for such a solution the
left--hand side of (\ref{second_moment_n1}) would always be
nonnegative but the right--hand side would become negative in finite
time.
(This argument is strictly formal because, to date, no-one has
constructed nonnegative, compactly-supported, {\it smooth} solutions on the
line.)  The blow-up result
is found by first proving the short-time existence
of a nonnegative, compactly-supported, weak solution: it
exists on $[0,T_0]$ where the larger $\|h_0\|_{H^1}$ is, the smaller
$T_0$ will be.  Also, the constructed solution satisfies the second
moment inequality (\ref{second_moment_n1}) at time $t = T_0$.
By ``time--stepping'' the short-time existence result, they
construct a solution on $[0,T^*)$ such that
the second moment inequality (\ref{second_moment_n1}) holds
at a sequence of times $T_i$ with $T_i \to T^*$.  It then follows
immediately
that $T^*$ must be finite and therefore the $\limsup_{t \to T^*} \| h(\cdot,t)
\|_{H^1}$ must be infinite.


\subsection*{Outline of results}

The main results of this paper are: short-time existence of nonnegative strong
solutions on $[-a,a]$ given nonnegative initial data, finite speed of propagation for these solutions
if their initial data had compact support within $(-a,a)$, and finite-time blow-up for solutions of
the Cauchy problem that have initial data with negative energy.

First, we consider equation (\ref{A.gen.thin-film}) on a bounded
interval $[-a,a]$ with periodic boundary conditions.  Given
nonnegative initial data that has finite ``entropy'', in Theorem
\ref{C:Th1} we prove the short-time existence of a nonnegative weak
solution if $n > 0$ and $m \geqslant n/2$.  The solution is a ``generalized
weak solution'' as described in Section \ref{B} and the entropy is
introduced in Section \ref{C}.  Additional
regularity is proven in Theorem \ref{C:Th2}: there is a second type of
entropy such that if this ``$\alpha$-entropy'' is also finite for the initial data
then there is a ``strong solution'' which satisfies Theorem
\ref{C:Th1} and also has some additional regularity.  We note that the
work \cite{BP1} described above was primarily
concerned with long-time\footnote{
Throughout this article, we use phrases like ``long-time'',
``global-in-time'' and ``exist for all time'' as shorthand for the
types of large-time results that have been proven in the thin film
literature to date: given a time $T < \infty$ there is a solution
$h(\cdot,t)$ for $t \in [0,T]$.  Specifically, $T$ can be taken
arbitrarily large.
}
existence results: for this reason the
authors only
addressed the existence theory for the subcritical ($m < n+2$) case.
However, given finite-entropy initial data their methods easily imply
a short-time result for general $n > 0$ as long as $m \geqslant n$.  Our
advance is prove the results for the larger range of $m \geqslant n/2$.
%
%
The left plot in Figure \ref{thm1_region} presents the $(n,m)$
parameter range for which our short-time existence results hold.  The
darker region represents the parameters for which the methods of
\cite{BP1} would have yielded the results.  The lighter region
represents the extended area where our methods also yield results.
\begin{figure}
\begin{center}
\includegraphics[width=3.5cm] {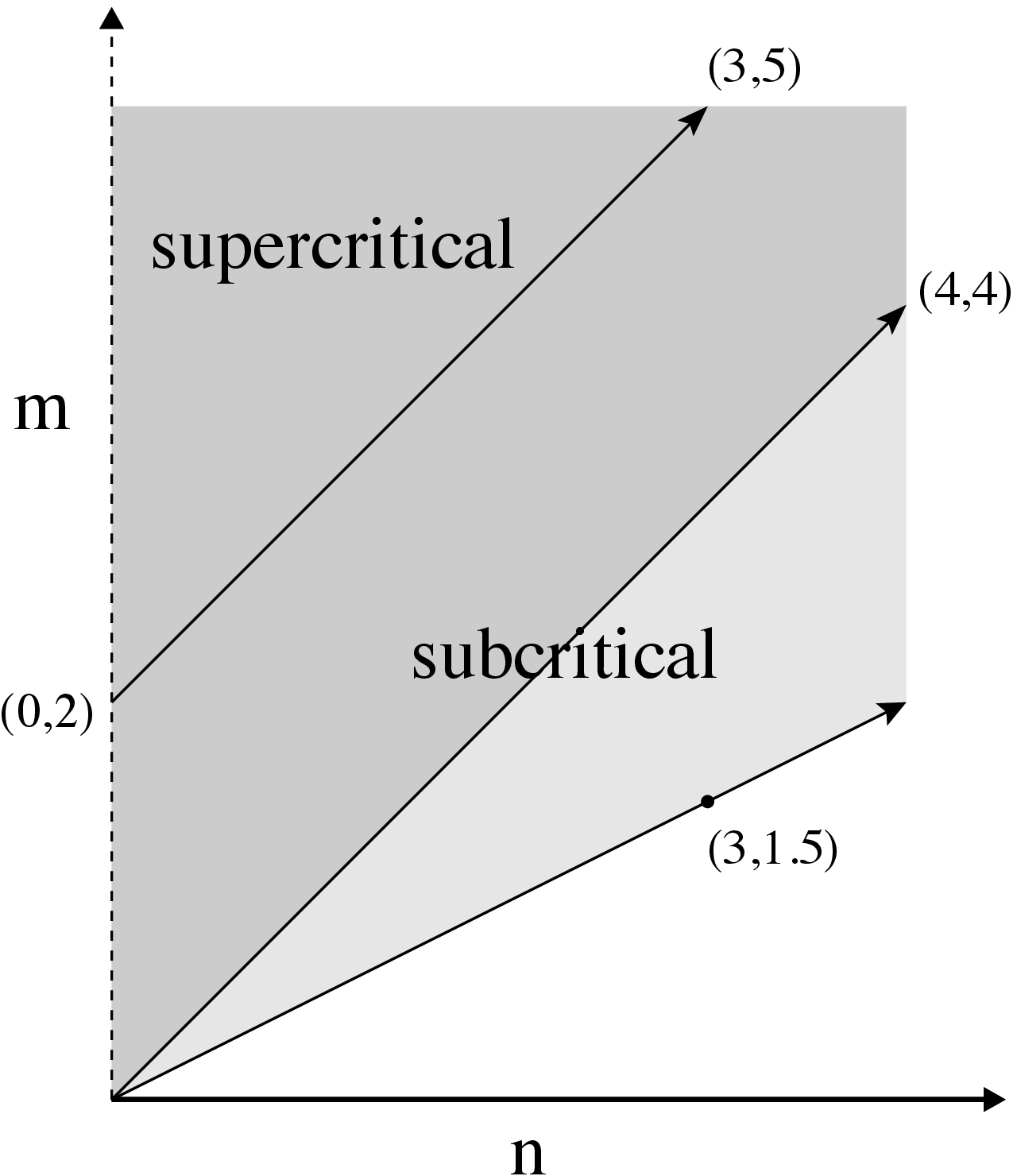}
\hspace{.5in}
\includegraphics[width=3.5cm] {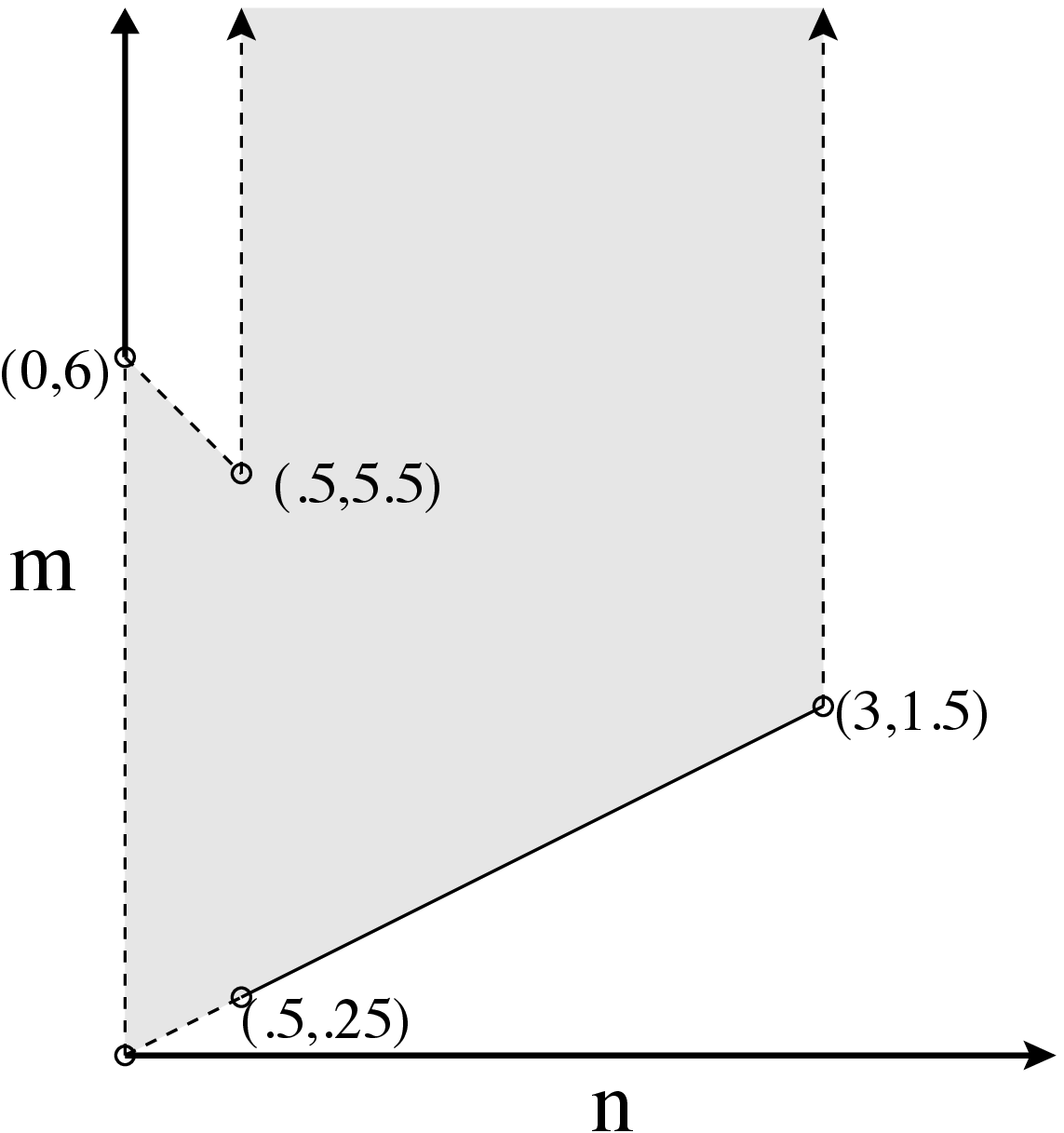}
\hspace{.5in}
\includegraphics[width=4cm] {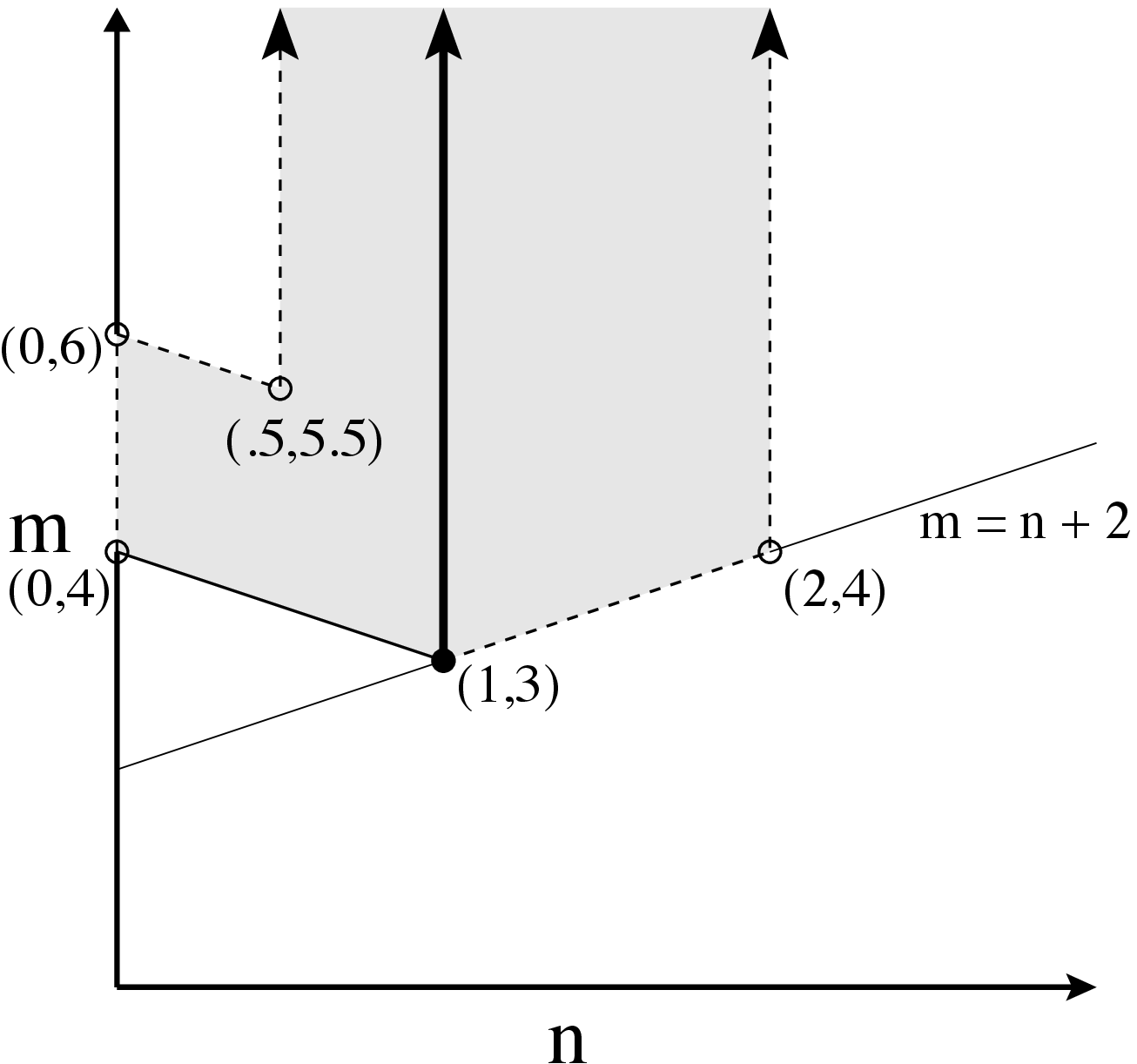}
\end{center}
\vspace{-.2in}
\caption{\label{thm1_region}
Left: the
darker
region represents the exponents $(n,m)$ for which the methods of \cite{BP1} yield
short-time
existence results like Theorems \ref{C:Th1} and \ref{C:Th2}.  The
larger, lighter region represents the additional exponents for which we've proved
our theorems.  Middle: For $(n,m)$ in the shaded region, we prove
that the strong solutions of Theorem \ref{C:Th2} can have
finite speed of propagation (Theorem \ref{C:Th4}).
Right: For $(n,m)$ in the shaded region, we construct solutions on the line
that blow up in finite time (Theorem \ref{C:Th3}).  The heavy solid line denotes the $(n,m)$ for which
this had already been done by \cite{BP2}.
}
\end{figure}

In Theorem \ref{C:Th4} we prove that if the initial data $h_0$ has
compact support then the strong solution of
Theorem \ref{C:Th2} will have finite speed of propagation.
Specifically, if the support of the initial data satisfies $\mbox{supp}(h_0) \subseteq [-r_0,r_0] \subset
(-a,a)$ then there is a nondecreasing function $\Gamma(t)$ and
a time $T_{speed}$ such
that
$\mbox{supp}(h(\cdot,t)) \subseteq
[-r_0-\Gamma(t),r_0+\Gamma(t)] \subset(-a,a)$
for every time $t \in [0,T_{speed}]$.  For $0<n<2$, we further prove that
there is a constant $C$ such that
$\Gamma(t) \leqslant C t^{1/(n+4)}$.  This power law behaviour is
the same as has been found for $h_t = - (|h|^n h_{xxx})_x$
for $0<n<2$ \cite{B6} and $2 \leqslant n < 3$ \cite{HulShish};
it corresponds to the rate of expansion of the self--similar
source--type solution.

The middle plot in Figure \ref{thm1_region} presents the $(n,m)$
parameter range for which we were able to prove finite speed of
propagation.  While we were successful in proving finite speed of
propagation for the entire expected range of $m \geqslant n/2$ if $1/2 < n
< 3$, if $0 < n \leqslant 1/2$ we could prove it for only a subset of $m
\geqslant n/2$.  This technical obstruction is discussed further in Section
\ref{G}.  The $n$ values are restricted to $n < 3$ because for $n \geq
3$ if initial data has compact support in $(-a,a)$ then it will have
infinite entropy and will not be admissible initial data for Theorems
\ref{C:Th1} and \ref{C:Th2}.

To prove Theorem \ref{C:Th4} we start by proving a local entropy
estimate similar to that of \cite{B6} for $0<n<2$ and a local
energy estimate similar to that of \cite{BSS, HulShish} for $1/2 < n <
3$.  Using these and well-chosen ``localization'' (or ``test'') functions, we find
systems of functional inequalities.  In \cite{G2}, Giacomelli and
Shishkov proved an extension of Stampacchia's lemma to systems of
inequalities, we use this to then finish the proof.

A Stampacchia-like lemma for a {\it single} inequality was used by
\cite{HulShish} to prove finite speed of propagation for $h_t = -
(|h|^n h_{xxx})_x$ and \cite{dPGG_2001} proved an extension of
Stampacchia's lemma (also for a single inequality) to study waiting time
phenomena.  Similar approaches were subsequently used
to study finite speed of propagation and waiting time phenomena for
related equations, see
\cite{Ansini2004,dPGG_2003,GG_2002,GG_2004,TarShish_2006,GG_2006,Tar_2002,Tar_2003,Tar_2006}.
Further, there are finite speed of propagation and waiting time results
\cite{TarShish_2006,Shish_2007}
that
use the
extension of Stampacchia's lemma to systems of \cite{G2} as well as
other types of extensions to systems \cite{Tar_2006}.


Having proven finite speed of propagation
in Theorem \ref{C:Th4}, we use this to prove a short-time existence
result for the Cauchy problem.  Specifically, for the range of exponents
$(n,m)$ of Theorem \ref{C:Th4} given nonnegative initial data with
bounded support in $\R^1$ we construct
a nonnegative, compactly supported strong solution on $\R^1 \times [0,T_0]$
that satisfies the bounds and regularity of Theorems \ref{C:Th1} and
\ref{C:Th2} (with those bounds taken over $\R^1$ rather than $\Omega$).
The larger the $H^1$ norm of the initial data, the shorter the time
$T_0$.
In Lemma \ref{F:lem0} we
prove that for a subset of these exponents (see Figure \ref{thm1_region}),
the entropy of the solution satisfies a second-moment
inequality at time $T_0$
$$
e^{-\widetilde{B}(T_0)} \int\limits_{-\infty}^{\infty} {x^2 \widetilde{G}_0 (h(x,T_0))\,dx} \leqslant
 \int\limits_{-\infty}^{\infty} {x^2 \widetilde{G}_{0}
(h_0)\,dx} +
\int \limits_{0}^{T_0} {e^{-\widetilde{B}(t)}
\left(
k_1 \mathcal{E}_{0}(0)  + k_2
\int\limits_{-\infty}^{\infty}  {x^2 h^2_{xx} \,dx} \right) \,
dt}
$$
where $k_1 = 2(4 -n)$, $k_2 =3a_0(n-1)/2$,
$\mathcal{E}_0(0)$
is the energy of the initial data (\ref{energy_is}),
$$
\widetilde{G}_0(z) = \tfrac{1}{ 2 -n }z^{ 2 -n},
\qquad \mbox{and} \qquad
\widetilde{B}(t)=
\tfrac{a_1^2 |( 1 - n)( 2- n)|}{2a_0(m -n+1)^2} \int\limits_0^t
{\|h(.,\tau)\|_{L^{\infty}(\Omega)}^{2m -n}\,d\tau}.
$$
Note that if $n=1$ then the second-moment inequality for the entropy
reduces to the
second moment inequality (\ref{second_moment_n1})
used by \cite{BP2}.

``Time-stepping'' this existence result, we construct a nonnegative, compactly supported,
strong solution on $\R^1 \times [0,T^*)$ such that
our second-moment inequality for the entropy
holds for a sequence of times $T_i \to T^*$.  The left-hand side of the inequality is
nonnegative.  If the initial data has negative energy, then the second term on the
right-hand side has an integrand which has the possibility of becoming negative.
In Theorem \ref{C:Th3} we prove that for
a range of $(n,m)$ values (see Figure \ref{thm1_region}) that if $T^* = \infty$ then
the constructed solution would yield a right-hand side that becomes negative in
finite time,
an impossibility.  Hence $T^*$ must be finite and therefore
the $\limsup_{t \to T^*} \| h(\cdot,t) \|_{H^1}$ must be infinite.

Ideally, we would have proven finite speed of propagation for all $(n,m)$ with
$0 < n < 3$ and all $n/2 \leqslant m$ and would have proven the finite time blow up result for
all $(n,m)$ with $0<n<3$ and all $m \geqslant n+2$.   We believe the obstructions
are technical ones.

We close by noting that our blow--up result does not give qualitative information
such as proving that there's a focussing singularity,
as suggested by numerical simulations \cite{BP1}.  However, there has been
some detailed study of the critical regime $m=n+2$.  There, a critical
mass $M_c$ has been identified \cite{WBB_2004},
there are self-similar, compactly-supported, source-type solutions with masses
in $[0,M_c)$ \cite{Beretta_1997}, and there are self-similar, compactly supported, solutions with masses
in $(M_c,M_1)$  that blow up in finite time \cite{SP_2005}.  Further, these self-similar blow-up solutions have been shown to
be linearly stable \cite{S_2009}.

\section{Generalized weak solution} \label{B}


We study the existence of a nonnegative weak  solution, $h(x,t)$, of the
initial--boundary  value problem
\begin{numcases}
{(\textup{P})}
h_t  +
\left( f(h) \left( a_0 h_{xxx} + a_1 D''(h) h_x \right) \right)_x = 0 \text{ in }Q_T, \hfill \label{B:1}\\
\tfrac{\partial^{i} h}{\partial x^i}(-a,t) = \tfrac{\partial^{i}
h}{\partial x^i}(a,t) \text{ for }
t > 0,\,i=\overline{0,3}, \hfill \label{B:2}\\
\qquad  \qquad h(0,x) = h_0 (x) \geqslant 0, \hfill \label{B:3}
\end{numcases}
%
where $a_0 > 0$,  $a_1 \geqslant 0$, $f(h) = |h|^n$, $D''(h) = |h|^{m-n}$,
$n > 0$, $m > 0$,
$\Omega = (-a,a)$, and $Q_T = \Omega \times (0,T)$.
We consider a weak solution like that considered in \cite{B2,B6}:
%
%
\begin{definition}[generalized weak solution] \label{B:defweak}
Let $n > 0$, $m > 0$, $a_0 > 0$, and $a_1 \geqslant 0$. A function
$h$ is a generalized weak solution of the problem $(\textup{P})$ if
\begin{align}
& \label{weak1}
h \in C^{1/2,1/8}_{x,t}(\overline{Q}_T) \cap L^\infty (0,T; H^1(\Omega )),\\
& \label{weak-d} h_t \in L^2(0,T; (H^1(\Omega))'),\\
& \label{weak2} h \in C^{4,1}_{x,t}(\mathcal{P}), \,\,\,
\sqrt{f(h)} \, \left( a_0 h_{xxx} + a_1 D''(h) h_x \right) \in
L^2(\mathcal{P}), \,\,
\end{align}
where $\mathcal{P} = \overline{Q}_T \setminus ( \{h=0\} \cup \{t=0\})$
and $h$ satisfies (\ref{B:1}) in the following sense:
\begin{equation}\label{integral_form}
\int\limits_0^T \langle h_t(\cdot,t), \phi \rangle \; dt -
\iint\limits_{\mathcal{P}} {f(h) ( a_0 h_{xxx} + a_1 D''(h)
h_x)\phi_x\,dx dt } \ = 0
\end{equation}
for all $\phi \in C^1(Q_T)$ with $\phi(-a,\cdot) = \phi(a,\cdot)$;
\begin{align}
&  \label{ID1} h(\cdot,t) \to h(\cdot,0) = h_0
\mbox{  pointwise \& strongly in $L^2(\Omega)$ as $t \to 0$}, \\
& \label{BC1} h(-a,t)=h(a,t) \; \forall t \in [0,T] \; \mbox{and}
\; \tfrac{\partial^{i}
h}{\partial x^i}(-a,t) = \tfrac{\partial^{i}h}{\partial x^i}(a,t)  \\
& \notag \mbox{for} \; i = \overline{1,3} \; \mbox{at all points
of the lateral boundary where $\{h \neq 0 \}$.}
\end{align}
\end{definition}

Because the second term of (\ref{integral_form}) has an integral
over $\mathcal{P}$ rather than over $Q_T$, the generalized weak
solution is ``weaker'' than a standard weak solution. Also note that
the first term of (\ref{integral_form}) uses $h_t \in L^2(0,T;
(H^1(\Omega))' )$; this is different from the definition of weak
solution first introduced by Bernis and Friedman \cite{B8}; there,
the first term was the integral of $h \phi_t$ integrated over $Q_T$.
We do not require a test function $\phi$ to be zero at both ends
$t=0$ and $t = T$ that is crucial for our construction of a
continuation of the weak solution.


\section{Main results}\label{C}

Our main results are:
the short-time existence of a nonnegative generalized weak solution (Theorem \ref{C:Th1}),
the short-time existence of a nonnegative strong solution (Theorem \ref{C:Th2}),
finite-speed of propagation
and
finite-time blow-up
for some of these strong solutions (Theorems \ref{C:Th4} and \ref{C:Th3} respectively).

The short-time existence of generalized weak solutions relies on an integral quantity:
$\int G_0(h(x,t)) \; dx$.  This
``entropy'' was introduced by Bernis and Friedman \cite{B8}, where
\begin{equation}\label{C:G_0-def}
G_{0}(z): =
\begin{cases}
\tfrac{ z^{2-n }}{(2-n)(1-n)}  & \text{if } n \in (0, 1) \cup (2,\infty)  \\
\tfrac{ z^{2-n }}{(2-n)(1-n)} + z + b & \text{if } n \in( 1, 2 ) \\
z \ln z - z +1  & \text{ if } n = 1 \\
- \ln z  + z/e & \text{ if } n = 2.
\end{cases}
\Longrightarrow
\quad
G_0''(z) = \tfrac{1}{z^n}
\end{equation}
with $b$ chosen to ensure that $G_0(z) \geqslant 0$ for all $z \geqslant 0$.

\begin{theorem}[Existence]\label{C:Th1}
Let $n > 0,\ m \geqslant \tfrac{n}{2} $, $a_0 > 0$, and $a_1
\geqslant 0$ in equation (\ref{A.gen.thin-film}). Assume that the nonnegative initial data $h_0 \in
H^1(\Omega)$ has finite entropy
\begin{equation}\label{C:inval}
\int\limits_{\Omega} {G_0(h_0(x)) \,dx} < \infty,
\end{equation}
where $G_0$ is given by (\ref{C:G_0-def})
and either 1) $h_0(-a) = h_0(a) = 0$ or  2) $h_0(-a) = h_0(a) \neq
0$ and $ \tfrac{\partial^{i} h_0}{\partial x^i}(-a) =
\tfrac{\partial^{i}h_0}{\partial x^i}(a) \text{ holds for }i=1,2,3$.
Then for some time $T_{loc}>0$ there exists a nonnegative
generalized weak solution $h$
on $Q_{T_{loc}}$
in the sense of the definition
\ref{B:defweak}. Furthermore,
\begin{equation} \label{Linf_H2}
h \in L^2(0,T_{loc};H^2(\Omega)).
\end{equation}
If
\begin{equation} \label{D_is}
D_0(z) =
\begin{cases}
 \tfrac{z^{m -
n +2}}{(m - n + 1)(m- n + 2)} & \mbox{if} \; m-n \notin \{ -2, -1 \}\\
- \log(z) & \mbox{if} \; m-n = -2 \\
z \log(z) & \mbox{if} \; m-n = -1
\end{cases},
\quad \Longrightarrow \quad
D_0''(z) = z^{m-n},
\end{equation}
\begin{equation} \label{Energy}
\mathcal{E}_0(T) := \int\limits_{\Omega} { \{\tfrac{a_0}{2}
h_{x}^2(x,T) - a_1 D_0(h(x,T)) \}\,dx},
\qquad
B_1(T):= \tfrac{a_1^2}{a_0}\int \limits_0^T { \|h(.,t)\|^{2m
-n}_{L^{\infty}(\Omega)}\,dt},
\end{equation}
and
$$
B_{2}(T) := \tfrac{a_1^4}{2a_0^3(2m -n + 1)^2} \int\limits_0^T \|
h(\cdot,t) \|_\infty^{4m-n}\; dt + \tfrac{a_1^2}{2a_0(m -n +1)^2}
\int\limits_0^T \| h(\cdot,t) \|_\infty^{2m-n} \; dt,
$$
then the weak solution satisfies
\begin{equation}\label{C:d2'}
%
\mathcal{E}_0(T_{loc}) + \iint\limits_{\{h >0 \}} {h^n (a_0 h_{xxx} + a_1 h^{m - n}h_x)^2 \,dx dt}
\leqslant \mathcal{E}_0(0),
\end{equation}
\begin{equation}\label{C:a2-new}
\int\limits_{\Omega} {h_{x}^2(x,T_{loc}) \,dx} \leqslant
e^{B_1(T_{loc} )}\int \limits_{\Omega} {h_{0x}^2(x) \,dx},
\end{equation}
and
\begin{equation}\label{C:a2-new2}
\int\limits_{\Omega} {\{h_{x}^2(x,T_{loc}) + G_0(h(x,T_{loc})\}
\,dx} \leqslant e^{B_2(T_{loc} )}\int \limits_{\Omega} {\{
h_{0x}^2(x) + G_0(h_0(x)) \,dx}.
\end{equation}
The time of existence, $T_{loc}$, is determined by $a_0$, $a_1$,
$| \Omega |$, $\int h_0$, $\| h_{0x} \|_2$, and $\int G_0(h_0)$.
\end{theorem}

There is nothing special about the time $T_{loc}$ in the bounds
(\ref{C:d2'}), (\ref{C:a2-new}), and (\ref{C:a2-new2}); given a countable collection
of times in $[0,T_{loc}]$, one can construct a weak solution for
which these bounds will hold at those times.  Also, we note that
the analogue of Theorem 4.2 in \cite{B8} also holds: there exists
a nonnegative weak solution with the integral formulation
%
\begin{equation} \label{alt_int}
\int\limits_0^T \langle h_t(\cdot,t), \phi \rangle \; dt
+ a_0 \iint\limits_{Q_T} (n h^{n-1} h_x h_{xx} \phi_x + h^n h_{xx} \phi_{xx}) \; dx dt
 - a_1 \iint\limits_{Q_T} { h^m h_x \phi_x \; dx
dt} = 0.
\end{equation}

We note that the existence theory for the long-wave stable case, $a_1 < 0$, has already been
considered by \cite{BP_PM,DPGS_2001}.

Theorem \ref{C:Th2} states that if the initial data also has finite ``$\alpha$-entropy'' then a solution
can be constructed which satisfies Theorem \ref{C:Th1} and has additional regularity.  This solution
is called a ``strong'' solution; see \cite{BertPugh1996,B2,BP_PM}.  The $\alpha$-entropy, defined below,
was discovered by Leo Kadanoff \cite{BertozziBrenner}.  Let
\begin{equation}\label{C:Galpha}
G_{0}^{(\alpha)} (z) =
\begin{cases}
\tfrac{z^{2-n+\alpha}}{(2-n+\alpha)(1-n+\alpha)} & \mbox{if } \alpha \in (-\infty,n-2) \cup (n-1,\infty) \\
\tfrac{z^{2-n+\alpha}}{(2-n+\alpha)(1-n+\alpha)} + z + c & \mbox{if }  \alpha \in (n-2,n-1) \\
z \ln z - z +1  & \mbox{if }   \alpha = n-1 \\
- \ln z  + z/e & \mbox{if }  \alpha = n-2,
\end{cases}
G_0^{(\alpha)''}(z) = \tfrac{z^\alpha}{z^n}
%
\end{equation}
where $c$ is chosen to ensure that $G_0^{(\alpha)}(z) \geqslant 0$ for all $z \geqslant 0$.

\begin{theorem}[Regularity]\label{C:Th2}
Assume the initial data $h_0$ satisfies the hypotheses of Theorem \ref{C:Th1} and also
has finite $\alpha$-entropy for some $\alpha \in (-1/2,1)$, $\alpha \neq 0$,
\begin{equation} \label{finite_alpha_ent}
\int\limits_{\Omega} {G_0^{(\alpha)}(h_0(x)) \,dx} <\infty.
\end{equation}
Then for some time time $T_{loc}^{(\alpha)} \in (0,T_{loc}]$
there exists a nonnegative generalized weak solution
on $Q_{T_{loc}^{(\alpha)}}$ that
satisfies Theorem \ref{C:Th1} and has the additional regularity
\begin{equation}\label{Linf_H2-2}
h^{\tfrac{\alpha + 2}{2}} \in L^{2}(0, T_{loc}^{(\alpha)};
H^2(\Omega)) \text{ and } h^{\tfrac{\alpha + 2}{4}} \in L^{2}(0,
T_{loc}^{(\alpha)}; W^1_4(\Omega)).
\end{equation}
\end{theorem}

\begin{theorem}[Finite Speed of Propagation]\label{C:Th4}
Consider the range of exponents $0 < n \leqslant \tfrac{1}{2}$ and
$n/2 <  m < 6 - n $ or $\tfrac{1}{2} < n < 3$ and
$m \geqslant n/2 $ (see Figure \ref{thm1_region}).
Assume the initial data satisfies the hypotheses of Theorem \ref{C:Th2}
and also has compact support
in $(-a,a)$:
 $\mbox{supp}(h_0) \subseteq [-r_0,r_0] \subset
(-a,a)$.
Then there exists
a time $0 < T_{speed} \leqslant T_{loc}^{(\alpha)}$ and
a nondecreasing function $\Gamma(t) \in C([0,T_{speed}])$
such that the
strong solution from
Theorem~\ref{C:Th2}, has finite speed propagation, i.e.
$$
\mbox{supp}(h(\cdot,t)) \subseteq [-r_0-\Gamma(t),r_0+\Gamma(t)] \subset
(-a,a)
$$
for all $t \in [0,T_{speed}]$.
Furthermore, if $0<n<2$ there exists a constant $C$ which depends
on $n$, $m$, $\alpha$, and $\int h_0$ such that $\Gamma(t) \leqslant C t^{1/(n+4)}$.
\end{theorem}


\begin{theorem}[Finite Time Blow Up]\label{C:Th3}
Consider the range of exponents
$0 < n \leqslant
\tfrac{1}{2}$ and $4 - n \leqslant m  < 6 - n$
or $\tfrac{1}{2} < n \leqslant 1$ and $m \geqslant 4 - n$
or $1 < n < 2$ and $m \geqslant n+2$ (see Figure \ref{thm1_region}).
Assume the nonnegative initial data $h_0$ has compact support,
negative energy (\ref{Energy})
$$
\mathcal{E}_0(0) = \int\limits_{-\infty}^{\infty} { \{\tfrac{a_0}{2}
h_{0x}^2 - a_1 D_0(h_0) \}\,dx} < 0,
$$
and satisfies the hypotheses of Theorem \ref{C:Th2} with
$\Omega = \R^1$.
Then there exists a finite time $T^* < \infty$ and a nonnegative,
compactly supported, strong solution $h$ on $\R^1 \times [0,T^*)$
such that
\begin{equation}\label{C:blow-up}
\mathop {\limsup} \limits_{t \to T^*} \|h(.,t)\|_{H^{1}(\R^1)} =
\mathop {\limsup} \limits_{t \to T^*}
\|h(.,t)\|_{L^{\infty}(\R^1)} = +\infty.
\end{equation}
The solution satisfies the bounds of Theorems \ref{C:Th1}
and \ref{C:Th2} with $\Omega=\R^1$.
\end{theorem}

\section{Proof of Existence of Generalized Solutions} \label{D}

Our proof of existence of generalized weak solutions defined in
the section \ref{B} follows the main concept of the proof from
\cite{B8}.

\subsection{Regularized Problem}\label{RegularizedProblem}

Given $\delta, \varepsilon > 0$, a regularized parabolic problem,
similar to that of Bernis and Friedman \cite{B8}, is considered:
\begin{numcases}
{(\textup{P}_{\delta,\varepsilon})}
 h_t  + \left({ f_{\delta\varepsilon}(h) (a_0 h_{xxx} + a_1 D''_{\varepsilon}(h)
 h_x)}\right)_x = 0, \qquad \hfill \label{D:1r'}\\
\tfrac{\partial^{i} h}{\partial x^i}(-a,t) = \tfrac{\partial^{i}
h}{\partial x^i}(a,t) \text{ for }
t > 0,\,i= \overline{0,3} , \hfill \label{D:2r'}\\
\qquad  \qquad h(x,0) = h_{0,\varepsilon}(x) \hfill
\label{D:3r'}
\end{numcases}
where
\begin{equation}\label{D:reg1}
f_{\delta \varepsilon} (z): = f_{\varepsilon} (z) + \delta =
\tfrac{|z|^{4 + n}}{|z|^4 + \varepsilon |z|^n}+ \delta,\
D''_{\varepsilon}(z): = \tfrac{|z|^{m - n}}{1 + \varepsilon |z|^{m
- n}}, \eps > 0, \delta > 0.
\end{equation}
We note that $f_{\delta \eps}  \in C^{1+\gamma}(\R^1)$ and
$f_{\delta \eps} D''_{\varepsilon} \in C^{1+\gamma}(\R^1)$
for some $\gamma \in (0,1)$.
The $\delta>0$ in (\ref{D:reg1}) makes the problem (\ref{D:1r'})
regular (i.e. uniformly parabolic). The parameter $\varepsilon$ is
an approximating parameter which has the effect of increasing the
degeneracy from $f(h) \sim |h|^n$ to $f_{\varepsilon}(h) \sim
h^4$.
For $\eps > 0$,
the nonnegative initial data, $h_0$, is approximated via
\begin{equation}\label{D:inreg}
\begin{gathered}
h_{0} + \varepsilon^\theta \leq
h_{0,\varepsilon}  \in
C^{4+\gamma}(\overline{\Omega}) \text{ for some } 0 < \theta <
\tfrac{2}{5},\\
\tfrac{\partial^{i} h_{0,\varepsilon}}{\partial x^i}(-a) =
\tfrac{\partial^{i}h_{0,\varepsilon}}{\partial x^i}(a)
\text{ for } i=\overline{0,3}, \\
h_{0,\varepsilon} \to h_{0}  \text{ strongly in }
H^1(\Omega) \text{ as } \varepsilon \to 0.
\end{gathered}
\end{equation}
The effect of $\varepsilon>0$ in (\ref{D:inreg}) is to both ``lift'' the initial data, making it positive,
and to smooth the initial data from $H^1(\Omega)$ to $C^{4+\gamma}(\Omega)$.

By E\u{i}delman \cite[Theorem 6.3, p.302]{Ed}, the regularized
problem has a unique classical solution $h_{\delta \varepsilon}
\in C_{x,t}^{4+\gamma,1+\gamma/4}( \Omega \times [0, \tau_{\delta
\varepsilon}])$ for some time $\tau_{\delta \varepsilon}
> 0$.
For any fixed value of  $\delta$ and $\varepsilon$, by
E\u{i}delman \cite [Theorem 9.3, p.316]{Ed} if one can prove a
uniform in time an a priori bound $|h_{\delta \varepsilon}(x,t)|
\leqslant A_{\delta \varepsilon}<\infty$ for some longer time interval
$[0,T_{loc,\delta \varepsilon}] \quad (T_{loc,\delta \varepsilon}
> \tau_{\delta \varepsilon}$) and for all $x \in \Omega$ then
Schauder-type interior estimates \cite [Corollary 2, p.213] {Ed}
imply that the solution $h_{\delta \varepsilon}$ can be continued
in time to be in $C_{x,t}^{4+\gamma,1+\gamma/4}( \Omega \times
[0,T_{loc,\delta \varepsilon}])$.

Although the solution $h_{\delta \varepsilon}$ is initially
positive, there is no guarantee that it will remain nonnegative. The
goal is to take $\delta \to 0$, $\varepsilon \to 0$ in such a way
that a) $T_{loc,\delta \varepsilon} \to T_{loc} > 0$, b) the
solutions $h_{\delta \varepsilon}$ converge to a (nonnegative)
limit, $h$, which is a generalized weak solution, and c) $h$
inherits certain a priori bounds.  This is done by proving various a
priori estimates for $h_{\delta \varepsilon}$ that are uniform in
$\delta$ and $\varepsilon$ and hold on a time interval $[0,T_{loc}]$
that is independent of $\delta$ and $\varepsilon$. As a result, $\{
h_{\delta \varepsilon} \}$ will be a uniformly bounded and
equicontinuous (in the $C_{x,t}^{1/2,1/8}$ norm) family of functions
in $\bar{\Omega} \times [0, T_{loc}]$. Taking $\delta \to 0$ will
result in a family of functions $\{ h_{\varepsilon} \}$ that are
classical, positive, unique solutions to the regularized problem
with $\delta = 0$. Taking a subsequence of $\varepsilon$ going
to zero will then result in
the desired generalized weak solution $h$. This last  step is where
the possibility of nonunique weak solutions arise; see \cite{B2}
for simple examples of how such a construction applied to $h_t = -
(|h|^n h_{xxx})_x$ can result in there being two different solutions
arising
from the same initial data.

\subsection{A priori estimates}

We start by proving  a priori estimates for classical
solutions of the regularized problem (\ref{D:1r'})--(\ref{D:inreg}).
Appendix \ref{A_priori_proofs} contains the proofs of the lemmas
in this section.

We introduce function $G_{\delta \eps}$ chosen such that
\begin{equation}\label{D:reg2}
G''_{\delta \varepsilon} (z) = \tfrac{1}{f_{\delta \varepsilon}
(z)} \quad \mbox{and} \quad G_{\delta \varepsilon}(z) \geqslant 0 \quad \forall z \in \R^1.
\end{equation}
This is analogous to the ``entropy'' function (\ref{C:G_0-def}) introduced by
Bernis and Friedman \cite{B8}.

\begin{lemma}\label{MainAE}
Let $ n > 0$ and $m \geqslant n/2 $. There exist $\delta_0 > 0$,
$\varepsilon_0 > 0$, and time $T_{loc}>0$ such that if $\delta \in
[0,\delta_0)$, $\varepsilon \in (0,\varepsilon_0)$ and
$h_{\delta \varepsilon}$ is a classical solution of the problem
(\ref{D:1r'})--(\ref{D:inreg}) with initial data $h_{0,\delta
\varepsilon}$ that
is built from a nonnegative function $h_0$ satisfying
the hypotheses of Theorem \ref{C:Th1} then for any $T \in
[0, T_{loc}]$ the following inequalities hold true:
\begin{align}\label{D:a13''}
& \int\limits_{\Omega} { \{h_{\delta \varepsilon, x}^2(x,T) +
\tfrac{a_1^2}{2a_0^2}(1 + \delta) G_{\delta \varepsilon}(h_{\delta \varepsilon}(x,T))\} \,dx} \\
& \hspace{1in}+ \notag a_0 \iint\limits_{Q_T} { f_{\delta
\varepsilon}(h_{\delta \varepsilon}) h^2_{\delta \varepsilon, xxx}
\,dx dt} \leqslant K_1 < \infty,
\end{align}
\begin{equation} \label{BF_entropy}
\int\limits_\Omega G_{\delta \varepsilon}(h_{\delta
\varepsilon}(x,T)) \; dx + \tfrac{a_0}{2} \iint\limits_{Q_T}
h_{\delta \varepsilon, xx}^2 \; dx dt \leqslant K_2 < \infty.
\end{equation}
The energy $\mathcal{E}_{\delta \varepsilon} (t)$ (see
(\ref{Energy})) satisfies:
\begin{equation}\label{D:d2}
\mathcal{E}_{\delta \varepsilon}(T) + \iint\limits_{Q_T}
{f_{\delta \varepsilon}(h_{\delta \varepsilon}) (a_0 h_{\delta
\varepsilon,xxx} + a_1 D''_{\varepsilon}(h_{\delta
\varepsilon})h_{\delta \varepsilon, x})^2} \; dx dt =
\mathcal{E}_{\delta \varepsilon}(0).
\end{equation}
The time $T_{loc}$ and the constants $K_1$, and $K_2$ are
independent of $\delta$ and $\varepsilon$.
\end{lemma}

The existence of $\delta_0$, $\varepsilon_0$, $T_{loc}$, $K_1$,
$K_2$, and $C_0$ is constructive; how to find them and what
quantities determine them is shown in Section
\ref{A_priori_proofs}.

Lemma \ref{MainAE} yields uniform-in-$\delta$-and-$\varepsilon$
bounds for $\int h_{\delta \varepsilon,x}^2$, $\int G_{\delta
\varepsilon}(h_{\delta \varepsilon})$,
$\iint f_{\delta \varepsilon}(h_{\delta \varepsilon}) h_{\delta
\varepsilon,xxx}^2$, and
$\iint h_{\delta
\varepsilon,xx}^2$ which will be key in constructing a nonnegative
generalized weak solution.  However, these bounds are found in a
different manner than in earlier work for the equation $h_t = -
(|h|^n h_{xxx})_x$, for example.  Although the inequality
(\ref{BF_entropy}) is unchanged, the inequality (\ref{D:a13''})
has an extra term involving $G_{\delta \varepsilon}$.  In the
proof, this term was introduced to control additional,
lower--order terms. This idea of a ``blended'' $\| h_x
\|_2$--entropy bound was first introduced by Shishkov and Taranets
especially for long-wave stable thin film equations with
convection \cite{T4}.


\begin{lemma}\label{lemLL}
Assume $\varepsilon_0$ and $T_{loc}$ are from Lemma \ref{MainAE},
$\delta = 0$, and $\varepsilon \in (0,\varepsilon_0)$.  If
$h_{\varepsilon}$ is a positive, classical solution of the problem
(\ref{D:1r'})--(\ref{D:inreg}) with initial data
$h_{0,\varepsilon}$ satisfying Lemma \ref{MainAE} then
\begin{equation}\label{D:a13-new}
\int\limits_{\Omega} {h_{\varepsilon,x}^2(x,T) \,dx} \leqslant
e^{B_{1,\varepsilon}(T)} \int\limits_\Omega
{h_{0\varepsilon,x}^2\,dx},
\end{equation}
\begin{equation}\label{D:a13-new2}
\int\limits_{\Omega} {\{ h_{\varepsilon,x}^2(x,T) +
G_{\varepsilon}(h_{\varepsilon}(x,T)) \}\,dx} \leqslant
e^{B_{2,\varepsilon}(T)} \int\limits_\Omega {\{
h_{0\varepsilon,x}^2 + G_{\varepsilon} (h_{0,\varepsilon}) \}\,dx}
\end{equation}
hold true for all $T \in [0,T_{loc}]$. Here
$$
B_{1,\varepsilon}(T) := \tfrac{a_1^2}{a_0} \int\limits_0^T \|
h_\varepsilon(\cdot,t) \|_\infty^{2m-n} \; dt,
$$
$$
B_{2,\varepsilon}(T) := \tfrac{a_1^4}{2a_0^3(2m -n + 1)^2}
\int\limits_0^T \| h_\varepsilon(\cdot,t) \|_\infty^{4m-n}\; dt +
\tfrac{a_1^2}{2a_0(m -n +1)^2} \int\limits_0^T \|
h_\varepsilon(\cdot,t) \|_\infty^{2m-n} \; dt.
$$
\end{lemma}

The final a priori bound uses the following functions, parametrized by
$\alpha$, chosen such that $G_\eps^{(\alpha)} \geqslant 0$ and
$G_\eps^{(\alpha)\prime \prime}(z) = z^\alpha/f_\eps(z)$:
\begin{equation}\label{E:reg4}
G_{\varepsilon}^{(\alpha)} (z) =
\begin{cases}
\tfrac{z^{2-n+\alpha}}{(2-n+\alpha)(1-n+\alpha)}
+ \eps \tfrac{z^{\alpha-2}}{(\alpha-3)(\alpha-2)}
& \mbox{if } \alpha \in (-\infty,n-2) \cup (n-1,\infty) \\
\tfrac{z^{2-n+\alpha}}{(2-n+\alpha)(1-n+\alpha)}
+ \eps \tfrac{z^{\alpha-2}}{(\alpha-3)(\alpha-2)}
+ z + c & \mbox{if }  \alpha \in (n-2,n-1) \\
z \ln z - z +1
+ \eps \tfrac{z^{\alpha-2}}{(\alpha-3)(\alpha-2)}
& \mbox{if }   \alpha = n-1 \\
- \ln z  + z/e
+ \eps \tfrac{z^{\alpha-2}}{(\alpha-3)(\alpha-2)}
& \mbox{if }  \alpha = n-2,
\end{cases}
%
%
\end{equation}

\begin{lemma}\label{MainAE2}
Assume $\varepsilon_0$ and $T_{loc}$ are from Lemma \ref{MainAE},
$\delta = 0$, and $\varepsilon \in (0,\varepsilon_0)$.  Assume
$h_\varepsilon$ is a positive, classical solution of the problem
(\ref{D:1r'})--(\ref{D:inreg}) with initial data
$h_{0,\varepsilon}$ that
is built from a nonnegative function $h_0$ satisfying
the hypotheses of Theorem \ref{C:Th2}
then there exists $K_3$, $\varepsilon_0^{(\alpha)}$ and
$T_{loc}^{(\alpha)}$ with $0 < \varepsilon_0^{(\alpha)} \leqslant
\varepsilon_0$ and $0 < T_{loc}^{(\alpha)} \leqslant T_{loc}$ such
that
\begin{equation}
\label{E:b12''} \int\limits_{\Omega} {
\{h_{\varepsilon,x}^2(x,T) +
G_{\varepsilon}^{(\alpha)}(h_\varepsilon(x,T))\}
\,dx} + \iint\limits_{Q_T} \left[ \beta
h_\varepsilon^\alpha h_{\varepsilon,xx}^2 + \gamma
h_\varepsilon^{\alpha-2} h_{\varepsilon,x}^4 \right]\;dx\,dt
\leqslant K_3 < \infty
\end{equation}
holds for all $T \in [0 , T_{loc}^{(\alpha)} ]$.
$K_3$ is determined by $\alpha$, $\varepsilon_0$, $a_0$,
$a_1$, $\Omega$ and $h_0$. Here,
$$
\beta =
\begin{cases}
a_0 & \mbox{if } 0 < \alpha < 1, \\
a_0 \tfrac{1+2\alpha}{4(1-\alpha)} & \mbox{if } -1/2 < \alpha < 0
\end{cases}
$$
and
$$
\gamma =
\begin{cases}
a_0 \tfrac{\alpha(1-\alpha)}{6}& \mbox{if } 0 < \alpha < 1, \\
a_0 \tfrac{(1+2\alpha)(1-\alpha)}{36} & \mbox{if } -1/2 < \alpha < 0.
\end{cases}
$$
Furthermore,
\begin{equation}\label{E:b13}
h_\varepsilon^{\tfrac{\alpha+2}{2}} \in B_{R_1}(0) \subset
L^{2}(0, T_{loc};
H^2(\Omega)) \quad \mbox{and} \quad
h_\varepsilon^{\tfrac{\alpha+2}{4}}  \in B_{R_2}(0) \subset
\in L^{2}(0, T_{loc};
W^1_4(\Omega))
\end{equation}
where the radii $R_1$ and $R_2$ are independent of $\eps$.
\end{lemma}

The $\alpha$--entropy, $\int G_0^{(\alpha)}(h) \; dx$, was first
introduced for $\alpha = - 1/2$ in \cite{BertozziBrenner} and an a
priori bound like that of Lemma \ref{MainAE2} and regularity
results like those of Theorem \ref{C:Th2} were found
simultaneously and independently in \cite{B2} and
\cite{BertPugh1996}.

\subsection{Proof of existence and regularity of solutions}

Bound (\ref{D:a13''}) yields uniform $L^\infty$ control for
classical solutions $h_{\delta \varepsilon}$, allowing the time of
existence $T_{loc,\delta \varepsilon}$ to be taken as $T_{loc}$
for all $\delta \in (0,\delta_0)$ and $\varepsilon \in
(0,\varepsilon_0)$. The existence theory starts by constructing a
classical solution $h_{\delta \varepsilon}$ on $[0,T_{loc}]$ that
satisfy the hypotheses of Lemma \ref{MainAE} if $\delta \in
(0,\delta_0)$ and $\varepsilon \in (0,\varepsilon_0)$.
The a priori bounds of Lemma \ref{MainAE} then allow
one to take the
regularizing parameter, $\delta$, to zero and prove
that there is a limit $h_\varepsilon$ and that $h_\varepsilon$ is
a generalized weak solution. One then proves additional regularity
for $h_\varepsilon$; specifically that it is strictly positive,
classical, and unique. It then follows that the a priori bounds
given by Lemmas \ref{MainAE}, \ref{lemLL}, and \ref{MainAE2} apply to
$h_\varepsilon$. This allows one to take the approximating
parameter, $\varepsilon$, to zero and construct the desired
generalized weak solution of Theorems \ref{C:Th1} and \ref{C:Th2}:

\begin{lemma}\label{preC:Th1}
Assume that the initial data $h_{0,\varepsilon}$ satisfies
(\ref{D:inreg}) and is built from a nonnegative function $h_0$
that satisfies the hypotheses of Theorem \ref{C:Th1}. Fix $\delta
= 0$ and $\varepsilon \in (0,\varepsilon_0)$ where $\varepsilon_0$
is from Lemma \ref{MainAE}. Then there exists a unique, positive,
classical solution $h_\varepsilon$ on $[0,T_{loc}]$ of problem
($\mbox{P}_{0, \varepsilon}$), see (\ref{D:1r'})--(\ref{D:inreg}),
with initial data $h_{0,\varepsilon}$ where $T_{loc}$ is the time
from Lemma \ref{MainAE}.
\end{lemma}

The proof uses a number of arguments like those presented by
Bernis \& Friedman \cite{B8} and we refer to that article as much
as possible.

\begin{proof}
Fix $\varepsilon \in (0,\varepsilon_0)$ and assume $\delta \in
(0,\delta_0)$. Because $G_{\delta \varepsilon}(z) \geqslant 0$,
the bound (\ref{D:a13''}) yields a
uniform-in-$\delta$-and-$\varepsilon$ upper bound on $|h_{\delta
\varepsilon}(x,T)|$ for $(x,T) \in \overline{\Omega} \times
[0,T_{loc}]$. As discussed in Subsection \ref{RegularizedProblem},
this allows the classical solution $h_{\delta \varepsilon}$ to be
extended from $[0,\tau_{\delta \varepsilon}]$ to $[0,T_{loc}]$.

By Section 2 of \cite{B8}, the a priori bound (\ref{D:a13''}) on
$\| h_x(\cdot,T_{loc}) \|_2$ implies that $h_{\delta \varepsilon} \in
C^{1/2,1/8}_{x,t}(\overline{Q_{T_{loc}}})$ and that $\{ h_{\delta
\varepsilon} \}$ is a uniformly bounded, equicontinuous family in
$\overline{Q_{T_{loc}}}$. By the Arzela-Ascoli theorem, there is a
subsequence $\{ \delta_k \}$, so that $h_{\delta_k \varepsilon}$
converges uniformly to a limit $h_\varepsilon \in
C^{1/2,1/8}_{x,t}(\overline{Q_{T_{loc}}})$.

We now argue that $h_\varepsilon$ is a generalized weak solution,
using methods similar to those of \cite[Theorem 3.1]{B8}.

By construction, $h_\varepsilon$ is in
$C^{1/2,1/8}_{x,t}(\overline{Q_{T_{loc}}})$, satisfying the first
part of (\ref{weak1}).  The strong convergence
$h_\varepsilon(\cdot , t) \to h_\varepsilon(\cdot,0)$ in
$L^2(\Omega)$ follows immediately. The uniform convergence of
$h_{\delta_k \varepsilon}$ to $h_\varepsilon$ implies the
pointwise convergence $h(\cdot,t) \to h(\cdot,0) = h_0$, and so
$h_\varepsilon$ satisfies (\ref{ID1}).

Because $h_{\delta \varepsilon}$ is a classical solution,
\begin{equation}\label{integral_form_classical}
\iint \limits_{Q_T} { h_{\delta \varepsilon,t} \phi \; dx dt}  -
\iint\limits_{Q_{T_{loc}}} { f_{\delta \varepsilon}(h_{\delta
\varepsilon}) ( a_0 h_{\delta \varepsilon, xxx} + a_1
D''_{\varepsilon}(h_{\delta \varepsilon})h_{\delta \varepsilon,
x})\phi_x \; dx dt }  = 0.
\end{equation}

The bound (\ref{D:a13''}) yields a uniform bound on
$$
\delta \iint\limits_{Q_{T_{loc}}} h_{\delta \varepsilon,xxx}^2 \;
dx dt
$$
for $\delta \in (0, \delta_0)$. It follows that
$$
\delta_k \iint\limits_{Q_{T_{loc}}}
 ( a_0 h_{\delta \varepsilon, xxx} + a_1 D''_{\varepsilon}(h_{\delta \varepsilon}) h_{\delta \varepsilon, x} )\phi_x\,dx dt   \to 0
 \qquad \mbox{as $\delta_k \to 0$.}
$$
Introducing the notation
\begin{equation} \label{flux_like}
H_{\delta \varepsilon} := f_{\delta \varepsilon}(h_{\delta
\varepsilon}) ( a_0 h_{\delta \varepsilon, xxx} + a_1
D''_{\varepsilon}(h_{\delta \varepsilon}) h_{\delta \varepsilon,
x})
\end{equation}
the integral formulation (\ref{integral_form_classical}) can be
written as
\begin{equation} \label{abstract_integral_form}
\iint \limits_{Q_T} { h_{\delta \varepsilon,t} \phi dx dt} =
\iint\limits_{Q_{T_{loc}}} H_{\delta \varepsilon}(x,t) \phi_x(x,t)
\; dx dt.
\end{equation}
By the $L^\infty$ control of $h_{\delta \varepsilon}$ and the
energy bound (\ref{D:d2}), $H_{\delta \varepsilon}$ is uniformly
bounded in $L^2(Q_{T_{loc}})$.  Taking a further subsequence of
$\{ \delta_k \}$ yields $H_{\delta_k \varepsilon}$ converging
weakly to a function $H_\varepsilon$ in $L^2(Q_{T_{loc}})$. The
regularity theory for uniformly parabolic equations implies that
$h_{\delta \varepsilon,t}$, $h_{\delta \varepsilon, x}$,
$h_{\delta \varepsilon, xx}$, $h_{\delta \varepsilon, xxx}$, and
$h_{\delta \varepsilon,xxxx}$ converge uniformly to
$h_{\varepsilon,t}$, \dots, $h_{\varepsilon, xxxx}$ on any compact
subset of $\{ h_\varepsilon > 0 \}$, implying (\ref{BC1}) and the
first part of (\ref{weak2}). Note that because the initial data
$h_{0,\delta \varepsilon}$ is in $C^4$ the regularity extends all the way
to $t=0$ which is excluded in the definition of $\mathcal{P}$ in
(\ref{weak2}).

The energy $\mathcal{E}_{\delta \varepsilon}(T_{loc})$ is not
necessarily positive. However, the a priori bound (\ref{D:a13''}),
combined with the $L^\infty$ control on $h_{\delta \varepsilon}$,
ensures that $\mathcal{E}_{\delta \varepsilon}(T_{loc})$ has a
uniform lower bound. As a result, the bound (\ref{D:d2}) yields a
uniform bound on
$$
\iint\limits_{Q_{T_{loc}}} f_{\delta \varepsilon} \left( h_{\delta
\varepsilon}) ( a_0 h_{\delta \varepsilon, xxx} + a_1
D''_{\varepsilon}(h_{\delta \varepsilon}) h_{\delta \varepsilon,
x} \right)^2 \; dx dt.
$$
Using this, one can argue that for any $\sigma >0$
$$
\iint\limits_{ \{h_\varepsilon < \sigma \} } \left|  H_{\delta_k
\varepsilon} \phi_x \right| \; dx dt \leqslant C \sigma^{n/2}
$$
for some $C$ independent of $\delta$, $\varepsilon$, and $\sigma$.
Taking $\delta_k \to 0$ and using that $\sigma$ is arbitrary, we
conclude
$$
H_{\delta_k \varepsilon} \to H_\varepsilon = f_{\varepsilon}
\left( h_{\varepsilon}) ( a_0 h_{\varepsilon, xxx} + a_1
D''_{\varepsilon}(h_{\delta \varepsilon})h_{\varepsilon, x}
\right) \,
 \chi_{ \{ h_\varepsilon > 0 \} }.
$$
As a result, taking $\delta_k \to 0$ in
(\ref{abstract_integral_form}) implies $h_\varepsilon$ satisfies
(\ref{integral_form}) and the second part of (\ref{weak2}).

The bound (\ref{D:a13''}) yields
a uniform bound on $\int h_{\delta
\varepsilon,x}^2(x,T) \; dx$ for every $T \in [0,T_{loc}]$.  As a
result, $\{ h_{\delta_k \varepsilon} \}$ is uniformly bounded in
$$
 L^\infty(0,T_{loc}; H^1(\Omega)).
$$
Therefore, another refinement of the sequence $\{ \delta_k \}$
yields $\{ h_{\delta_k \varepsilon} \}$ weakly convergent in this
space. As a result, $h_\varepsilon \in  L^\infty(0,T_{loc};
H^1(\Omega))$ and the second part of  (\ref{weak1}) holds.

Having proven then $h_\varepsilon$ is a generalized weak solution,
we now prove that $h_\varepsilon$ is a strictly positive,
classical, unique solution.  This uses the entropy $\int G_{\delta
\varepsilon}(h_{\delta \varepsilon})$ and the a priori bound
(\ref{BF_entropy}). This bound  is, up to the coefficient $a_0$,
identical to the a priori bound (4.17) in \cite{B8}. By
construction, the initial data $h_{0,\varepsilon}$ is positive
(see (\ref{D:inreg})), hence $\int
G_{\varepsilon}(h_{0,\varepsilon}) \; dx < \infty$. Also, by
construction $f_\varepsilon(z) \sim z^4$ for $z \ll 1$. This
implies that the generalized weak solution $h_\varepsilon$ is
strictly positive \cite[Theorem 4.1]{B8}.  Because the initial
data $h_{0,\varepsilon}$ is in $C^4(\bar{\Omega})$, it follows
that $h_\varepsilon$ is a classical solution in
$C^{4,1}_{x,t}(\overline{Q_{T_{loc}}})$. This implies that
$h_{\varepsilon}(\cdot,t) \to h_{\varepsilon}(\cdot,0)$
strongly\footnote{ Unlike the definition of weak solution given in
\cite{B8}, Definition \ref{B:defweak} does not include that the
solution converges to the initial data strongly in $H^1(\Omega)$.
} in $C^1(\bar{\Omega})$.  The proof of  Theorem 4.1 in \cite{B8}
then implies that $h_\varepsilon$ is unique.
\end{proof}

\begin{proof}[Proof of Theorem \ref{C:Th1}]
As in the proof of Lemma \ref{preC:Th1}, following \cite{B8},
there is a subsequence $\{ \varepsilon_k \}$ such that
$h_{\varepsilon_k}$ converges uniformly to a function $h \in
C^{1/2,1/8}_{x,t}$ which is a generalized weak solution in the
sense of Definition \ref{B:defweak} with $f(h) = |h|^n$.

The initial data is assumed to have finite entropy: $\int
G_0(h_0) < \infty$ where $G_0$ is given by (\ref{C:G_0-def}).
This, combined with $f(h) = |h|^n$, implies
that the generalized weak solution $h$ is nonnegative and, if
$n\in[2,4)$ the set
of points $\{ h = 0 \}$ in $Q_{T_{loc}}$ has zero measure and
$h$ is positive, smooth, and unique
if $n\geqslant 4$
\cite[Theorem 4.1]{B8}.

To prove (\ref{C:d2'}), start by taking $T=T_{loc}$ in the a
priori bound (\ref{D:d2}). As $\varepsilon_k \to 0$, the
right-hand side of (\ref{D:d2}) is unchanged. Now, consider the
$\varepsilon_k \to 0$ limit of
$$
\mathcal{E}_{\varepsilon_k}(T_{loc}) =
 \int\limits_{\Omega} {(\tfrac{a_0}{2}
h_{\varepsilon_k,x}^2(x,T_{loc}) - a_1
D_{\varepsilon_k}(h_{\varepsilon_k}(x,T_{loc})) ) \,dx},
$$
where $D_\eps$ is defined in (\ref{D:reg1}).
By the uniform convergence of $h_{\varepsilon_k}$ to $h$, the
second term in the energy converges strongly as
$\varepsilon_k \to 0$. Hence the bound (\ref{D:d2}) yields a uniform
bound on $\{ \int_\Omega h_{\varepsilon_k,x}^2(x,T_{loc}) \; dx
\}$. Taking a further refinement of $\{ \varepsilon_k \}$, yields
$h_{\varepsilon_k,x}(\cdot,T_{loc})$ converging weakly in
$L^2(\Omega)$. In a Hilbert space, the norm of the weak limit is
less than or equal to the $\liminf$ of the norms of the functions
in the sequence, hence $ \mathcal{E}_0( T_{loc} ) \leqslant
\liminf_{\varepsilon_k \to 0}
\mathcal{E}_{\varepsilon_k}(T_{loc}). $ A uniform bound on $\iint
f_\varepsilon(h_\varepsilon) \left(a_0 h_{\varepsilon,xxx} + \dots
\right)^2 \; dx$ also follows from (\ref{D:d2}).  Hence
$\sqrt{f_{\varepsilon_k}(h_{\varepsilon_k})} \left(a_0
h_{\varepsilon_k,xxx} + \dots \right)$ converges weakly in
$L^2(Q_{T_{loc}})$, after taking a further subsequence.
We write the weak limit as two integrals: one over $\{ h = 0 \}$
and one over $\{ h > 0 \}$.
We can determine the weak limit on  $\{ h > 0 \}$:
as in the
proof of Lemma \ref{preC:Th1}, the regularity theory for uniformly
parabolic equations allows one to argue that the weak limit is $
h^{n/2} \left(a_0 h_{xxx} + \dots \right)$ on $\{ h > 0 \}$.
Using
that 1) the norm of the weak limit is less than or equal to the
$\liminf$ of the norms of the functions in the sequence and that
2) the $\liminf$ of a sum is greater than or equal to the sum of
the $\liminf$s and dropping the nonnegative term
arising from the integral over $\{ h = 0 \}$
results in the desired bound (\ref{C:d2'}).

It follows from (\ref{BF_entropy}) that $h_{\varepsilon_k,xx}$
converges weakly to some $v$ in $L^2(Q_{T_{loc}})$, combining with
strong convergence in $L^2(0,T; H^1(\Omega ))$ of
$h_{\varepsilon_k}$ to $h$ by Lemma \ref{A.1} and with the
definition of weak derivative, we obtain that $v = h_{xx}$ and $h
\in L^2(0,T_{loc};  H^2(\Omega))$ that implies (\ref{Linf_H2}).
Hence $h_{\varepsilon,t} \to h_t \text{ weakly in } L^{2}(0, T;
(H^1(\Omega))')$ that implies (\ref{weak-d}). By Lemma \ref{A.2}
we also have $h \in C([0,T_{loc}],L^2(\Omega))$.
\end{proof}

\begin{proof}[Proof of Theorem \ref{C:Th2}]
Fix $\alpha \in (-1/2,1)$. The initial data $h_0$ is assumed to
have finite entropy $\int G_0^{(\alpha)}(h_0(x)) \; dx < \infty$,
hence Lemma \ref{MainAE2} holds for the approximate solutions $\{
h_{\varepsilon_k} \}$ where this sequence of approximate solutions
is assumed to be the one at the end of the proof of Theorem
\ref{C:Th1}. By (\ref{E:b13}),
$$
\left\{ h_{\varepsilon_k}^{\tfrac{\alpha+2}{2}} \right\} \quad
\mbox{is uniformly bounded in $\varepsilon_k$ in $L^{2}(0,
T_{loc}; H^2(\Omega))$}
$$
and
$$
\left\{ h_{\varepsilon_k}^{\tfrac{\alpha+2}{4}} \right\} \quad
\mbox{is uniformly bounded in $\varepsilon_k$ in $L^{2}(0,
T_{loc}; W^1_4(\Omega))$}.
$$
Taking a further subsequence in $\{ \varepsilon_k \}$, it follows
from the proof of \cite [Lemma 2.5, p.330]{DGG} that these sequences
converge weakly in $L^{2}(0, T_{loc}; H^2(\Omega))$ and $L^{2}(0,
T_{loc}; W^1_4(\Omega))$, to $ h^{\tfrac{\alpha+2}{2}}$ and $
h^{\tfrac{\alpha+2}{4}}$ respectively.
\end{proof}

\section{The subcritical regime: long--time existence of solutions} \label{F}


\begin{lemma}\label{Marinas_bound}
Let $h \in H^1(\Omega)$ be a nonnegative function and let
$M = \int
\limits_{\Omega} {h(x)\,dx}$. Then
\begin{equation}\label{D:nint}
\| h \|_{L^p(\Omega)}^p \leqslant k_1 M^{\tfrac{p+ 2}{3}}
\biggl(\int\limits_{\Omega} {h^2_{x} \,dx}\biggr)^{\tfrac{p -
1}{3}} + k_2 M ^{p}, \ \ p \geqslant 1,
\end{equation}
where $k_1 = 2^{\tfrac{4-p}{3}}3^{\tfrac{2(p-1)}{3}}
(1-\epsilon)^{-1}$, $k_2 = |\Omega|^{1 - p}(1
-(1-\epsilon)^{\tfrac{1}{p-1}})^{1 - p}$, and $\epsilon \in (0,1)$.
\end{lemma}

Note that by taking $h$ to be a constant function, one finds that
the constant $k_2 M^p$ in (\ref{D:nint}) is sharp when $\epsilon
\to 1$.

\begin{proof}
Let $v = h - M/|\Omega|$.  By (\ref{Lady_ineq}),
$$
\| v \|_{L^p(\Omega)}^p \leqslant
(\tfrac{3}{2})^{\tfrac{2(p-1)}{3}} \biggl(\int\limits_{\Omega}
{v^2_{x} \,dx}\biggr)^{\tfrac{p-1}{3}} \biggl(\int\limits_{\Omega}
{|v| \,dx}\biggr)^{\tfrac{p+2}{3}}.
$$
Hence, due to the inequality
$$
|a - b|^p \geqslant (1 - \epsilon) a^p - c_0(\epsilon,p) b^p
$$
for any $a \geqslant 0,\ b \geqslant 0,\ p
> 1, \ \epsilon \in (0,1)$, $ c_0(\epsilon,p) \geqslant \tfrac{1 - \epsilon}{(1
-(1-\epsilon)^{\tfrac{1}{p-1}})^{p-1}}$,
\begin{multline*}
\| h \|_{L^p(\Omega)}^p \leqslant
\tfrac{1}{1-\epsilon}(\tfrac{3}{2})^{\tfrac{2(p-1)}{3}}
\biggl(\int\limits_{\Omega} {h^2_{x} \,dx}\biggr)^{\tfrac{p-1}{3}}
\biggl(\int\limits_{\Omega} {\left|h -  \tfrac{M}{|\Omega|}\right|
\,dx}\biggr)^{\tfrac{p+2}{3}} + \\
\tfrac{1}{(1 -(1-\epsilon)^{\tfrac{1}{p-1}})^{p-1}}
\tfrac{M^p}{|\Omega|^{p-1}}\leqslant
\tfrac{1}{1-\epsilon}(\tfrac{3}{2})^{\tfrac{2(p-1)}{3}}
\biggl(\int\limits_{\Omega} {h^2_{x} \,dx}\biggr)^{\tfrac{p-1}{3}}
(2M)^{\tfrac{p+2}{3}} + \\
\tfrac{1}{(1
-(1-\epsilon)^{\tfrac{1}{p-1}})^{p-1}}\tfrac{M^p}{|\Omega|^{p-1}}.
\end{multline*}

\end{proof}


\begin{lemma}\label{F:ThSub}
Let $ 0 < n/2 \leqslant m < n + 2$. Let $h$ be the generalized
solution of Theorem~\ref{C:Th1}. Then
\begin{equation}\label{F:globest}
\tfrac{a_0}{4} \|  h(.,T_{loc}) \|^2_{H^1(\Omega)}  \leqslant
\mathcal{E}_0(0) + c_1 M^{\tfrac{m - n + 4}{2- m + n}} + c_2 M^{m -
n + 2} + c_3 M^2,
\end{equation}
where $\mathcal{E}_0(0)$ is defined in (\ref{Energy}), and $M =
\int h_0$. Moreover, if $m = n +2$ and $0 < M < M_c$ then
\begin{equation}\label{F:globest'}
\tfrac{a_0}{4} \|  h(.,T_{loc}) \|^2_{H^1(\Omega)} \leqslant
\tfrac{2}{3(1 - c_4 M^2) } \mathcal{E}_0(0) +  \tfrac{c_5}{1 - c_4
M^2 } M^{4}  + c_6 M^2.
\end{equation}
\end{lemma}

\begin{proof}[Proof of Lemma~\ref{F:ThSub}]
We present the proof for the $m-n \neq -1,-2$ case
in (\ref{D_is}), leaving the $m-n=-1,-2$ cases to the reader.
The first step is to find a priori bound (\ref{F:globest}) that is
the analogue of Proposition 2.2 in \cite{BP1}. From (\ref{C:d2'})
we deduce
\begin{equation}\label{F:ee1}
\tfrac{a_0}{2} \int\limits_{\Omega} {h_{x}^2 \,dx} \leqslant
\mathcal{E}_0(0) + \tfrac{a_1}{(m - n + 1)(m- n +
2)}\int\limits_{\Omega} {h^{m - n +2} \,dx}.
\end{equation}
Due to (\ref{D:nint}), we have
\begin{equation}\label{F:ee2}
\int\limits_{\Omega} {h^{m - n +2} \,dx} \leqslant k_1 M^{\tfrac{m
- n + 4}{3}} \Bigl( \int\limits_{\Omega} {h_{x}^2 \,dx}
\Bigr)^{\tfrac{m - n + 1}{3}} + k_2 M^{m - n + 2}.
\end{equation}
Thus, from (\ref{F:ee1}), in view of (\ref{F:ee2}), we find that
\begin{equation}\label{F:ee3}
\tfrac{a_0}{2}\int\limits_{\Omega} {h_{x}^2 \,dx} - \tfrac{a_1 k_1
M^{\tfrac{m - n + 4}{3}}}{(m - n + 1)(m- n + 2)}  \Bigl(
\int\limits_{\Omega} {h_{x}^2 \,dx} \Bigr)^{\tfrac{m - n + 1}{3}}
\leqslant \mathcal{E}_0(0) + \tfrac{a_1 k_2 M^{m - n + 2}}{(m - n
+ 1)(m- n + 2)}.
\end{equation}
Moreover, due to (\ref{D:nint}), in view of the Young inequality
(\ref{Young}), we have
\begin{equation}\label{F:ee2''}
\int\limits_{\Omega} {h^2 \,dx} \leqslant 6^{2/3} M^{\tfrac{4}{3}}
\Bigl( \int\limits_{\Omega} {h_{x}^2 \,dx} \Bigr)^{\tfrac{1}{3}} +
\tfrac{M^{2}}{|\Omega|} \leqslant \tfrac{1}{4}
\int\limits_{\Omega} {h_{x}^2 \,dx} + (\tfrac{8\sqrt{3}}{3} +
|\Omega|^{-1}) M^2.
\end{equation}
In the subcritical ($m<n+2$) case, $\tfrac{m - n + 1}{3} < 1$ and using
(\ref{F:ee3}) and (\ref{F:ee2''}) we deduce (\ref{F:globest}) with
\begin{equation}\label{F:ee4-c}
\begin{gathered}
c_1 = \bigl(\tfrac{a_1 k_1 }{(m - n + 1)(m- n +
2)}\bigr)^{\tfrac{3}{2- m +n}} \bigl(\tfrac{8(m -
n+1)}{3a_0}\bigr)^{\tfrac{m - n + 1}{2- m +n}}\tfrac{2 - m +
n}{3},\\
c_2 = \tfrac{a_1 k_2}{(m - n + 1)(m- n + 2)}, \ c_3 =
\tfrac{a_0}{2}(\tfrac{8\sqrt{3}}{3} + |\Omega|^{-1}).
\end{gathered}
\end{equation}
In the critical ($m=n+2$) case, $\tfrac{m - n + 1}{3} = 1$.
If $M < M_c= \bigl(\tfrac{6a_0}{a_1 k_1}\bigr)^{\tfrac{1}{2}}$ then using
(\ref{F:ee3}) we arrive at
\begin{equation}\label{F:ee4-2}
\int\limits_{\Omega} {h_{x}^2 \,dx} \leqslant \tfrac{12}{6a_0 -
a_1 k_1 M^2 } \bigl( \mathcal{E}_0(0) + \tfrac{a_1 k_2}{12}M^{4}
\bigr)
\end{equation}
Using (\ref{F:ee2''}), from (\ref{F:ee4-2})
we obtain (\ref{F:globest'}) with
\begin{equation}\label{F:ee4-c'}
c_4 = \tfrac{a_1 k_1}{6a_0} , c_5 = \tfrac{a_1 k_2}{18}, \ c_6 =
\tfrac{a_0}{3}(\tfrac{8\sqrt{3}}{3} + |\Omega|^{-1}).
\end{equation}
\end{proof}

Under certain conditions, a bound closely related to
(\ref{F:globest}) implies that if the solution of Theorem
\ref{C:Th1} is initially constant then it will remain constant for
all time:
\begin{theorem}\label{constancy}
Assume $m = n$, the coefficient $a_1 \geqslant 0 $ in (\ref{B:1}), and
$| \Omega |^2 < a_0/|a_1|$. If the initial data is constant, $h_0
\equiv C > 0$, then the solution of Theorem \ref{C:Th1} satisfies
$h(x,t) = C$ for all $x \in \bar{\Omega}$ and all $t > 0$.
\end{theorem}
For the long--wave unstable case ($a_1>0$) the hypotheses correspond
to the domain not being ``too large''.  We note that it is not yet known
whether or not solutions from Theorem \ref{C:Th1} are unique and so
Theorem \ref{constancy} does have content: it ensures that the
approximation method isn't producing unexpected (nonconstant) solutions
from constant initial data.

\begin{proof}
Consider the approximate solution $h_\varepsilon$. The definition
of $\mathcal{E}_\varepsilon(T)$ combined with the uniform-in-time
bound (\ref{D:d2}) implies
\begin{equation} \label{gen_bound2}
\tfrac{a_0}{2} \int\limits_\Omega h_{\varepsilon,x}^2(x,T) \; dx
\leqslant \mathcal{E}_\varepsilon(0) + \tfrac{|a_1|}{2}
\int\limits_\Omega h_\varepsilon^2(x,T) \; dx.
\end{equation}
Letting
$M_\varepsilon = \int h_{0,\varepsilon}\,dx$ and
applying
Poincar\'e's inequality (\ref{Poincare}) to $v_\varepsilon =
h_\varepsilon - M_\varepsilon/|\Omega|$ and using $\int
h_\varepsilon^2 \, dx= \int v_\varepsilon^2 \, dx +
M_\varepsilon^2/|\Omega|$ yields
$$
\left( \tfrac{a_0}{2} - \tfrac{|a_1| \, |\Omega|^2}{2} \right)
\int\limits_\Omega h_{\varepsilon,x}^2(x,t) \; dx \leq
\mathcal{E}_\varepsilon(0) + \tfrac{|a_1| M_\varepsilon^2}{2
|\Omega|} .
$$
If $h_{0,\varepsilon} \equiv C_\varepsilon = C +
\varepsilon^\theta$ this becomes
$$
\left( \tfrac{a_0}{2} - \tfrac{|a_1| |\Omega|^2}{2} \right)
\int\limits_\Omega h_{\varepsilon,x}^2(x,T) \; dx \leqslant (|a_1|
- a_1) \tfrac{C_{\varepsilon}^2 |\Omega|}{2}.
$$
If $a_1 \geqslant 0$ and $|\Omega|^2 < a_0/|a_1|$ then $\int
h_{\varepsilon,x}^2(x,T) \; dx = 0$ for all $T \in
[0,T_{\varepsilon,loc}]$ and that this, combined with the
continuity in space and time of $h_\varepsilon$, implies that
$h_\varepsilon \equiv C_\varepsilon$ on $Q_{T_{\varepsilon,loc}}$.

Taking the sequence $\{ \varepsilon_k \}$ that yields convergence
to the solution $h$ of Theorem \ref{C:Th1}, $h \equiv C$ on
$Q_{T_{loc}}$.
\end{proof}

This $H^1$ control in time of the generalized solution
given by Lemma \ref{F:ThSub}
is now used
to extend the short--time existence result of Theorem \ref{C:Th1}
to a long--time existence result:
\begin{theorem} \label{global}
Let $ 0 < n/2 \leqslant m < n + 2$. Let $T_g$ be an arbitrary
positive finite number. The generalized weak solution $h$ of
Theorem \ref{C:Th1} can be continued in time from $[0,T_{loc}]$ to
$[0,T_g]$ in such a way that $h$ is also a generalized weak
solution and satisfies all the bounds of Theorem \ref{C:Th1} (with
$T_{loc}$ replaced by $T_g$).
\end{theorem}
Similarly, the short--time existence of strong solutions (see
Theorem \ref{C:Th2}) can be extended to long--time existence.

\begin{proof}
To construct a weak solution up to time $T_g$, one applies the local
existence theory iteratively, taking the solution at the final time of
the current time interval as initial data for the next time interval.

Introduce the times
\begin{equation} \label{D:time_ints}
0 = T_0 < T_1 < T_2 < \dots < T_N < \dots \quad \mbox{where} \quad
T_N := \sum_{n=0}^{N-1} T_{n,loc}
\end{equation}
and $T_{n,loc}$ is the interval of existence (\ref{Tloc_is}) for a
solution with initial data $h(\cdot,T_n)$:
\begin{equation} \label{D:Tloc_is}
T_{n,loc} := \tfrac{9}{20 c_{11}(\gamma_1 - 1)} \min\left\{ 1,
\left( \int\limits_\Omega h_x^2(x,T_n) + \tfrac{2c_2}{a_0}
G_0(h(x,T_n)) \; dx\right)^{-(\gamma_1 - 1)} \right\}
\end{equation}
where $\gamma_1 = \max\{3,2m-n\}$ and $c_2$ and $c_{11}$ are given
in the proof of Lemma \ref{MainAE}.

The proof proceeds by contradiction.  Assume there exists initial
data $h_0$, satisfying the hypotheses of Theorem \ref{C:Th1}, that
results in a weak solution that cannot be extended arbitrarily in
time:
$$
\sum_{k=0}^{\infty} T_{n,loc} = T^* < \infty \quad \Longrightarrow
\quad \lim_{n \to \infty} T_{n,loc} = 0.
$$
From the definition (\ref{D:Tloc_is}) of $T_{n,loc}$, this implies
\begin{equation} \label{D:diverges}
\lim_{n \to \infty}
 \int\limits_\Omega (h_x^2(x,T_n)
+ \tfrac{2c_2}{a_0} G_0(h(x,T_n))) \; dx = \infty.
\end{equation}
%
By (\ref{F:globest}),
$$
\tfrac{a_0}{4} \int\limits_\Omega h_x^2(x,T_n) \; dx \leq
\mathcal{E}_0(T_{n-1}) +  K,
$$
where $K = c_1 M^{\tfrac{m - n + 4}{2- m + n}} + c_2 M^{m - n + 2}
+ c_3 M^2$. By (\ref{C:d2'}),
$$
\mathcal{E}_0(T_{n-1}) \leqslant \mathcal{E}_0(T_{n-2})
\leqslant \dots \mathcal{E}_0(0).
$$
Combining these,
\begin{equation} \label{D:here}
\tfrac{a_0}{4} \int\limits_\Omega h_x^2(x,T_n) \; dx \leq
\mathcal{E}_0(0) + K.
\end{equation}

By assumption, $T_n \to T^* < \infty$ as $n \to \infty$ hence
$\int h_x^2(x,T_n)$ remains bounded. Assumption (\ref{D:diverges})
then implies that $\int G_0(h(x,T_n))  \to \infty$ as $n \to
\infty$.

To continue the argument, we step back to the approximate
solutions $h_\varepsilon$.  Let $T_{n,\varepsilon}$ be the
analogue of $T_n$ and $T_{n,loc,\varepsilon}$, defined by
(\ref{Tloc_eps_is}), be the analogue of $T_{n,loc}$. By
(\ref{D:aa2}),
\begin{align} \label{D:2}
& \int\limits_\Omega
G_\varepsilon(h_\varepsilon(x,T_{n,\varepsilon})) \; dx
\leqslant \int\limits_\Omega G_\varepsilon(h_\varepsilon(x,T_{n-1,\varepsilon})) \; dx \\
& \hspace{2in} \notag + c_{10}
\int\limits_{T_{n-1,\varepsilon}}^{T_{n,\varepsilon}} \max \left\{
1, \Bigl( \int\limits_\Omega h_{\varepsilon,x}^2(x,T) \; dx
\Bigr)^{\gamma_2} \right\} \; dT
\end{align}
with $\gamma_2 = 3$.
Using the bound (\ref{D:d2}), one can prove the analogue of Lemma
\ref{F:ThSub} for the approximate solution $h_\varepsilon$.
However the bound (\ref{F:globest}) would be replaced by a bound
on $\| h_\varepsilon(\cdot,T) \|_{H^1}$ which holds for all $T \in
[0,T_{\varepsilon,loc}]$.  This bound would then be used to prove
a bound like (\ref{D:here}) to prove boundedness of $\int
h_{\varepsilon,x}^2(x,T)$ for all $T \in [0,T_{n,\varepsilon}]$.
Using this bound,

\begin{align}
\notag & \int\limits_{T_{n-1,\varepsilon}}^{T_{n,\varepsilon}}
\Bigl( \int\limits_\Omega h_{\varepsilon,x}^2(x,T) \; dx
\Bigr)^{\gamma_2} dT \leq
\bigl(\tfrac{4}{a_0} \bigr)^{\gamma_2} \int\limits_{T_{n-1,\varepsilon}}^{T_{n,\varepsilon}}
\left( \mathcal{E}_\varepsilon(0) + K \right)^{\gamma_2} \,dT \\
& \hspace{1in} \label{D:3} = \bigl(\tfrac{4}{a_0}
\bigr)^{\gamma_2} \left( \mathcal{E}_\varepsilon(0) + K
\right)^{\gamma_2}\, T_{n-1,loc,\varepsilon}.
\end{align}
If the initial data is such that $4/a_0 \;
(\mathcal{E}_\varepsilon(0) + K) < 1$ then before using
(\ref{D:3}) in (\ref{D:2}) we replace $K$ by a larger value so
that $4/a_0 \; (\mathcal{E}_\varepsilon(0) + K)>1$. Using
(\ref{D:3}) in (\ref{D:2}), it follows that
\begin{equation} \label{D:4}
 \int\limits_\Omega G_\varepsilon(h_\varepsilon(x,T_{n,\varepsilon})) \; dx
 \leqslant \int\limits_\Omega
G_\varepsilon(h_\varepsilon(x,T_{n-1,\varepsilon})) \; dx + \beta
\, T_{n-1,loc,\varepsilon}
\end{equation}
for $\beta = c_{10}\bigl(\tfrac{4}{a_0} \bigr)^{\gamma_2} \left(
\mathcal{E}_\varepsilon(0) + K \right)^{\gamma_2}$.
Here $\beta$ depends on $|\Omega|$, the coefficients of the PDE,  and
on the initial data $h_{0,\varepsilon}$.

One now takes the sequence $\{ \varepsilon_k \}$ that was used to
construct the weak solution of Theorem \ref{C:Th1} on the interval
$[T_{n-1},T_n]$.  Taking $\varepsilon_k \to 0$, (\ref{D:4}) yields
\begin{equation}
\label{D:4a} \int\limits_\Omega G_0(h(x,T_{n}))\; dx \leq
\int\limits_\Omega G_0(h(x,T_{n-1}))\,dx + \beta \, T_{n-1,loc}.
\end{equation}
Applying (\ref{D:4a}) iteratively,
\begin{align*}
&  \int\limits_\Omega G_0(h(x,T_n)) \; dx \leqslant  \int\limits_\Omega
G_0(h_0(x)) \; dx
+ \beta \sum_{k=0}^{n-1} T_{k,loc} \\
& \hspace{1in} =  \int\limits_\Omega G_0(h_0(x)) \; dx + \beta \,
T_{n} .
\end{align*}
This upper bound proves that $\int G_0(h(x,T_n))$ cannot diverge
to infinity as $n \to \infty$, finishing the proof.
\end{proof}

\section{Strong positivity}

\begin{proposition}\label{F:LemPos}
Let
$m \geqslant n/2$. Assume the initial data $h_0$
satisfies $h_0(x) > 0 $ for all $x \in \omega \subseteq \Omega$
where $\omega$ is an open interval. Then

\begin{enumerate}
\item
if $n > 3/2$ and $\alpha \in (-1/2, \min\{1,n-2\})$
then the strong solution from Theorem \ref{C:Th2} satisfies
$h(x,T) > 0$ for almost every $x \in \omega$,
for all $T \in [0,T_{loc}^{(\alpha)}]$;
\item
if $n > 2$ and $\alpha \in (-1/2, \min\{1, 3n/4 - 2\})$
then the strong solution from Theorem \ref{C:Th2} satisfies
$h(x,T) > 0$ for all $x \in \omega$, for almost
every $T \in [0,T_{loc}^{(\alpha)}]$;
\item
if $n \geqslant 7/2$ and $m \geqslant n-1/2$
then the strong solution from Theorem \ref{C:Th2}
satisfies $h(x,T) > 0$ for all $(x, T) \in
\overline{Q}_{T_{loc}^{(\alpha)}}$.
\end{enumerate}
\end{proposition}

The proof of Proposition~\ref{F:LemPos} depends on a local version
of the a priori bound (\ref{E:b12''}) of Lemma \ref{MainAE2}:

\begin{lemma} \label{F:local_BF}  Let $\omega \subseteq \Omega$ be an open interval
and $\zeta \in C^2(\bar{\Omega})$ such that $\zeta > 0$ on
$\omega$, $\text{supp}\,\zeta = \overline{\omega}$, and
$(\zeta^4)^{\prime} = 0$ on $\partial\Omega$. If $\omega =
\Omega$, choose $\zeta$ such that $ \zeta(-a) = \zeta(a)
> 0$. Let $\xi := \zeta^4$.

If the initial data $h_0$ and the time $T_{loc}^{(\alpha)}$ are as
in Theorem \ref{C:Th2} then for all $T \in [0,T_{loc}^{(\alpha)}]$
the strong solution $h$ from Theorem \ref{C:Th2} satisfies
\begin{equation} \label{local_BF_finite}
\int\limits_\Omega {\xi(x) \; G_0^{(\alpha)}(h(x,T))} \; dx <
\infty
\end{equation}
\end{lemma}
The proof of Lemma \ref{F:local_BF} is given in Appendix
\ref{A_priori_proofs}. The proof of Proposition \ref{F:LemPos} is
essentially a combination of the proofs of Theorem 6.1 and Corollary 4.5
in \cite{B8} and is provided here for the reader's
convenience.
\begin{proof}[Proof of Proposition~\ref{F:LemPos}]
Choose the test function $\zeta(x)$ to satisfy the
hypotheses of Lemma \ref{F:local_BF}.  Hence,
(\ref{local_BF_finite}) holds for every $T \in
[0,T_{loc}^{(\alpha)}]$.

Proof of 1):
Assume it is not true that $h(x,T)>0$ for almost every $x \in \omega$,
for all $T \in  [0,T_{loc}^{(\alpha)}]$.
Then there is a
time $T \in [0,T_{loc}^{(\alpha)}]$ such that the set $\{ x \; :
\; h(x,T) = 0 \} \cap \omega$ has positive measure.  Then
because $\alpha - n + 2 < 0$,
$$
\infty > \int\limits_\Omega \xi(x)  h^{\alpha - n + 2}(x,T) \; dx
\geqslant \int\limits_{ \{h(\cdot,T) = 0 \} \cap \omega } \xi(x)
 h^{\alpha - n + 2}(x,T) \; dx = \infty.
$$
This
contradiction implies there can be no time at which $h$ vanishes
on a set of positive measure in $\omega$, as desired.

Proof of 2):
We start by noting that by
(\ref{Linf_H2-2}),
$(h^{\tfrac{\alpha + 2}{2}})_{xx}(\cdot,T) \in L^2(\Omega)$ for
almost all $T \in [0,T_{loc}^{(\alpha)}]$ hence $h^{\tfrac{\alpha +
2}{2}}(\cdot,T) \in C^{3/2}(\Omega)$ for almost all $T \in
[0,T_{loc}^{(\alpha)}]$.  To prove 2), it therefore
suffices to show that
if $T$ such that $h^{\tfrac{\alpha + 2}{2}}(\cdot,T) \in C^{3/2}(\Omega)$
then $h(x,T)>0$ on $\omega$.  Assume this is not true and
there is an $x_0 \in \omega$ and
$T_0$ such that $h(x_0,T_0) = 0$ and
$h^{\tfrac{\alpha + 2}{2}}(\cdot,T_0) \in C^{3/2}(\Omega)$.
Then there is a $L$
such that
$$
h^{\tfrac{\alpha + 2}{2}}(x,T_0) = |h^{\tfrac{\alpha +
2}{2}}(x,T_0)-h^{\tfrac{\alpha + 2}{2}}(x_0,T_0)| \leqslant L |
x-x_0|^{3/2}.
$$
Hence
$$
\infty > \int\limits_\Omega {\xi(x) h^{\alpha - n +2 }(x,T_0)} \;
dx \geqslant L^{\tfrac{2(\alpha - n +2)}{\alpha + 2}} \int\limits_\Omega
\xi(x) |x - x_0|^{- \tfrac{3(n
 - \alpha - 2)}{\alpha + 2}} \; dx = \infty.
$$
This contradiction implies there can be no point $x_0$ such
that $h(x_0,T_0) = 0$, as desired.  Note that we used $\xi > 0$ on
$\omega$ and $x_0 \in \omega$ to conclude that the integral
diverges.


Proof of 3):
Taking $\alpha = -\tfrac{1}{2}$ in (\ref{E:b5}), the approximate solution
$h_\eps$ satisfies
\begin{equation} \label{E:here}
\int\limits_{\Omega} {G^{(-1/2)}_{\varepsilon}(h_\eps(x,T))\,dx} \leqslant
\int\limits_{\Omega}
{G^{(-1/2)}_{\varepsilon}(h_{0\varepsilon})\,dx}
+\iint\limits_{Q_T} {h_\eps^{m - n - \tfrac{1}{2}} h^2_{\eps,x}\,dx dt}.
\end{equation}
Now we use the estimate
\begin{multline*}
\iint\limits_{Q_T} {h_\eps^{m -n- \tfrac{1}{2}} h_{\eps,x}^2 \,dx dt}= -
\tfrac{2}{2 m-2 n +1}\iint\limits_{Q_T} {h_\eps^{m -n + \tfrac{1}{2}}h_{\eps,xx} \,dx dt} \\
 \leqslant
\tfrac{2}{2 m-2 n + 1} \left(\iint\limits_{Q_T} { h^2_{\eps,xx}\,dx dt}
\right)^{\tfrac{1}{2}} \left(\iint\limits_{Q_T} { h_\eps^{2 m-2 n +1}\,dx
dt}\right)^{\tfrac{1}{2}}.
\end{multline*}
Using $m \geqslant n-1/2$ and the $L^\infty$ bound on $h_\eps$ from (\ref{D:a13''})
as well as
(\ref{BF_entropy}), we obtain
\begin{equation}\label{E:ggg}
\iint\limits_{Q_T} {h_\eps^{m -n- \tfrac{1}{2}} h_{\eps,x}^2 \,dx dt}
\leqslant C \sqrt{T}, \text{ i.\,e. } \iint\limits_{Q_T} {\left(h_\eps^{\tfrac{2 m
-2 n+3}{4} } \right)_{x}^2 \,dx dt} \leqslant C \sqrt{T}.
\end{equation}
Returning to (\ref{E:here}) and taking $\eps_k \to 0$ along the subsequence that yielded
the strong solution $h$, we find that
\begin{equation}\label{E:gg1}
\int\limits_{\Omega} {G^{(-1/2)}_{0}(h)\,dx} = \tfrac{1}{(n -
3/2)(n - 1/2)}\int\limits_{\Omega} {h^{3/2 - n}(x,T)\,dx}
\leqslant K < \infty.
\end{equation}
We now prove $h(x,T) > 0$ for all $x \in \bar{\Omega}$, and  for
every $T \in [0,T_{loc}^{(\alpha)}]$. By (\ref{weak1}),
$h(\cdot,T) \in C^{1/2}(\bar{\Omega})$ for all $T \in
[0,T_{loc}^{(\alpha)}]$. Assume $T_0$ is such that
$h(x_0,T_0) = 0$ at some $x_0 \in
\Omega$. Then there is a $L$ such that
$$
h(x,T_0) = |h(x,T_0)-h(x_0,T_0)| \leqslant L | x-x_0|^{1/2}.
$$
Hence since $n \geqslant 7/2$
$$
\infty > \int\limits_\Omega {h^{\tfrac{3}{2} - n }(x,T_0)} \; dx \geq
L^{\tfrac{3 - 2n}{2}} \int\limits_\Omega |x - x_0|^{- \tfrac{2n
 - 3}{4}} \; dx = \infty.
$$
This contradiction implies there can be no point $x_0$ such
that $h(x_0,T_0) = 0$, as desired.
\end{proof}

\section{Finite speed of propagation}\label{G}

\subsection{Local entropy estimate}

\begin{lemma}\label{G:lem1}
Let $\zeta \in C^{1,2}_{t,x}(\bar{Q}_T)$ such that
$\text{supp}\,\zeta \subset \Omega$, $(\zeta^4)_{x} = 0$ on
$\partial\Omega$, and $\zeta^4(-a,t) = \zeta^4(a,t)$. Assume that $
- \tfrac{1}{2} < \alpha < 1$, and $\alpha \neq 0$. Then there exist a
weak solution $h(x,t)$ in the sense of the theorem
Theorem~\ref{C:Th2}, constants $C_i \ (i=1,2,3)$ dependent on
$n,\,m,\,\alpha,\, a_0$, and $a_1$, independent of $\Omega$, such
that for all $T \leqslant T_{loc}^{(\alpha)}$
\begin{multline}\label{G:main1}
\int\limits_{\Omega} {\zeta^4(x,T)G^{(\alpha)}_{0}(h(x,T))\,dx} -
\iint\limits_{Q_T} {(\zeta^4)_t G^{(\alpha)}_{0}(h)\,dx dt} + C_1
\iint\limits_{Q_T} { (h^{\tfrac{\alpha + 2}{2}})^2_{xx}\zeta^4 \,dx
dt} \leqslant \\
\int\limits_{\Omega} {\zeta^4(x,0)G^{(\alpha)}_{0}(h_{0})\,dx} +
C_2 \iint\limits_{Q_T} { h^{\alpha +2} (\zeta_{x}^4 + \zeta^2
\zeta_{xx}^2) \,dx dt} + C_3 \iint\limits_{Q_T} { h^{2(m - n + 1)+
\alpha}\zeta^4 \,dx dt}.
\end{multline}
\end{lemma}

\begin{proof}[Sketch of Proof of Lemma~\ref{G:lem1}]
In the following, we denote the positive, classical solution
$h_\varepsilon$ constructed in Lemma~\ref{MainAE2} by $h$
(whenever there is no chance of confusion).

Let $\phi(x,t) = \zeta^4(x,t)$. Recall the entropy function
$G^{(\alpha)}_{\varepsilon}(z)$ defined by (\ref{E:reg4}).
Multiplying (\ref{D:1r'}) by $\phi(x,t)
G'^{(\alpha)}_{\varepsilon}(h_{\varepsilon})$, and integrating
over $Q_T$ yields
\begin{multline}\label{G:in1}
\int\limits_{\Omega}
{\phi(x,T)G^{(\alpha)}_{\varepsilon}(h(x,T))\,dx} -
\int\limits_{\Omega}
{\phi(x,0)G^{(\alpha)}_{\varepsilon}(h_{0\varepsilon})\,dx} - \\
\iint\limits_{Q_T} {\phi_t(x,t)G^{(\alpha)}_{\varepsilon}(h)\,dx
dt} = \iint\limits_{Q_T} {\phi_x f_{\varepsilon}(h)
G'^{(\alpha)}_{\varepsilon}(h) (a_0 h_{xxx} + a_1
D''_{\varepsilon}(h)h_x)\,dx dt} + \\
 \iint\limits_{Q_T} {\phi h^{\alpha}(a_0  h_x h_{xxx} + a_1  D''_{\varepsilon}(h)h_x^2 ) \,dx dt} =: I_{1} +
 I_{2}.
\end{multline}
We now bound the terms $I_1$ and $I_2$.  First,
\begin{multline}\label{G:in2}
I_{1}  = - a_0 \iint\limits_{Q_T} {\phi_{xx}
f_{\varepsilon}(h)G'^{(\alpha)}_{\varepsilon}(h) h_{xx}\,dx dt} -
a_0 \iint\limits_{Q_T} {\phi_x (h^{\alpha} +
f'_{\varepsilon}(h)G'^{(\alpha)}_{\varepsilon}(h))h_x h_{xx}\,dx dt}-\\
a_1 \iint\limits_{Q_T} {\phi_{xx} F^{(\alpha)}_{\varepsilon}(h)
\,dx dt}  = - a_0 \iint\limits_{Q_T} {\phi_{xx}
f_{\varepsilon}(h)G'^{(\alpha)}_{\varepsilon}(h) h_{xx}\,dx dt} +
\tfrac{a_0}{2} \iint\limits_{Q_T} {\phi_{xx}h^{\alpha} h_x^2\,dx
dt} + \\
\tfrac{a_0 \alpha}{2} \iint\limits_{Q_T} {\phi_{x}h^{\alpha -1}
h_x^3\,dx dt} - a_0 \iint\limits_{Q_T} {\phi_x
f'_{\varepsilon}(h)G'^{(\alpha)}_{\varepsilon}(h)h_x h_{xx}\,dx
dt} - a_1 \iint\limits_{Q_T} {\phi_{xx}
F^{(\alpha)}_{\varepsilon}(h) \,dx dt} ,
\end{multline}
where $F^{(\alpha)}_{\varepsilon}(z) := \int \limits_0^z
{f_{\varepsilon}(s)G'^{(\alpha)}_{\varepsilon}(s)g'_{\varepsilon}(s)\,ds}$,
\begin{multline}\label{G:in3}
I_{2} =  - a_0 \iint\limits_{Q_T} {(\phi_x h^{\alpha}  h_x  h_{xx}
+ \alpha \phi h^{\alpha -1} h^2_{x} h_{xx} +  \phi h^{\alpha}
h^2_{xx}) \,dx dt} + \\
a_1 \iint\limits_{Q_T} {\phi h^{\alpha} D''_{\varepsilon}(h)h_x^2
\,dx dt}  = \tfrac{a_0}{2} \iint\limits_{Q_T} {\phi_{xx}
h^{\alpha} h_x^2 \,dx dt} +
\tfrac{5a_0 \alpha}{6} \iint\limits_{Q_T} {\phi_x h^{\alpha - 1} h_x^3 \,dx dt} +\\
\tfrac{a_0 \alpha(\alpha - 1)}{3} \iint\limits_{Q_T} {\phi
h^{\alpha - 2} h_x^4 \,dx dt}  - a_0 \iint\limits_{Q_T} {\phi
h^{\alpha} h^2_{xx} \,dx dt} + a_1\iint\limits_{Q_T} {\phi
h^{\alpha} D''_{\varepsilon}(h)h_x^2 \,dx dt}.
\end{multline}

One easily finds that for all $\varepsilon > 0$ and all $z \geqslant 0$
$$
|f_{\varepsilon}(z)G'^{(\alpha)}_{\varepsilon}(z)| \leqslant
\tfrac{z^{\alpha + 1}}{|\alpha - n + 1|},\
|f_{\varepsilon}'(z)G'^{(\alpha)}_{\varepsilon}(z)| \leqslant
\tfrac{s z^{\alpha }}{|\alpha - n + 1|},\ |F_{\varepsilon}(z)|
\leqslant \tfrac{z^{m - n +  \alpha + 2}}{|\alpha - n + 1|(m - n +
\alpha + 2)} + o(\eps).
$$
Using these bounds, and the Cauchy inequality, we bound $I_1 +
I_2$:
\begin{multline}\label{G:in4}
I_1 + I_2 \leqslant \tfrac{a_0 \alpha(\alpha - 1)}{3}
\iint\limits_{Q_T} { h^{\alpha - 2} h_x^4 \phi\,dx dt}  - a_0
\iint\limits_{Q_T} { h^{\alpha} h^2_{xx} \phi \,dx dt} + \\
a_1\iint\limits_{Q_T} { h^{m - n + \alpha} h_x^2 \phi \,dx dt} +
a_0 \iint\limits_{Q_T} { h^{\alpha} h_x^2 |\phi_{xx}| \,dx dt} +
\\
\tfrac{a_0}{|\alpha - n + 1|} \iint\limits_{Q_T} { h^{\alpha+1}
|h_{xx}| |\phi_{xx}| \,dx dt} + \tfrac{a_0 s}{|\alpha - n + 1|}
\iint\limits_{Q_T} {h^{\alpha}|h_x| |h_{xx}| |\phi_{x}| \,dx dt} +
\\
+ \tfrac{4 a_0 \alpha}{3} \iint\limits_{Q_T} { h^{\alpha - 1}
h_{x}^3 \phi_{x} \,dx dt} + \tfrac{a_1}{|\alpha - n + 1|(m - n +
\alpha + 2)} \iint\limits_{Q_T} {h^{m - n + \alpha + 2}
|\phi_{xx}| \,dx dt} .
\end{multline}
Due to (\ref{G:in4}),  we deduce from (\ref{G:in1}) that
\begin{multline}\label{G:in5}
\int\limits_{\Omega}
{\phi(x,T)G^{(\alpha)}_{\varepsilon}(h(T))\,dx} -
\iint\limits_{Q_T} {\phi_t(x,t)G^{(\alpha)}_{\varepsilon}(h)\,dx
dt} +  a_0 \iint\limits_{Q_T} { h^{\alpha} h^2_{xx} \phi \,dx dt}
 + \\
\tfrac{a_0 \alpha(1 - \alpha)}{3} \iint\limits_{Q_T} { h^{\alpha -
2} h_x^4 \phi\,dx dt}  \leqslant \int\limits_{\Omega}
{\phi(x,0)G^{(\alpha)}_{\varepsilon}(h_{0\varepsilon})\,dx} +\\
a_1\iint\limits_{Q_T} { h^{m - n + \alpha} h_x^2 \phi \,dx dt} +
a_0 \iint\limits_{Q_T} { h^{\alpha} h_x^2 |\phi_{xx}| \,dx dt} +
\\
\tfrac{a_0}{|\alpha - n + 1|} \iint\limits_{Q_T} { h^{\alpha+1}
|h_{xx}| |\phi_{xx}| \,dx dt} + \tfrac{a_0 s}{|\alpha - n + 1|}
\iint\limits_{Q_T} {h^{\alpha}|h_x| |h_{xx}| |\phi_{x}| \,dx dt} +
\\
\tfrac{4 a_0 \alpha}{3} \iint\limits_{Q_T} { h^{\alpha - 1}
h_{x}^3 \phi_{x} \,dx dt} + \tfrac{a_1}{|\alpha - n + 1|(m - n +
\alpha + 2)} \iint\limits_{Q_T} {h^{m - n + \alpha + 2}
|\phi_{xx}| \,dx dt},
\end{multline}
where $\alpha \neq n - 1$. Recalling $\phi = \zeta^4$, and using
the Young's inequality (\ref{Young}) and simple transformations,
from (\ref{G:in5}) we find that
\begin{multline}\label{G:in6}
\int\limits_{\Omega}
{\zeta^4(x,T)G^{(\alpha)}_{\varepsilon}(h(x,T))\,dx} -
\iint\limits_{Q_T} {(\zeta^4)_t G^{(\alpha)}_{\varepsilon}(h)\,dx
dt} +  \\
C_1 \iint\limits_{Q_T} { (h^{\tfrac{\alpha + 2}{2}})^2_{xx}\zeta^4
\,dx dt} \leqslant \int\limits_{\Omega}
{\zeta^4(x,0)G^{(\alpha)}_{\varepsilon}(h_{0,\varepsilon})\,dx} +\\
C_2 \iint\limits_{Q_T} { h^{\alpha +2} (\zeta_{x}^4 + \zeta^2
\zeta_{xx}^2) \,dx dt} + C_3 \iint\limits_{Q_T} { h^{2(m - n + 1)+
\alpha}\zeta^4 \,dx dt}.
\end{multline}

We now argue that the $\varepsilon \to 0$ limit of the right-hand
side of (\ref{G:in6}) is finite and bounded by $K$, allowing us to
apply Fatou's lemma to the left-hand side of (\ref{G:in6}),
concluding
$$
\int\limits_\Omega \zeta^4(x,T) \; G^{(\alpha)}_0(h(x,T)) \; dx-
\iint\limits_{Q_T} {(\zeta^4)_t G^{(\alpha)}_{0}(h)\,dx dt} + C_1
\iint\limits_{Q_T} { (h^{\tfrac{\alpha + 2}{2}})^2_{xx}\zeta^4 \,dx
dt} \leqslant K < \infty
$$
for every $T \in [0,T_{loc}^{(\alpha)}]$, as desired.  (Note that
in taking $\varepsilon \to 0$ we will choose the exact same
sequence $\varepsilon_k$ that was used to construct the weak
solution $h$ of Theorem \ref{C:Th2}.  Also, in applying Fatou's
lemma we used the fact that $\{ h = 0 \}$ having measure zero in
$Q_{T_{loc}^{(\alpha)}}$ implies $\{ h(\cdot,T) \}$ has measure
zero in $\Omega$.)

It suffices to show that $\int \zeta^4(x,0)
G^{(\alpha)}_\varepsilon(h_{0, \varepsilon}) \to \int \zeta^4(x,0)
G^{(\alpha)}_0(h_0) < \infty$ as $\varepsilon \to 0$ (the rest of
items is bonded as $h \in L^{\infty}(0, T_{loc}^{(\alpha)};
H^1(\Omega))$). This uses the Lebesgue Dominated Convergence
Theorem. First, note that
$$
G^{(\alpha)}_\varepsilon(z) = \tfrac{z^{\alpha -n+ 2 }}{(\alpha -
n + 2)(\alpha - n + 1)} + \tfrac{\varepsilon z^{\alpha -
2}}{(\alpha - 2)(\alpha - 3)} = G^{(\alpha)}_0(z) +
\tfrac{\varepsilon z^{\alpha -2}}{(\alpha -  2)(\alpha - 3)},
$$
hence if $h_0(x) > 0$ then
$$
 G^{(\alpha)}_\varepsilon(h_{0, \varepsilon}(x)) =
G^{(\alpha)}_0(h_0(x) + \varepsilon^\theta) + \tfrac{\varepsilon
(h_0(x) + \varepsilon^\theta)^{\alpha -2}}{(\alpha -
2)(\alpha - 3)} \leqslant G^{(\alpha)}_0(h_0(x) +
\varepsilon^\theta) + \tfrac{\varepsilon^{1 - \theta(2 - \alpha)}}{|(\alpha - 2)(\alpha - 3)|}.
$$
Because $h_0$ has finite entropy ($\int G^{(\alpha)}_0(h_0) <
\infty$) it is positive almost everywhere in $\Omega$. Using this
and the fact that $\theta$ was chosen so that $\theta < 1/(2 -
\alpha) < 2/5$, we have $|\phi(x,0) \,
G^{(\alpha)}_\varepsilon(h_{0,\varepsilon}(x))| \leqslant \phi(x,0)
(G^{(\alpha)}_0(h_0(x)) + c) \leqslant C (G_0(h_0(x)) + c )$ almost
everywhere in $x$ and for all $\varepsilon < \varepsilon_0$. The
dominating function is in $L^1$, because $h_0$ has finite entropy.

It remains to show pointwise convergence $\phi(x,0)
G^{(\alpha)}_\varepsilon(h_{0, \varepsilon}(x)) \to \phi(x,0)
G_0(h_0(x))$ almost everywhere in $x$:
\begin{align*}
& \left| G^{(\alpha)}_\varepsilon(h_{0,\varepsilon}(x))  -
G^{(\alpha)}_0(h_0(x)) \right| \leq
\left| G^{(\alpha)}_\varepsilon(h_{0,\varepsilon}(x))  - G^{(\alpha)}_0(h_{0, \varepsilon}(x)) \right| \\
& \hspace{.4in} + \left| G^{(\alpha)}_0(h_{0, \varepsilon}(x))  -
G^{(\alpha)}_0(h_0(x)) \right| = \tfrac{\varepsilon
h_{0,\varepsilon}^{\alpha -2}(x)}{(\alpha - 2)(\alpha - 3)} +
\left| G^{(\alpha)}_0(h_{0, \varepsilon}(x))  - G^{(\alpha)}_0(h_0(x)) \right| \\
& \hspace{.4in} \leqslant \tfrac{\varepsilon^{1 - \theta(2 - \alpha)}}{|(\alpha - 2)(\alpha -3)|} +\left|
G^{(\alpha)}_0(h_{0,\varepsilon}(x))  - G^{(\alpha)}_0(h_0(x))
\right|
\end{align*}
As before, $\tfrac{\varepsilon^{1 - \theta(2 - \alpha)}}{|(\alpha - 2)(\alpha - 3)|}$ goes to zero by the
choice of $\theta$. The term $\left|
G^{(\alpha)}_0(h_{0,\varepsilon}(x)) - G^{(\alpha)}_0(h_0(x))
\right|$ goes to zero for almost every $x \in \Omega$ because
$G^{(\alpha)}_0(z)$ is continuous everywhere except at $z=0$.

The proof is similar for the case $\alpha = n - 1$ and $\alpha = n
- 2$.
\end{proof}

\subsection{Proof of Theorem~\ref{C:Th4} for the case $0 < n < 2$}

Let $0 < n < 2$, and let $\text{supp}\, h_0 \subseteq (-r_0,r_0)
\Subset \Omega$. For an arbitrary $s \in (0, a - r_0)$ and $\delta
> 0$ we consider the families of sets
\begin{equation}\label{G:om}
\Omega (s) = \Omega \setminus (-r_0 - s, r_0 + s),\ Q_T(s)= (0,T)
\times \Omega (s).
\end{equation}
We introduce a nonnegative cutoff function $\eta(\tau)$ from the
space $C^2(\mathbb{R}^1)$ with the following properties:
\begin{equation}\label{G:test1}
\eta(\tau) = \left\{
\begin{aligned}
\hfill 0 \  & \ \tau \leqslant 0,\\
\hfill \tau^2(3 -2\tau) \  & \ 0  < \tau < 1,\\
\hfill  1 \ & \ \tau \geqslant 1.
\end{aligned}\right.
\end{equation}
Next we introduce our main cut-off functions $\eta _{s,\delta }
(x) \in C^2 (\bar{\Omega})$ such that $0 \leqslant \eta_{s,\delta
} (x) \leqslant 1 \  \forall\,  x \in \bar{\Omega}$ and possess
the following properties:
\begin{equation}\label{G:test2}
\eta _{s,\delta } (x) = \eta \bigl( \tfrac{|x| - (r_0 +
s)}{\delta} \bigr) = \left\{
\begin{aligned} \hfill  1 \;
& , x \in \Omega(s + \delta),\\
\hfill  0\; & , x \in \Omega(s ), \\
\end{aligned}\right. \ \ \  | (\eta _{s,\delta })_x| \leqslant \tfrac{3} {\delta },\
| (\eta _{s,\delta })_{xx}| \leqslant \tfrac{6} {\delta^2}
\end{equation}
for all $s > 0,\ \delta >0 : r_0 + s + \delta < a$. Choosing
$\zeta^4(x,t) = \eta _{s,\delta } (x) e^{-\tfrac{t}{T}}$, from
(\ref{G:main1}) we arrive at
\begin{multline}\label{G:main2}
\int\limits_{\Omega (s + \delta)} {h^{\alpha - n + 2}(T)\,dx} +
\tfrac{1}{T} \iint\limits_{Q_T(s + \delta)} { h^{\alpha - n +
2}\,dx dt} + C\iint\limits_{Q_T(s + \delta)} { (h^{\tfrac{\alpha +
2}{2}})^2_{xx} \,dx dt} \leqslant \\
\tfrac{C}{\delta^4} \iint\limits_{Q_T(s)} { h^{\alpha +2}\,dx dt}
+ C \iint\limits_{Q_T(s)} { h^{2(m - n + 1)+ \alpha} \,dx dt} =: C
\sum \limits_{i = 1}^2 {\delta^{- \alpha_i}\iint\limits_{Q_T(s)}{
h^{\xi_i}}}
\end{multline}
for all $s > 0$, where we consider $(n - 1)_+ < \alpha < 1$. We apply the
Gagliardo-Nirenberg interpolation inequality (see Lemma~\ref{A.4})
in the region $\Omega(s+\delta)$ to a function $v := h^
{\tfrac{\alpha + 2}{2}}$ with $a = \tfrac{2\xi_i }{\alpha + 2},\ b =
\tfrac{2(\alpha - n +2)}{\alpha + 2},\ d = 2,\ i = 0,\ j = 2$, and
$\theta_i = \tfrac{(\alpha + 2)(\xi_i - \alpha + n -
2)}{\xi_i(4\alpha - 3n + 8)}$ under the conditions:
\begin{equation}\label{G:con1'}
\alpha - n + 2 < \xi_i   \text{ for } i = 1,2.
\end{equation}
Integrating the resulted inequalities with respect to time and
taking into account (\ref{G:main2}), we arrive at the following
relations:
\begin{equation}\label{G:in11'}
\iint\limits_{Q_T(s +\delta)}{ h^{\xi_i}} \leqslant C\, T^{1 -
\tfrac{\theta_{i}\xi_i}{\alpha + 2} } \Biggl( \sum \limits_{i =
1}^2 {\delta^{- \alpha_i}\iint\limits_{Q_T(s)}{
h^{\xi_i}}}\Biggr)^{1 + \nu_{i}} \!\!\!\!\! + C\, T \Biggl( \sum
\limits_{i = 1}^2 {\delta^{- \alpha_i}\iint\limits_{Q_T(s)}{
h^{\xi_i}}}\Biggr)^{\tfrac{\xi_i} {\alpha - n + 2}}\!\!\!\!,
\end{equation}
where $\nu_{i} = \tfrac{4(\xi_i - \alpha + n - 2)}{4\alpha - 3n +
8}$. These inequalities are true provided that
\begin{equation}\label{G:con2'}
\tfrac{\theta_{i}\xi_i}{\alpha + 2} < 1 \Leftrightarrow  \xi_i <
5\alpha - 4n + 10 \text{ for } i = 1,2.
\end{equation}
Simple calculations show that inequalities (\ref{G:con1'}) and
(\ref{G:con2'}) hold with some $(n - 1)_+ < \alpha < 1$ if and
only if
$$
0 < n < 2, \ \ \tfrac{n}{2} < m < 6 - n.
$$
Since all integrals on the right-hand sides of (\ref{G:in11'})
vanish as $T \to 0$, the finite speed of propagations  follows
from (\ref{G:in11'}) by applying Lemma~\ref{A.5} with $s_1 = 0$
and sufficiently small $T$. Hence,
\begin{equation}\label{G:sp1}
\textnormal{supp}\,h(T,.) \subset (-r_0 - \Gamma_0 (T), r_0 +
\Gamma_0 (T) ) \Subset \Omega \ \text{for all } T: 0 \leqslant T
\leqslant T_{speed}.
\end{equation}

Due to (\ref{G:sp1}), we can consider the solution $h(x,t)$ with
compact support in the whole space $\mathbb{R}^1$ and for $T
\leqslant T_{speed}$. In this case, we can repeat the previous
procedure for $\Omega(s) = \mathbb{R}^1 \setminus (-r_0 - s, r_0 +
s)$ and we obtain
\begin{equation}\label{G:int11-2}
G_i(s + \delta) := \iint\limits_{Q_T(s +\delta)}{ h^{\xi_i}}
\leqslant C\, T^{1 - \tfrac{\theta_{i}\xi_i}{\alpha + 2} } \Biggl(
\sum \limits_{i = 1}^2 {\delta^{- \alpha_i}\iint\limits_{Q_T(s)}{
h^{\xi_i}}}\Biggr)^{1 + \nu_{i}},
\end{equation}
instead of (\ref{G:in11'}), and
\begin{equation}\label{G:sp1-2}
\Gamma_0(T) = C \bigl(T^{(1 - \tfrac{\theta_{1}\xi_1}{\alpha + 2}
)(1+ \nu_2)}T^{(1 - \tfrac{\theta_{2}\xi_2}{\alpha + 2} )\nu_1(1+
\nu_1)} (G(0))^{\nu_1}\bigr)^{\tfrac{1}{4(1+ \nu_1)(1+ \nu_2)}} ,
\end{equation}
$$
H(s) = C T^{(1 - \tfrac{\theta_{2}\xi_2}{\alpha + 2} )(1+
\nu_1)}T^{(1 - \tfrac{\theta_{1}\xi_1}{\alpha + 2} )\nu_2(1+
\nu_2)} (G_2(0))^{\nu_2},
$$
where
\begin{equation}\label{G:sp1-3}
G(0) = C (T^{(1 - \tfrac{\theta_{1}\xi_1}{\alpha + 2} )(1+
\nu_2)}(G_2(0))^{1+ \nu_1} + T^{(1 - \tfrac{\theta_{2}\xi_2}{\alpha
+ 2} )(1+ \nu_1)} (G_1(0))^{1+ \nu_2}).
\end{equation}
Now we need to estimate $G(0)$. With that end in view, we obtain
the following estimates:
\begin{equation}\label{G:est11}
G_i(0) \leqslant C_1\,(C_2 + C_3 T)^{\tfrac{\xi_i -1}{\alpha + 5}}
T^{1 - \tfrac{\xi_i -1}{\alpha + 5} }, \ \ i=1,2.
\end{equation}
where $1 < \xi_i < \alpha + 6$, and $C_i$ depends on $h_0(x)$
only. Really, applying the Gagliardo-Nirenberg interpolation
inequality (see Lemma~\ref{A.4}) in $\Omega = \mathbb{R}^1$ to a
function $v := h^{\tfrac{\alpha + 2}{2}}$ with $a = \tfrac{2\xi_i
}{\alpha + 2},\ b = \tfrac{2}{\alpha + 2},\ d = 2,\ i = 0,\ j = 2$,
and $\widetilde{\theta}_i = \tfrac{(\alpha + 2)(\xi_i -
1)}{\xi_i(\alpha +5)}$ under the condition $\xi_i > 1$, we deduce
that
\begin{equation}\label{G:est12}
\int\limits_{\mathbb{R}^1}{ h^{\xi_i}dx} \leqslant
c\,\|h_0\|_1^{\tfrac{2(3\xi_i + \alpha + 2)}{(\alpha + 2)(\alpha +
5)}} \Bigl( \int\limits_{\mathbb{R}^1} { (h^{\tfrac{\alpha +
2}{2}})^2_{xx} \,dx} \Bigr)^{\tfrac{\xi_i - 1}{\alpha + 5}}.
\end{equation}
Integrating (\ref{G:est12}) with respect to time and taking into
account the H\"{o}lder inequality ($\tfrac{\xi_i - 1}{\alpha + 5}
\leqslant 1 \Rightarrow \xi_i < \alpha + 6 \Rightarrow m \leqslant
n + 2$), we arrive at the following relations:
\begin{equation}\label{G:est13}
G_i(0)\leqslant c\,\|h_0\|_1^{\tfrac{2(3\xi_i + \alpha +
2)}{(\alpha + 2)(\alpha + 5)}} T^{1 - \tfrac{\xi_i - 1}{\alpha +
5}} \Bigl( \iint\limits_{Q_T} { (h^{\tfrac{\alpha + 2}{2}})^2_{xx}
\,dx} \Bigr)^{\tfrac{\xi_i - 1}{\alpha + 5}}, \ m \leqslant n + 2.
\end{equation}
From (\ref{G:est13}), due to (\ref{E:b9a}) and (\ref{F:globest}),
we find (\ref{G:est11}).

Inserting (\ref{G:est11}) into (\ref{G:sp1-3}), we obtain after
straightforward computations that
\begin{equation}\label{G:est14}
\Gamma_0(T) \leqslant \Gamma(T) = C\,T^{\tfrac{1}{n + 4}}
\end{equation}
for $T \leqslant T_{speed}$, where $\tfrac{n}{2} < m \leqslant n +
2$, and $T_{speed}$ finds from the condition $H(0) < 1$.

In the case of $m > n + 2$, we have (\ref{G:est13}) for $i = 1$
only, i.\,e.
\begin{equation}\label{G:est13-3}
G_1(0)\leqslant c\,\|h_0\|_1^{\tfrac{2(3\xi_1 + \alpha +
2)}{(\alpha + 2)(\alpha + 5)}} T^{1 - \tfrac{\xi_1 - 1}{\alpha +
5}} \Bigl( \iint\limits_{Q_T} { (h^{\tfrac{\alpha + 2}{2}})^2_{xx}
\,dx} \Bigr)^{\tfrac{\xi_1 - 1}{\alpha + 5}}.
\end{equation}
Hence
\begin{multline}\label{G:hren-0}
G_1(0)\leqslant C ( C +   C ( 1 -  (1 - \tfrac{k_1(2m -
n)\|h_0\|_{H^1}^{2m - n}}{2} T)^{\tfrac{m + 2}{2m-n}}) T^{1 -
\tfrac{\xi_1 - 1}{\alpha + 5}}
 \leqslant C_0 T^{1 -
\tfrac{\xi_1 - 1}{\alpha + 5}}.
\end{multline}
Next we need estimate the term $G_2(0)$. Really, by using the
inequality (\ref{Lady_ineq}),
we obtain
\begin{multline}\label{G:hren}
G_2(0)\leqslant b_2 \|h_0\|_1^{\tfrac{\xi_2+ 2}{3}}
\int\limits_{0}^T {\|h_x\|_2^{\tfrac{2}{3}(\xi_2 - 1)}dt} =
 b_2 \|h_0\|_1^{\tfrac{\xi_2+ 2}{3}} \|h_0\|_{H^1}^{\tfrac{2(\xi_2 -
1)}{3}} \times \\
\int\limits_{0}^T {(1 - \tfrac{k_1(2m - n)\|h_0\|_{H^1}^{2m -
n}}{2}  t)^{- \tfrac{2(\xi_2 - 1)}{3(2m - n)}}dt} =  b_2
\|h_0\|_1^{\tfrac{\xi_2+ 2}{3}} \|h_0\|_{H^1}^{- \tfrac{2m + n -
2(\alpha + 1))}{3}} \times \\
\tfrac{6}{k_1(2m + n - 2(\alpha + 1))}
 \bigl( 1 -  (1 - \tfrac{k_1(2m
- n)\|h_0\|_{H^1}^{2m - n}}{2} T)^{\tfrac{2m + n - 2(\alpha +
1)}{3(2m-n)}}\bigr) \leqslant C_0
\end{multline}
for all $T \leqslant T_{loc}$, where $C_0$ depends on $h_0(x)$.
Inserting (\ref{G:est11}) for $i = 1$ and (\ref{G:hren}) into
(\ref{G:sp1-3}), we obtain (\ref{G:est14}).

\subsection{Local energy estimate}

\begin{lemma}\label{G:lem2}
Let $n \in \bigl(\tfrac{1}{2}, 3 \bigr)$, and $m >
\tfrac{2(n-1)_+}{3}$, and $\beta
> \tfrac{1 -n}{3}$. Let $\zeta \in C^{2}(\bar{\Omega})$
such that $\text{supp}\,\zeta$ in $\Omega$, and
$(\zeta^6)^{\prime} = 0$ on $\partial\Omega$, and $\zeta(-a) =
\zeta(a)$. Then there exist constants $C_i \ (i=1,2,3)$ dependent
on $n,\,m,\,a_0$, and $a_1$, independent of $\Omega$ and
$\varepsilon$, such that for all $T \leqslant T_{loc}$
\begin{multline}\label{G:main3}
\int\limits_{\Omega} { \zeta^6 h^2_x(x,T) \,dx} +
\int\limits_{\Omega} { \zeta^4 h^{\beta+1}(T) \,dx} + C_1
\iint\limits_{Q_T} {\zeta^6 (h^{\tfrac{n + 2}{2}})_{xxx}^2 \,dx dt}
\leqslant \int\limits_{\Omega} {\zeta^6 h_{0}^2(x)\,dx} +\\
\int\limits_{\Omega} {\zeta^4 h_{0}^{\beta + 1}\,dx} +
C_2\iint\limits_{Q_T} {h^{n + 2} (\zeta_x^6 + \zeta^3
|\zeta_{xx}|^3)\,dx dt} +  C_3 \iint\limits_{Q_T} { h^{3m -2n +2}\zeta^6 \,dx dt} + \\
C_4 \iint\limits_{Q_T} { \{ h^{n + 2\beta} \zeta_x^2 +
\chi_{\{\zeta
> 0\}} h^{n + 3\beta -1} + h^{\tfrac{6m - n + 6\beta +
4}{5}}\zeta^{\tfrac{12}{5}} \zeta_x^{\tfrac{6}{5}} + h^{\tfrac{3m - n
+ 3\beta + 1}{2}}\zeta^{3} \}\,dx dt}.
\end{multline}
\end{lemma}

\begin{proof}[Sketch of Proof of Lemma~\ref{G:lem2}]
Let $\phi(x) = \zeta^6(x)$. Multiplying (\ref{D:1r'}) by $-
(\phi(x)h_x)_x$, and integrating on $Q_T$, yields
\begin{multline}\label{G:in7}
\tfrac{1}{2}\int\limits_{\Omega} { \phi(x)h^2_x(x,T) \,dx} -
\tfrac{1}{2} \int\limits_{\Omega} {\phi(x)
h_{0\varepsilon,x}^2(x)\,dx} = \\
- \iint\limits_{Q_T} {f_{\varepsilon}(h) (a_0 h_{xxx} + a_1
D''_{\varepsilon}(h)h_x) (\phi_{xx} h_x + 2 \phi_x h_{xx} + \phi
h_{xxx})\,dx dt} = \\
- \iint\limits_{Q_T} {f_{\varepsilon}(h) (a_0 h_{xxx} + a_1
D''_{\varepsilon}(h)h_x) \phi_{xx} h_x \,dx dt} - \\
2 \iint\limits_{Q_T} {f_{\varepsilon}(h) (a_0 h_{xxx} + a_1
D''_{\varepsilon}(h)h_x) \phi_x h_{xx}\,dx dt} -
\\
\iint\limits_{Q_T} {f_{\varepsilon}(h) (a_0 h_{xxx} + a_1
D''_{\varepsilon}(h)h_x) \phi h_{xxx}\,dx dt}  =: I_{1} + I_{2} +
I_3.
\end{multline}

We now bound the terms $I_1$, $I_2$ and $I_3$.  First,
\begin{multline}\label{G:in8}
I_{1}  = - a_0 \iint\limits_{Q_T} {\phi_{xx} f_{\varepsilon}(h)
h_{xxx}h_x\,dx dt} - a_1 \iint\limits_{Q_T} {\phi_{xx}
f_{\varepsilon}(h)D''_{\varepsilon}(h) h_x^2 \,dx dt} =\\
- 6 a_0 \iint\limits_{Q_T} {f_{\varepsilon}(h) h_{xxx}h_x
\zeta^4(5 \zeta_x^2 + \zeta\zeta_{xx})\,dx dt} - 6a_1
\iint\limits_{Q_T} {f_{\varepsilon}(h)D''_{\varepsilon}(h) h_x^2
\zeta^4(5 \zeta_x^2 + \zeta\zeta_{xx})\,dx dt}\leqslant \\
\epsilon_1 \iint\limits_{Q_T} {\zeta^6 \{f_{\varepsilon}(h)
h_{xxx}^2 + h^{n-4}h_x^6\}\,dx dt} + C(\epsilon_1)
\iint\limits_{Q_T} {h^{n + 2}
(\zeta_x^6 + \zeta^3 |\zeta_{xx}|^3)\,dx dt} + \\
C(\epsilon_1) \iint\limits_{Q_T} {h^{3m -2n +2}\zeta^6\,dx dt},
\end{multline}
\begin{multline}\label{G:in9}
I_{2}  = - 2a_0 \iint\limits_{Q_T} {\phi_{x} f_{\varepsilon}(h)
h_{xxx}h_{xx}\,dx dt} - 2 a_1 \iint\limits_{Q_T} {\phi_{x}
f_{\varepsilon}(h)D''_{\varepsilon}(h)h_{xx} h_x \,dx dt} =\\
- 12a_0 \iint\limits_{Q_T} { f_{\varepsilon}(h) h_{xxx}h_{xx}
\zeta^5 \zeta_x \,dx dt} - 12 a_1 \iint\limits_{Q_T} {
f_{\varepsilon}(h)D''_{\varepsilon}(h)h_{xx} h_x \zeta^5 \zeta_x
\,dx dt} \leqslant \\
\epsilon_2 \iint\limits_{Q_T} {\zeta^6 \{ f_{\varepsilon}(h)
h_{xxx}^2 + h^{n - 2} h_x^2 h^2_{xx} + h^{n - 1} |h_{xx}|^3\}\,dx
dt} + \\
C(\epsilon_2) \iint\limits_{Q_T} { h^{n + 2} \zeta_x^6 \,dx dt} +
C(\epsilon_2)\iint\limits_{Q_T} { h^{3m - 2n + 2}\zeta^6 \,dx dt}
,
\end{multline}
\begin{multline}\label{G:in10}
I_{3}  = - a_0 \iint\limits_{Q_T} {\phi f_{\varepsilon}(h)
h_{xxx}^2 \,dx dt} - a_1 \iint\limits_{Q_T} {\phi
f_{\varepsilon}(h)D''_{\varepsilon}(h) h_{xxx} h_x \,dx dt}
\leqslant \\
- a_0 \iint\limits_{Q_T} {\zeta^6 f_{\varepsilon}(h) h_{xxx}^2
\,dx dt} +  \epsilon_3 \iint\limits_{Q_T} {\zeta^6 (
f_{\varepsilon}(h)h^2_{xxx} + h^{n- 4}h_x^6) \,dx dt} +\\
C(\epsilon_3) \iint\limits_{Q_T} { h^{3m -2n +2}\zeta^6 \,dx dt}.
\end{multline}

Now, multiplying (\ref{D:1r'}) by $\zeta^4 (h + \gamma)^{\beta}$,
$\beta
> \tfrac{1 -n}{3}, \ \gamma > 0$ and integrating on $Q_T$, using
the Young's inequality (\ref{Young}), letting $\gamma \to 0$, we
obtain the following estimate
\begin{multline}\label{G:in11}
\int\limits_{\Omega} { \zeta^4 h^{\beta+1}(T) \,dx} \leqslant
\int\limits_{\Omega} {\zeta^4 h_{0\varepsilon}^{\beta + 1}\,dx} +
\epsilon_4 \iint\limits_{Q_T} {\zeta^6 \{ f_{\varepsilon}(h)
h_{xxx}^2 + h^{n
- 4} h_x^6 \}\,dx dt} + \\
C(\epsilon_4) \iint\limits_{Q_T} { \{ h^{n + 2\beta} \zeta_x^2 +
\chi_{\{\zeta
> 0\}} h^{n + 3\beta -1} + h^{\tfrac{6m - n + 6\beta +
4}{5}}\zeta^{\tfrac{12}{5}} \zeta_x^{\tfrac{6}{5}} + h^{\tfrac{3m - n
+ 3\beta + 1}{2}}\zeta^{3} \}\,dx dt},
\end{multline}
where $\beta > \max \{\tfrac{1 -n}{3}, - m + \tfrac{n - 1}{3}\} =
\tfrac{1 -n}{3}$ as $m > \tfrac{2(n-1)_+}{3}$.

If we now add inequalities (\ref{G:in7}) and (\ref{G:in11}), in
view of (\ref{G:in8})--(\ref{G:in10}), then, applying
Lemma~\ref{A.3}, choosing $\epsilon_i
>0$, and letting $\varepsilon \to 0$, we obtain (\ref{G:main3}).
\end{proof}

\subsection{Proof of Theorem~\ref{C:Th4} for the case $\tfrac{1}{2} <
n < 3$}

Let $\eta_{s,\delta}(x)$ denote by (\ref{G:test2}). Setting
$\zeta^6(x) = \eta_{s,\delta}(x)$ into (\ref{G:main3}), after
simple transformations, we obtain
\begin{multline}\label{G:main4}
\int\limits_{\Omega(s + \delta)} {h^2_x(x,T) \,dx}+
\int\limits_{\Omega(s + \delta)} {h^{\beta+1}(T) \,dx} + C
\iint\limits_{Q_T(s+ \delta)} {(h^{\tfrac{n + 2}{2}})_{xxx}^2 \,dx
dt} \leqslant \\
\tfrac{C}{\delta^6}\iint\limits_{Q_T(s)} {h^{n + 2} \,dx dt} +
\tfrac{C}{\delta^2}\iint\limits_{Q_T(s)} {h^{n + 2\beta} \,dx dt}
+ \tfrac{C}{\delta^{\tfrac{6}{5}}}\iint\limits_{Q_T(s)}
{h^{\tfrac{6m - n + 6\beta + 4}{5}} \,dx dt} + \\
C \iint\limits_{Q_T(s)} { \{ h^{3m -2n +2} + h^{n + 3\beta -1} +
h^{\tfrac{3m - n + 3\beta + 1}{2}}\} \,dx dt}=: C \sum \limits_{i =
1}^6 {\delta^{- \alpha_i}\iint\limits_{Q_T(s)}{ h^{\xi_i}}}
\end{multline}
for all for all $s > 0,\ \delta >0 : r_0 + s + \delta < a$. We
apply the Gagliardo-Nirenberg interpolation inequality
(Lemma~\ref{A.4}) in the region $\Omega(s+\delta)$ to a function
$v := h^{\tfrac{n+2}{2}}$ with $a = \tfrac{2\xi_i }{n + 2},\ b =
\tfrac{2(\beta + 1)}{n + 2},\ d = 2,\ i = 0,\ j = 3$, and $\theta_i
= \tfrac{(n + 2)(\xi_i - \beta - 1)}{\xi_i(n + 5\beta + 7)}$ under
the conditions:
\begin{equation}\label{G:con1}
\beta < \xi_i - 1  \text{ for } i = \overline{1,6}.
\end{equation}
Integrating the resulted inequalities with respect to time and
taking into account (\ref{G:main4}), we arrive at the following
relations:
\begin{equation}\label{G:in12}
\iint\limits_{Q_T(s +\delta)}{ h^{\xi_i}} \leqslant C\, T^{1 -
\tfrac{\theta_{i}\xi_i}{n + 2} } \Biggl( \sum \limits_{i = 1}^6
{\delta^{- \alpha_i}\iint\limits_{Q_T(s)}{ h^{\xi_i}}}\Biggr)^{1 +
\nu_{i}} \!\!\!\!\! + C\, T \Biggl( \sum \limits_{i = 1}^5
{\delta^{- \alpha_i}\iint\limits_{Q_T(s)}{
h^{\xi_i}}}\Biggr)^{\tfrac{\xi_i} {\beta + 1}}\!\!\!\!,
\end{equation}
where $\nu_{i} = \tfrac{6(\xi_i - \beta - 1)}{n + 5\beta + 7}$.
These inequalities are true provided that
\begin{equation}\label{G:con2}
\tfrac{\theta_{i}\xi_i}{n + 2} < 1 \Leftrightarrow \beta >
\tfrac{\xi_i - n - 8}{6} \text{ for } i = \overline{1,6}.
\end{equation}
Simple calculations show that inequalities (\ref{G:con1}) and
(\ref{G:con2}) hold with some $\beta
> \tfrac{1 - n}{3}$ if and only if
$$
\tfrac{1}{2} < n < 3, \ \ m > \tfrac{n}{2}.
$$
Since all integrals on the right-hand sides of (\ref{G:in12})
vanish as $T \to 0$, the finite speed of propagations  follows
from (\ref{G:in12}) by applying Lemma~\ref{A.5} with $s_1 = 0$ and
sufficiently small $T$. Hence,
\begin{equation}\label{G:sp}
\textnormal{supp}\,h(T,.) \subset (-r_0 - \Gamma (T), r_0 + \Gamma
(T) ) \Subset \Omega \ \text{for all } T: 0 \leqslant T \leqslant
T_{speed}.
\end{equation}

\section{Finite time blow up}


\begin{lemma}\label{F:lem0}
Let $0 < n < 2,\, m \geqslant  \max \{n + 2, 4 - n\}$. Then the weak
solution $h$ from Theorem~\ref{C:Th1} satisfies the
second-moment inequality:
\begin{multline}\label{F:d9-01}
e^{-\widetilde{B}(T)}
\int\limits_{\Omega} {x^2 \widetilde{G}_0 (h(x,T))\,dx} \leqslant
 \int \limits_{\Omega} {x^2 \widetilde{G}_{0}
(h_0)\,dx} + \\
\int \limits_{0}^T { \biggl(k_1 \mathcal{E}_{0}(0) + W'(t) + k_2
\int\limits_{\Omega} {x^2 h^2_{xx} \,dx} \biggr) e^{-\widetilde{B}(t)}
dt}
\end{multline}
for all $T \in [0, T_{loc}]$, where $k_1 = 2(4 -n),\ k_2 =
\tfrac{3a_0(n-1)}{2}$. Here
$$
\widetilde{G}_0(z) = \tfrac{1}{ 2 -n }z^{ 2 -n}, \
\widetilde{B}(T):= \tfrac{a_1^2|1 - n|( 2- n)}{2a_0(m -n+1)^2}
\int\limits_0^T {\|h(.,\tau)\|_{L^{\infty}(\Omega)}^{2m
-n}\,d\tau},
$$
\begin{equation}\label{F:W}
W(T) := \int\limits_{0}^T { x \bigl( 2a_0 h  h_{xx} - a_0(2  - n)
h_{x}^2 + 2a_1 m D_{0}(h)\bigr) \biggl |_{\partial\Omega}\,dt}.
\end{equation}
Moreover, if $h(\cdot,t)$ has a compact support in $\Omega$ then
$W'(t) \equiv 0$.
\end{lemma}

The proof of Lemma \ref{F:lem0} is given in Appendix \ref{Bapp}.

\begin{proof}[Proof of Theorem~\ref{C:Th3}]
Assume $0 < n < 2$ and $m$ satisfies the hypotheses of
Theorem \ref{C:Th4}. By Theorems~\ref{C:Th2}, \ref{C:Th4}, there is a
compactly supported strong solution $h$ for $t \in [0, T_0]$.
Hence, we can assume that $W'(t) = 0 $ in Lemma~\ref{F:lem0}.

First, we construct a sequence of times $T_0 < T_1 <...$ and
extend the strong solution $h$ from the time interval $[0, T_i]$ to
the time interval $[0, T_{i+1}]$. Taking it as an initial datum,
we obtain a time interval of existence
$$
T_1 - T_0 = \tfrac{1}{2 d_7(\gamma_3 - 1)}\biggl(
\int\limits_{\Omega} { \{\tfrac{1}{2}h_{x}^2(x, T_0) +
G_{0}^{(\alpha)}(h(x,T_0))\}\,dx}\biggr)^{-(\gamma_3 - 1)},
$$
by (\ref{Tloc_alpha_is}).
Applying Theorem
\ref{C:Th4} to the time interval $[T_0, T_1]$, we have a strong
solution with compact support that satisfies all a priori
estimates with the time interval $[0, T_0]$ replaced by $[0,
T_1]$.   In this way, we construct a nonnegative, compactly
supported, strong solution on $\R^1 \times [0,T^*)$ where
$$
T^* = \lim_{i \to \infty} T_i.
$$
If $T^* < \infty$ then
$$
\tfrac{1}{2e_7(\gamma_3 - 1)}\biggl( \int\limits_{\Omega} {
\{\tfrac{1}{2}h_{x}^2(x, T_i) +
G_{0}^{(\alpha)}(h(x,T_i))\}\,dx}\biggr)^{1- \gamma_3} = T_{i + 1}
- T_i \to 0.
$$
Hence, the $H^1$ norms at times $T_i$ must blow up.
And, due to (\ref{C:a2-new2}), the $L^{\infty}$-norm of the
solution at times $T_i$ must also blow up.
Therefore, to finish the proof it suffices to prove that $T^* < \infty$.

Let
\begin{equation}\label{F:defg}
V(0) =
\Int {x^2
\widetilde{G}_0 (h_0)\,dx}
\quad \mbox{and} \quad
V(T_i) := e^{- \widetilde{B}(T_i)}
\Int {x^2
\widetilde{G}_0 (h(T_{i}))\,dx}.
\end{equation}
Using that $k_1 > 0$ and $\mathcal{E}_0(T_i) \leqslant \mathcal{E}_0(0)$, we
apply the inequality (\ref{F:d9-01}) iteratively to find
\begin{equation} \label{general_2nd_moment}
V(T_i) \leqslant V(0) + k_1 \mathcal{E}_0(0) \int\limits_0^{T_i} e^{-\widetilde{B}(t)} \; dt
+ k_2 \int\limits_0^{T_i} \Int x^2 h_{xx}^2(x,t) \; dx e^{-\widetilde{B}(t)} \; dt
\end{equation}

\textbf{Case $n=1$:} \cite{BP2}  In this case, $\widetilde{B}(t) = 0$ and $k_2 = 0$.  Hence
(\ref{general_2nd_moment}) becomes
$$
V(T_i) \leqslant V(0) + k_1 \mathcal{E}_0(0) \, T_i.
$$
If $T^* = \infty$ then the right-hand side will become negative, contradicting that
the left-hand side is nonnegative.  Hence $T^* < \infty$.

\textbf{Case $0< n < 1$:}
In this case, $k_2 < 0$ and (\ref{general_2nd_moment}) implies
\begin{equation} \label{F:sme3n}
V(T_i) \leqslant V(0) + k_1 \mathcal{E}_0(0) \int\limits_0^{T_i} e^{-\widetilde{B}(t)} \; dt
\end{equation}
We wish to argue that if $T^* = \infty$ then the right-hand side of (\ref{F:sme3n})
must become negative, leading to a contradition.

Let $g(t) = \int \limits_{0}^{t}
{e^{-\widetilde{B}(s)} ds} $.
Using (\ref{H1_control}), we have $\|h(x,t)\|_{\infty} \leqslant K_1
(T_i - t)^{-\tfrac{1}{2(2m - n -1)}} $ for all $t < T_i$.
$$
\widetilde{B}(s) \leqslant \tfrac{a_1^2( 1 - n)( 2- n)}{2a_0(m
-n+1)^2} K_1^{2m - n} \int\limits_0^s {(T_i - s)^{-\tfrac{2m -
n}{2(2m - n
-1)}} \,d\tau} =
K_2 \bigl( T_i^{\tfrac{2m-n-2}{2(2m - n -1)}} - (T_i -
s)^{\tfrac{2m-n-2}{2(2m - n -1)}} \bigr),
$$
whence
\begin{multline}\label{F:gT2}
g(T_i) \geqslant e^{- K_2 T_i^{\tfrac{2m-n-2}{2(2m - n -1)}}} \int
\limits_0^{T_i}{ e^{K_2 (T_i - s)^{\tfrac{2m-n-2}{2(2m - n -1)}}}
ds} = \\
\tfrac{\int \limits_0^{T_i}{ e^{K_2 s^{\tfrac{2m-n-2}{2(2m - n
-1)}}} ds}}{ e^{K_2 T_i^{\tfrac{2m-n-2}{2(2m - n -1)}}}} \sim
T_i^{\tfrac{2m-n}{2(2m - n -1)}} \text{ as } T_i \to + \infty.
\end{multline}
Since $ \mathcal{E}_{0} (0) < 0$, if $T_i \to + \infty$ then for
large times the right-hand side of (\ref{F:sme3n}) would be
negative: an impossibility. Therefore $\mathop {\lim} \limits_{i
\to +\infty} T_i = T^* < \infty$.

\textbf{Case $1 < n < 2$:}
In this case, we write inequality (\ref{general_2nd_moment})
in the form:
\begin{equation} \label{BU_1n2}
V(T_i) \leqslant V(0) + k_1 \mathcal{E}_0(0) \, g(T_i)
+ k_2 \, f(T_i)
\end{equation}
where
$$
f(T_i) = \int\limits_0^{T_i}  e^{-\widetilde{B}(t)} \Int x^2 h_{xx}^2(x,t) \; dx \; dt.
$$
Assume $T^* = \infty$ and
that $f(T_i)$ grows slower than $g(T_i)$ as $T_i \to \infty$.  Then
the right--hand side of
(\ref{BU_1n2}) will become negative in finite time, which is an impossibility.
It therefore suffices to show that $f(T_i)$ grows more slowly than $g(T_i)$.

Due to (\ref{D:aa2}) and
(\ref{H1_control}) as $\varepsilon \to 0$, we have
\begin{multline}\label{F:aa2-2}
\tfrac{a_0}{2} \int\limits_0^{T_i} \Int {h^2_{xx}\,dx dt} \leqslant
\Int {G_{0}(h_{0})\,dx} + C \int\limits_0^{T_i} {
\| h_x \|_{2} ^{2(m - n + 1)} dt} \leqslant \\
\Int {G_{0}(h_{0})\,dx} + C \int\limits_0^{T_i} {
(T_i - t)^{- \tfrac{m - n + 1}{2m - n - 1}} dt} =
\Int {G_{0}(h_{0})\,dx} +  \tfrac{C(2m - n - 1)}{m
- 2} T_i ^{\tfrac{m - 2}{2m - n -1}}.
\end{multline}
In view of (\ref{G:est14}), we know that $|x| \leqslant C(h_0)$
for all $t \leqslant T_i$. Hence, due to (\ref{F:aa2-2}), we
deduce
$$
f(T_i)
\leqslant
\Int
 {x^2
h^2_{xx}\,dx dt}
\leqslant C
\int\limits_0^{T_i} \Int {h^2_{xx}\,dx
dt} \leqslant C
( 1 + T_i ^{\tfrac{m - 2}{2m - n -1}}).
$$
Comparing with the exponent in (\ref{F:gT2}), we see
that $f(T_i)$ grows more slowly than $g(T_i)$, as desired.

\end{proof}

\appendix

\section{Proofs of A Priori Estimates}
\label{A_priori_proofs}

The first observation is that the periodic boundary conditions
imply that classical solutions of equation (\ref{D:1r'}) conserve
mass:
\begin{equation}\label{D:mass}
\int\limits_{\Omega} {h_{\delta \varepsilon}(x,t)\,dx} =
\int\limits_{\Omega} {h_{0,\delta \varepsilon}(x) \,dx} =
M_{\delta \varepsilon } < \infty \text{ for all } t
> 0.
\end{equation}
Further, (\ref{D:inreg}) implies $M_{\delta \varepsilon} \to M =
\int h_0$ as $\varepsilon, \delta \to 0$. The initial data in this
article have $M > 0$, hence $M_{\delta \varepsilon} > 0$ for
$\delta$ and $\varepsilon$ sufficiently small.

Also, we will relate the $L^p$ norm of $h$ to the $L^p$ norm of
its zero-mean part as follows:
$$
|h(x)| \leqslant \left| h(x)-\tfrac{M}{|\Omega|} \right| +
\tfrac{M}{|\Omega|} \Longrightarrow \| h \|_p^p \leqslant 2^{p-1} \; \|
v \|_p^p + \left(\tfrac{2}{  |\Omega| } \right)^{p-1} \; M^{p}
$$
where $v := h-M/|\Omega|$ and we have assumed that $M \geqslant
0$. We will use the Poincar\'{e} inequality which holds for  any
zero-mean function in $H^1(\Omega)$
\begin{equation} \label{Poincare}
\| v \|_p^p \leqslant b_1 \|v_x \|_p^p \qquad 1 \leqslant p < \infty
\end{equation}
with $b_1 = |\Omega|^p$.

Also used will be an interpolation inequality \cite[Theorem 2.2, p.
62]{Lady} for functions of zero mean in $H^1(\Omega)$:
\begin{equation} \label{Lady_ineq}
\| v \|_p^p \leqslant b_2 \, \| v_x \|_2^{ap} \; \| v \|_r^{(1-a)p}
\end{equation}
where $r \geqslant 1$, $p \geqslant r$,
$$
a = \tfrac{1/r - 1/p}{1/r + 1/2}, \qquad b_2 = \left(1 + r/2
\right)^{a p}.
$$
It follows that for any zero-mean function $v$ in $H^1(\Omega)$
\begin{equation} \label{LpH1}
\| v \|_p^p  \leqslant b_3 \| v_x \|_2^p, \quad \Longrightarrow \quad
\| h \|_p^p \leqslant b_4 \| h_x \|_2^p + b_5 M_{\delta \varepsilon}^p
\end{equation}
where
$$
b_3 =
\begin{cases}
b_1 \; | \Omega |^{(2-p)/p} & \mbox{if} \quad 1 \leqslant p \leqslant 2 \\
b_1^{(p+2)/2} \; b_2 & \mbox{if} \quad  2 < p < \infty
\end{cases}, \quad
b_4 = 2^{p-1} \, b_3, \quad b_5 = \left( \tfrac{2}{|\Omega|}
\right)^{p-1}
$$
To see that (\ref{LpH1}) holds, consider two cases.  If $1
\leqslant p < 2$, then by (\ref{Poincare}), $\| v \|_p$ is
controlled by $\| v_x \|_p$.  By the H\"older inequality, $\| v_x
\|_p$ is then controlled by $\| v_x \|_2$.  If $p > 2$ then by
(\ref{Lady_ineq}), $\|v \|_p$ is controlled by $\|v_x\|_2^a \|v
\|_2^{1-a}$ where $a = 1/2 - 1/p$.  By  the Poincar\'e inequality,
$\| v \|_2^{1-a}$ is controlled by $\| v_x \|_2^{1-a}$.

If $0 < p < 1$ then, instead of (\ref{LpH1}), we obtain
\begin{equation} \label{LpH1-2}
\| h \|_p^p \leqslant \widetilde{b}_4 \| h_x \|_2^p + \widetilde{b}_5
M_{\delta \varepsilon}^p
\end{equation}
where
$$
\widetilde{b}_4 = |\Omega|^{1-\tfrac{p}{2}} \, b_4^{\tfrac{p}{2}}, \quad
\widetilde{b}_5 = |\Omega|^{1-\tfrac{p}{2}} \, b_5^{\tfrac{p}{2}}.
$$

The Cauchy inequality $ab \leqslant \epsilon a^2 + b^2/(4
\epsilon)$ with $\epsilon > 0$ will be used often as will Young's
inequality
\begin{equation} \label{Young}
a b \leqslant \epsilon a^p + \tfrac{b^q}{q \, (\epsilon p)^{q/p}},
\qquad \tfrac{1}{p} + \tfrac{1}{q} = 1, \; \epsilon > 0.
\end{equation}

\begin{proof}[Sketch of Proof of Lemma~\ref{MainAE}]
In the following, we denote the classical solution $h_{\delta
\varepsilon}$ by $h$ whenever there is no chance of confusion.

To prove the bound (\ref{D:a13''}) one starts by multiplying
(\ref{D:1r'}) by $- h_{xx}$, integrating over $Q_T$, and using the
periodic boundary conditions (\ref{D:2r'}) yields
\begin{align}
\label{D:a1} & \tfrac{1}{2} \int\limits_{\Omega} {h_x^2(x,T) \,dx}
+ a_0
\iint\limits_{Q_T} {f_{\delta \varepsilon}(h) h^2_{xxx} \,dx dt} \\
\notag & \hspace{.2in} = \tfrac{1}{2}\int\limits_{\Omega}
{{h_{0,\delta \varepsilon,x}}^2(x) \,dx} - a_1 \iint\limits_{Q_T}
{f_{\delta \varepsilon}(h) D''_{\varepsilon}(h)h_x h_{xxx}\,dx
dt}.
\end{align}
The Cauchy inequality is used to bound some terms on the
right-hand of (\ref{D:a1}):
\begin{multline}\label{D:a2}
a_1\iint\limits_{Q_T} {f_{\delta \varepsilon}(h)
D''_{\varepsilon}(h) h_x h_{xxx}\,dx dt} \leqslant
\tfrac{a_0}{2}\iint\limits_{Q_T} {f_{\delta \varepsilon}(h)
h^2_{xxx} \,dx dt} + \\
\tfrac{a_1^2}{2a_0}\iint\limits_{Q_T} {f_{\delta
\varepsilon}(h)(D''_{\varepsilon}(h))^2 h_{x}^2 \,dx dt} .
\end{multline}
Using (\ref{D:a2}) in (\ref{D:a1}) yields
\begin{multline}\label{D:a5}
\tfrac{1}{2}\int\limits_{\Omega} {h_x^2(x,T) \,dx} +
\tfrac{a_0}{2} \iint\limits_{Q_T} {f_{\delta \varepsilon}(h)
h^2_{xxx}\,dx dt}  \leqslant
 \tfrac{1}{2}\int\limits_{\Omega} {
 {h_{0,\delta \varepsilon,x}}^2\,dx} + \\
 \tfrac{a_1^2}{2a_0}\iint\limits_{Q_T} {f_{\delta
\varepsilon}(h)(D''_{\varepsilon}(h))^2 h_{x}^2 \,dx dt} \leqslant
\tfrac{1}{2} \int\limits_{\Omega} { {h_{0,\delta
\varepsilon,x}}^2\,dx} + \\
\tfrac{a_1^2}{2a_0}\iint\limits_{Q_T} {|h|^{2m-n} h_{x}^2 \,dx dt}
+ \tfrac{\delta}{\varepsilon^2}
\tfrac{a_1^2}{2a_0}\iint\limits_{Q_T} {h_{x}^2 \,dx dt} .
\end{multline}
Above, we used the bounds $f_{\delta \varepsilon}(z) \leqslant |z|^n +
\delta$, $ D''_{\varepsilon}(z) \leqslant |z|^{m-n}$, and $
D''_{\varepsilon}(z) \leqslant \varepsilon^{-1}$.
By the Cauchy inequality, bound (\ref{LpH1}), (\ref{Lady_ineq})
and bound (\ref{LpH1-2}),
\begin{align}
&\notag  \iint\limits_{Q_T} |h|^{2m-n} h_x^2 \; dx dt \leqslant
\tfrac{1}{2} \iint\limits_{Q_T} h^{2(2m-n)} \; dx dt +
\tfrac{1}{2}  \iint\limits_{Q_T} h_x^4 \; dx dt \\
& \notag \leqslant \tfrac{b_4}{2} \int\limits_0^T  \left(
\int\limits_\Omega h_x^2 \; dx \right)^{2m-n} \; dt +
\tfrac{b_5}{2} \; M_{\delta \varepsilon}^{2(2m-n)} \; T +
\tfrac{b_2}{2} \int\limits_0^T \| h_{xx}(\cdot,t) \|_2 \; \| h_x(\cdot,t) \|_2^3 \; dt \\
& \label{D:a9} \hspace{.4in} \leqslant \tfrac{1}{2}
\iint\limits_{Q_T} h_{xx}^2 \; dx dt + \tfrac{b_2^2}{8}
\int\limits_0^T \left( \int\limits_\Omega h_x^2 \; dx \right)^3 \;
dt + \tfrac{b_4}{2} \int\limits_0^T \left( \int\limits_\Omega
h_x^2 \; dx \right)^{2m-n} + c_1 \; T,
\end{align}
where $c_1 = M_{\delta \varepsilon}^{2(2m-n)} \; b_5/2$,  $m
\geqslant \tfrac{n}{2} $.
From (\ref{D:a5}), due to (\ref{D:a9}), we arrive at
\begin{align}
& \notag \tfrac{1}{2} \int\limits_{\Omega} {h_x^2(x,T) \,dx} +
\tfrac{a_0}{2}
\iint\limits_{Q_T} {f_{\delta \varepsilon}(h) h^2_{xxx}\,dx dt} \\
& \notag \hspace{.2in} \leqslant \tfrac{1}{2}\int\limits_{\Omega}
{{h_{0,\delta \varepsilon,x}}^2 \,dx} + c_2 \iint\limits_{Q_T}
{h^2_{xx}\,dx dt}
+ c_3 \int\limits_0^T \left( \int\limits_\Omega h_x^2 \; dx \right)^3 \; dt \\
& \notag \hspace{.6in} + c_4  \int\limits_0^T \left(
\int\limits_\Omega h_x^2 \; dx \right)^{2m-n} \; dt + c_5
\iint\limits_{Q_T} {h_{x}^2 \,dx dt}  + c_6 \; T \\
& \hspace{.2in} \leqslant \label{D:a11}
\tfrac{1}{2}\int\limits_{\Omega} {{h_{0,\delta \varepsilon,x}}^2
\,dx} + c_2 \iint\limits_{Q_T} {h^2_{xx}\,dx dt} + c_7
\int\limits_0^T \max\left\{ 1 , \left( \int\limits_\Omega h_x^2 \;
dx \right)^{\gamma_1} \right\} \; dt
\end{align}
where $\gamma_1 = \max \{3, 2m - n \}$, $m \geqslant \tfrac{n}{2}
$,
\begin{align*}
& c_2 = \tfrac{a_1^2}{4 a_0}, \quad c_3 = \tfrac{a_1^2 b_2^2}{16
a_0}, \quad
c_4 = \tfrac{a_1^2 b_4}{4 a_0}, \  c_5= \tfrac{a_1^2}{2 a_0} \tfrac{\delta}{\varepsilon^2} \\
& c_6 = \tfrac{a_1^2}{2a_0}c_1, \quad c_7 = c_3 + c_4 + c_5 + c_6.
\end{align*}

Now, multiplying (\ref{D:1r'}) by $G'_{\delta \varepsilon}(h)$,
integrating over $Q_T$, and using the periodic boundary conditions
(\ref{D:2r'}), we obtain
\begin{multline}\label{D:aa1}
\int\limits_{\Omega} {G_{\delta \varepsilon}(h(x,T))\,dx} + a_0
\iint\limits_{Q_T} {h^2_{xx}\,dx dt} = \int\limits_{\Omega}
{G_{\delta \varepsilon}(h_{0, \delta \varepsilon})\,dx} +
a_1 \iint\limits_{Q_T} {D''_{\varepsilon}(h)h^2_{x}\,dx dt} \\
\leqslant \int\limits_{\Omega} {G_{\delta \varepsilon}(h_{0,
\delta \varepsilon})\,dx} + \tfrac{a_1}{\varepsilon}
\iint\limits_{Q_T} {h^2_{x}\,dx dt}
\end{multline}
for limiting process on $ \delta \to 0$, and
\begin{multline}\label{D:aa1-2}
\int\limits_{\Omega} {G_{\varepsilon}(h(x,T))\,dx} + a_0
\iint\limits_{Q_T} {h^2_{xx}\,dx dt} \\
\leqslant \int\limits_{\Omega} {G_{\varepsilon}(h_{0,
\varepsilon})\,dx} + a_1 \iint\limits_{Q_T} {h^{m-n}h^2_{x}\,dx
dt}, \ \ h_{\varepsilon} > 0
\end{multline}
for limiting process on $ \varepsilon \to 0$. In the case of
(\ref{D:aa1-2}), we estimate the right-hand using the following
equality
\begin{equation}\label{D:simp}
\int\limits_{\Omega} {v^{a}v^2_{x}\,dx} =
\tfrac{1}{a+1}\int\limits_{\Omega} {v^{a+ 1}v_{xx}\,dx}, \ \ v >
0, \ a
> -1.
\end{equation}
Really, for $h_{\varepsilon} > 0$ we deduce by the Cauchy
inequality and bound (\ref{LpH1}) that
\begin{multline}\label{D:aa1-10}
a_1 \iint\limits_{Q_T} {h^{m-n}h^2_{x}\,dx dt} = - \tfrac{a_1}{m -
n + 1} \iint\limits_{Q_T} {h^{m-n+1}h_{xx}\,dx dt} \\
\leqslant \tfrac{a_0}{2}\iint\limits_{Q_T} {h^2_{xx}\,dx dt} +
\tfrac{a_1^2}{2a_0(m -n + 1)^2} \iint\limits_{Q_T} {h^{2(m-n+
1)}\,dx dt} \\
\leqslant \tfrac{a_0}{2}\iint\limits_{Q_T} {h^2_{xx}\,dx dt} +
c_8\int\limits_0^T \Bigl( \int\limits_\Omega h_x^2 \, dx \Bigr)^{m
- n + 1}\,dt + c_9 T,
\end{multline}
where $c_8 = \tfrac{a_1^2}{2a_0(m -n +1)^2}b_4 $, $c_9 =
\tfrac{a_1^2}{2a_0(m -n +1)^2}b_5 M_{\varepsilon}^{2(m - n + 1)}$,
and $m \geqslant n - 1$. Thus, from (\ref{D:aa1}) due to
(\ref{D:aa1-10}), we deduce
\begin{multline}\label{D:aa2}
\int\limits_{\Omega} {G_{\varepsilon}(h(x,T))\,dx} +
\tfrac{a_0}{2} \iint\limits_{Q_T} {h^2_{xx}\,dx dt} \leqslant
\int\limits_{\Omega} {G_{\varepsilon}(h_{0, \varepsilon})\,dx}  \\
+ c_8\int\limits_0^T \Bigl( \int\limits_\Omega h_x^2 \, dx
\Bigr)^{m - n + 1}dt + c_9 T \leqslant \int\limits_{\Omega}
{G_{\varepsilon}(h_{0, \varepsilon})\,dx} \\
+ c_{10}
\int\limits_0^T \max\left\{ 1 , \left( \int\limits_\Omega h_x^2 \;
dx \right)^{\gamma_2} \right\} \; dt,
\end{multline}
where $\gamma_2 = \max \{3, m - n +1\}$, $m \geqslant n - 1$, $
c_{10} = c_8 + c_9$.

Further, from (\ref{D:a11}) and (\ref{D:aa1}) we find for limiting
process on $ \delta \to 0$
\begin{align} \notag
& \int\limits_{\Omega} {h_x^2 \,dx} + \tfrac{2 c_2}{a_0}
\int\limits_\Omega G_{\delta \varepsilon}(h(x,T)) \; dx + a_0
\iint\limits_{Q_T}
{f_{\delta \varepsilon}(h) h^2_{xxx}  \,dx dt} \\
& \hspace{.1in} \notag \leqslant \int\limits_\Omega {h_{0,\delta
\varepsilon,x}}^2 \; dx + \tfrac{2 c_2}{a_0} \left(
\int\limits_\Omega G_{\delta \varepsilon}(h(x,T)) \; dx +
\tfrac{a_0}{2} \iint\limits_{Q_T} h_{xx}^2 \; dx dt \right) \\
& \hspace{.1in} \notag + 2 c_7 \int\limits_0^T \max\left\{ 1,
\left( \int\limits_\Omega h_x^2(x,t) \; dx \right)^{\gamma_1}
\right\} \; dt
\leqslant \int\limits_\Omega {h_{0,\delta \varepsilon,x}}^2 \; dx \\
& \hspace{.1in} \notag + \tfrac{2 c_2}{a_0} \left(
\int\limits_\Omega G_{\delta \varepsilon}(h_{0,\delta
\varepsilon}) \; dx + \tfrac{a_1}{\varepsilon} \iint\limits_{Q_T}
{h^2_{x}\,dx dt} \right) \\
& \hspace{.1in} \label{D:aa3}+
 2 c_7 \int\limits_0^T \max\left\{ 1, \left( \int\limits_\Omega h_x^2(x,t) \; dx \right)^{\gamma_1} \right\} \; dt
\leqslant  \int\limits_\Omega {h_{0,\delta \varepsilon,x}}^2 \; dx \\
  & \notag  \hspace{.1in} + \tfrac{2 c_2}{a_0}
\int\limits_\Omega G_{\delta \varepsilon}(h_{0,\delta
\varepsilon}) \; dx + c_{11} \int\limits_0^T \max\left\{ 1, \left(
\int\limits_\Omega h_x^2(x,t) \; dx \right)^{\gamma_1} \right\} \;
dt
\end{align}
where $c_{11} = \tfrac{2 c_2 a_1}{\varepsilon a_0}  + 2 c_7$,
$\gamma_1 = \max \{3, 2m - n \}$, $m \geqslant \tfrac{n}{2} $.
Similarly, from (\ref{D:a11}) and (\ref{D:aa2}) we find for
limiting process on $ \varepsilon \to 0$
\begin{multline}\label{D:aa3-2}
\int\limits_{\Omega} {h_x^2 \,dx} + \tfrac{2 c_2}{a_0}
\int\limits_\Omega G_{\varepsilon}(h(x,T)) \; dx + a_0
\iint\limits_{Q_T} {f_{\varepsilon}(h) h^2_{xxx}  \,dx dt}
\leqslant \int\limits_\Omega {h_{0,\varepsilon,x}}^2 \; dx \\
 + \tfrac{2 c_2}{a_0} \left(
\int\limits_\Omega G_{\varepsilon}(h_{0,\varepsilon}) \; dx +
c_{10} \int\limits_0^T \max\left\{ 1 , \left( \int\limits_\Omega
h_x^2(x,t) \;
dx \right)^{\gamma_2} \right\} \; dt \right) \\
+  2 c_7 \int\limits_0^T \max\left\{ 1, \left( \int\limits_\Omega
h_x^2(x,t) \; dx \right)^{\gamma_1} \right\} \; dt
\leqslant  \int\limits_\Omega {h_{0,\varepsilon,x}}^2 \; dx \\
+ \tfrac{2 c_2}{a_0} \int\limits_\Omega
G_{\varepsilon}(h_{0,\varepsilon}) \; dx + c_{11} \int\limits_0^T
\max\left\{ 1, \left( \int\limits_\Omega h_x^2(x,t) \; dx
\right)^{\gamma_1} \right\} \; dt,
\end{multline}
where $c_{11} = \tfrac{2 c_2 c_{10}}{\varepsilon a_0}  + 2 c_7$,
$\gamma_1 = \max \{3, 2m - n, m -n + 1 \}$, $m \geqslant \max \{
\tfrac{n}{2}, n - 1 \} $.

Applying the nonlinear Gr\"onwall lemma \cite{Bihari} to
$$
v(T) \leqslant v(0) + c_{11} \int\limits_0^T \max\{1,v^{\gamma_1}(t)\}
\; dt
$$
with $v(t) = \int (h_x^2(x,t) + 2 c_2/a_0 \: G_{\delta
\varepsilon}(h(x,t)) )\; dx$ yields
$$
v(t) \leqslant
\begin{cases}
\begin{cases}
v(0) + c_{11} t & \mbox{if} \quad t < t_0 := \tfrac{1-v(0)}{c_{11}} \\
\left( 1 - c_{11}(\gamma_1-1) (t-t_0) \right)^{-1/(\gamma_1-1)} &
\mbox{if} \quad t \geqslant t_0
\end{cases}
& \mbox{if} \quad v(0) < 1 \\
\left( v(0)^{1- \gamma_1} - c_{11}(\gamma_1-1) t
\right)^{-1/(\gamma_1-1)} & \mbox{if} \quad v(0) \geqslant 1.
\end{cases}
$$
From this,
\begin{align} \label{H1_control}
& \int\limits_\Omega \{ h_x^2(x,t) +  \tfrac{2c_2}{a_0} G_{\delta \varepsilon}(h(x,t)) \}\; dx \\
& \hspace{.4in} \leqslant 2^{\tfrac{1}{\gamma_1-1}} \max\left\{ 1,
\int\limits_\Omega ({h_{0,\delta \varepsilon,x}}^2(x) +
\tfrac{2c_2}{a_0} G_{\delta \varepsilon}(h_{0,\delta
\varepsilon}(x))) \; dx \right\} = K_{\delta \varepsilon} < \infty
\notag
\end{align}
for all $t \in [0,T_{\delta \varepsilon,loc}]$ where
\begin{equation} \label{Tloc_eps_is}
T_{\delta \varepsilon,loc} := \tfrac{1}{2c_{11}(\gamma_1-1)}
\min\left\{ 1, \left( \int\limits_\Omega ({h_{0,\delta
\varepsilon,x}}^2(x) + \tfrac{2c_2}{a_0} G_{\delta
\varepsilon}(h_{0,\delta \varepsilon}(x)))\,dx
\right)^{-(\gamma_1-1)} \right\}.
\end{equation}
Using the $\delta \to 0, \varepsilon \to 0$ convergence of the
initial data and the choice of $\theta \in (0,2/5)$ (see
(\ref{D:inreg})) as well as the assumption that the initial data
$h_0$ has finite entropy (\ref{C:inval}), the times $T_{\delta
\varepsilon, loc}$ converge to a positive limit and the upper
bound $K$ in (\ref{H1_control}) can be taken finite and
independent of $\delta$ and $\varepsilon$ for $\delta$ and
$\varepsilon$ sufficiently small.  (We refer the reader to the end
of the proof of Lemma \ref{F:local_BF} in this Appendix for a
fuller explanation of a similar case.) Therefore there exists
$\delta_0>0$ and $\varepsilon_0>0$ and $K$ such that the bound
(\ref{H1_control}) holds for all $0 \leqslant \delta < \delta_0$ and $0
< \varepsilon < \varepsilon_0$ with $K$ replacing $K_{\delta
\varepsilon}$ and for all
\begin{equation} \label{Tloc_is}
0 \leqslant t \leqslant T_{loc} := \tfrac{9}{10} \lim_{\varepsilon
\to 0, \delta \to 0} T_{\delta \varepsilon,loc}.
\end{equation}

Using the uniform bound on $\int h_x^2$ that (\ref{H1_control})
provides, one can find a uniform-in-$\delta$-and-$\varepsilon$
bound for the right-hand-side of (\ref{D:aa3}) yielding the
desired a priori bound (\ref{D:a13''}).  Similarly, one can find a
uniform-in-$\delta$-and-$\varepsilon$ bound for the
right-hand-side of (\ref{D:aa2}) yielding the desired a priori
bound (\ref{BF_entropy}).

To prove the bound (\ref{D:d2}), multiply (\ref{D:1r'}) by $- a_0
h_{xx} - a_1 D'_{\varepsilon}(h) $, integrate over $Q_T$,
integrate by parts, use the periodic boundary conditions
(\ref{D:2r'}), to find
\begin{equation}\label{D:d0}
\mathcal{E}_{\delta \varepsilon}(T) + \iint\limits_{Q_T}
{f_{\delta \varepsilon}(h) (a_0 h_{xxx} + a_1
D''_{\varepsilon}(h)h_x)^2 \,dx dt} = \mathcal{E}_{\delta
\varepsilon}(0).
\end{equation}

The parameters $\delta_0$ and $\varepsilon_0$ are determined by
$a_0$, $a_1$,  $| \Omega |$, $\int h_0$, $\| h_{0x} \|_2$, $\int
h_0^{2-n}$, by how quickly $M_{\delta \varepsilon}$ converges to
$M$,
and by how quickly the approximate initial data  (\ref{D:inreg}),
$h_{0,\delta \varepsilon}$, converge to $h_0$ in $H^1(\Omega)$.

The time $T_{loc}$ and the constants $K_1$, and $K_2$  are
determined by $\delta_0$, $\varepsilon_0$, $a_0$, $a_1$, $| \Omega
|$, $\int h_0$, $\| h_{0x} \|_2$, and $\int h_0^{2-n}$.
\end{proof}

\begin{proof}[Sketch of Proof of Lemma~\ref{lemLL}]
In the following, we denote the positive, classical solution
$h_{\varepsilon}$ by $h$ whenever there is no chance of confusion.

Taking $\delta \to 0$ in (\ref{D:a5}) yields
\begin{align} 
&\notag \tfrac{1}{2}\int\limits_{\Omega} {h_x^2 \,dx} +
\tfrac{a_0}{2}
\iint\limits_{Q_T} {f_{\varepsilon}(h) h^2_{xxx}\,dx dt} \\
& \notag
\leqslant \tfrac{1}{2}\int\limits_{\Omega} {h_{0\varepsilon,
x}^2\,dx} + \tfrac{a_1^2}{2a_0} \iint\limits_{Q_T} {|h|^{2m-n}
h^2_{x} \,dx dt} \\
&  \label{D:a5-n}
\leqslant \tfrac{1}{2} \int\limits_{\Omega} { h_{0\varepsilon,
x}^2\,dx} + \tfrac{a_1^2}{2a_0} \int\limits_0^T \| h(\cdot,t)
\|_\infty^{2m-n} \; \int\limits_\Omega h_x^2(x,t) \; dx\,dt.
\end{align}
Applying the nonlinear Gr\"onwall lemma \cite{Bihari} to
$$
v(T) \leqslant v(0) + \int\limits_0^T A_1(t) \; v(t)\; dt
$$
with $v(t) = \int h_x^2(x,t) \; dx$, $A_1(t) = a_1^2/a_0 \|
h(\cdot,t) \|_\infty^{2m-n}$ yields
$$
v(T) \leqslant v(0)\; e^{B_1(T)}, \text{ where }B_1(t) =
\int\limits_{0}^t {A_1(s)\;ds} .
$$
Similarly, taking $\delta \to 0$ in (\ref{D:a5}) and
(\ref{D:aa1}), due to (\ref{D:simp}) and (\ref{Young}), yield
\begin{multline}\label{D:a5-nnn}
\int\limits_{\Omega} {h_x^2 \,dx} + a_0 \iint\limits_{Q_T}
{f_{\varepsilon}(h) h^2_{xxx}\,dx dt} \leqslant
\int\limits_{\Omega} {h_{0\varepsilon, x}^2\,dx} +
\tfrac{a_1^2}{a_0} \iint\limits_{Q_T} {h^{2m-n}h^2_{x} \,dx dt} \\
=  \int\limits_{\Omega} { h_{0\varepsilon, x}^2\,dx} -
\tfrac{a_1^2}{a_0(2m -n + 1)} \iint\limits_{Q_T} {h^{2m-n+
1}h_{xx}
\,dx dt} \leqslant \int\limits_{\Omega} { h_{0\varepsilon, x}^2\,dx} \\
+ \tfrac{a_0}{2}\iint\limits_{Q_T} {h^2_{xx}\,dx dt} +
\tfrac{a_1^4}{2a_0^3(2m -n + 1)^2} \iint\limits_{Q_T} {h^{2(2m-n+
1)}\,dx dt} \leqslant \int\limits_{\Omega} { h_{0\varepsilon, x}^2\,dx}\\
+ \tfrac{a_0}{2}\iint\limits_{Q_T} {h^2_{xx}\,dx dt} +
\tfrac{a_1^4}{2a_0^3(2m -n + 1)^2} \int\limits_0^T \| h(\cdot,t)
\|_\infty^{4m-n} \; \int\limits_\Omega G_{\varepsilon}(h(x,t)) \;
dx\,dt,
\end{multline}
\begin{multline}\label{D:aa1-0}
\int\limits_{\Omega} {G_{\varepsilon}(h(x,T))\,dx} + a_0
\iint\limits_{Q_T} {h^2_{xx}\,dx dt} \leqslant
\int\limits_{\Omega} {G_{ \varepsilon}(h_{0,\varepsilon})\,dx} +
a_1 \iint\limits_{Q_T} {h^{m-n}h^2_{x}\,dx dt} \\
\leqslant \int\limits_{\Omega} {G_{
\varepsilon}(h_{0,\varepsilon})\,dx} - \tfrac{a_1}{m - n + 1}
\iint\limits_{Q_T} {h^{m-n+1}h_{xx}\,dx dt}\leqslant
\int\limits_{\Omega} {G_{\varepsilon}(h_{0,\varepsilon})\,dx} \\
+ \tfrac{a_0}{2}\iint\limits_{Q_T} {h^2_{xx}\,dx dt} +
\tfrac{a_1^2}{2a_0(m -n + 1)^2} \iint\limits_{Q_T} {h^{2(m-n+
1)}\,dx dt} \leqslant  \int\limits_{\Omega}
{G_{\varepsilon}(h_{0,\varepsilon})\,dx} \\
+ \tfrac{a_0}{2}\iint\limits_{Q_T} {h^2_{xx}\,dx dt} +
\tfrac{a_1^2}{2a_0(m -n +1)^2} \int\limits_0^T \| h(\cdot,t)
\|_\infty^{2m-n} \; \int\limits_\Omega G_{\varepsilon}(h(x,t)) \;
dx\,dt .
\end{multline}
Summing (\ref{D:a5-nnn}) and (\ref{D:aa1-0}), we find that
\begin{multline}\label{D:sum0}
\int\limits_{\Omega} {\{ h_x^2(x,T) + G_{\varepsilon}(h(x,T))
\}\,dx} +
a_0 \iint\limits_{Q_T} {f_{\varepsilon}(h) h^2_{xxx}\,dx dt} \\
\leqslant \int\limits_{\Omega} { \{ h_{0\varepsilon, x}^2 +
G_{\varepsilon}(h_{0,\varepsilon}) \}\,dx} + \int\limits_0^T
A_2(t) \; \int\limits_\Omega \{ h_x^2(x,t) +
G_{\varepsilon}(h(x,t)) \} \; dx\,dt,
\end{multline}
where $A_2(t) = \tfrac{a_1^4}{2a_0^3(2m -n + 1)^2} \| h(\cdot,t)
\|_\infty^{4m-n} + \tfrac{a_1^2}{2a_0(m -n +1)^2} \| h(\cdot,t)
\|_\infty^{2m-n}$. Applying the nonlinear Gr\"onwall lemma
\cite{Bihari} to
$$
v(T) \leqslant v(0) + \int\limits_0^T A_2(t) \; v(t)\; dt
$$
with $v(t) = \int \{ h_x^2(x,t) + G_{\varepsilon}(h(x,t)) \} \;
dx$,  yields
$$
v(T) \leqslant v(0)\; e^{B_2(T)}, \text{ where }B_2(t) =
\int\limits_{0}^t {A_2(s)\;ds} .
$$
\end{proof}

\begin{proof}[Sketch of Proof of Lemma~\ref{MainAE2}]
In the following, we denote the positive, classical solution
$h_{\varepsilon}$ by $h$ whenever there is no chance of confusion.

Multiplying (\ref{D:1r'}) with $\delta = 0$ by
$G'^{(\alpha)}_{\varepsilon} (h)$, integrating over $Q_T$, taking
$\delta \to 0$, and using the periodic boundary conditions
(\ref{D:2r'}), yields
\begin{align}
& \label{E:b1} \int\limits_{\Omega}
{G^{(\alpha)}_{\varepsilon}(h(x,T))\,dx} + a_0 \iint\limits_{Q_T}
{h^{\alpha}h^2_{xx}\,dx dt}  + a_0 \tfrac{\alpha(1 - \alpha)}{3}
\iint\limits_{Q_T} {h^{\alpha -2
}h^4_{x}\,dx dt} \\
& = \int\limits_{\Omega}
{G^{(\alpha)}_{\varepsilon}(h_{0\varepsilon})\,dx} + a_1
\iint\limits_{Q_T} {h^{\alpha} D''_{\varepsilon}(h) h^2_{x}\,dx
dt} \leqslant \int\limits_{\Omega}
{G^{(\alpha)}_{\varepsilon}(h_{0\varepsilon})\,dx} + a_1
\iint\limits_{Q_T} {h^{\alpha + m - n} h^2_{x}\,dx dt} . \notag
\end{align}

\textbf{Case 1: $0 < \alpha < 1$.} The coefficient multiplying
$\iint h^{\alpha-2} h_x^4$ in (\ref{E:b1}) is positive and can
therefore be used to control the term $\iint h^{\alpha + m -n}
h_x^2$ on the right--hand side of (\ref{E:b1}).  Specifically,
using the Cauchy-Schwartz inequality and the Cauchy inequality,
\begin{multline}\label{E:b2}
a_1 \iint\limits_{Q_T} {h^{\alpha + m - n} h^2_{x}\,dx dt}
\leqslant a_1 \iint\limits_{Q_T} {h^{\alpha + m - n} h^2_{x}\,dx
dt} \leqslant \\
\tfrac{a_0\alpha(1 - \alpha)}{6}\iint\limits_{Q_T} {h^{\alpha-2}
h^4_{x}\,dx dt} + \tfrac{3a_1^2}{2a_0 \alpha(1 -
\alpha)}\iint\limits_{Q_T} {h^{\alpha + 2(m-n + 1)} \,dx dt}.
\end{multline}
Using the bound (\ref{E:b2}) in (\ref{E:b1}) yields
\begin{align}
& \label{E:b3} \int\limits_{\Omega}
{G^{(\alpha)}_{\varepsilon}(h(x,T))\,dx} + a_0 \iint\limits_{Q_T}
{h^{\alpha}h^2_{xx}\,dx dt}  + a_0 \tfrac{\alpha(1 - \alpha)}{6}
\iint\limits_{Q_T} {h^{\alpha -2
}h^4_{x}\,dx dt} \\
& \leqslant \int\limits_{\Omega}
{G^{(\alpha)}_{\varepsilon}(h_{0\varepsilon})\,dx}
+\tfrac{3a_1^2}{2a_0\alpha(1 - \alpha)}\iint\limits_{Q_T}
{h^{\alpha + 2(m-n+1)} \,dx dt}.
\end{align}
By  (\ref{LpH1}),
\begin{equation}\label{E:b9}
\iint\limits_{Q_T} {h^{\alpha + 2(m-n+1)} \,dx dt} \leqslant b_4
 \int \limits_0^T {\left(
\int\limits_{\Omega} {h^2_{x} \,dx} \right)^{\tfrac{\alpha}{2}+ m
- n + 1} dt} + b_5 M_\varepsilon^{\alpha+2(m-n+1)} \; T.
\end{equation}
Using (\ref{E:b9}) in (\ref{E:b3}) yields
\begin{align}
&\notag \int\limits_{\Omega}
{G^{(\alpha)}_{\varepsilon}(h(x,T))\,dx} + a_0 \iint\limits_{Q_T}
{h^{\alpha}h^2_{xx}\,dx dt}  + a_0 \tfrac{\alpha(1 - \alpha)}{6}
\iint\limits_{Q_T} {h^{\alpha -2 }h^4_{x}\,dx dt}  \leqslant \\
& \notag \hspace{.2in} \int\limits_{\Omega}
{G^{(\alpha)}_{\varepsilon}(h_{0\varepsilon})\,dx} + d_1
\int\limits_0^T \left( \int\limits_\Omega h_x^2 \; dx
\right)^{\tfrac{\alpha}{2} + m - n
+ 1} \; dt + d_2 \; T \\
& \hspace{.2in} \label{E:b9a} \leqslant \int\limits_{\Omega}
{G^{(\alpha)}_{\varepsilon}(h_{0\varepsilon})\,dx} + d_3
\int\limits_0^T \max\left\{ 1, \left( \int\limits_\Omega h_x^2 \;
dx \right)^{\tfrac{\alpha}{2}+ m - n + 1} \right\} \; dt
\end{align}
where
$$
d_1 = b_4 \; \tfrac{3 a_1^2}{2 a_0 \alpha (1-\alpha)}, \quad d_2 =
b_5 \; \tfrac{3 a_1^2}{2 a_0 \alpha (1-\alpha)} \;
M_\varepsilon^{\alpha+2(m - n + 1)}, \quad d_3 = d_1+d_2.
$$

Taking $\delta \to 0$ in (\ref{D:a5}) yields
\begin{equation}\label{E:b9b}
\int\limits_{\Omega} {h_x^2 \,dx} + a_0 \iint\limits_{Q_T}
{f_{\varepsilon}(h) h^2_{xxx}\,dx dt}  \leqslant
\int\limits_{\Omega} {h_{0\varepsilon, x}^2\,dx} +
\tfrac{a_1^2}{a_0} \iint\limits_{Q_T} { h^{2m - n} h^2_{x} \,dx
dt}.
\end{equation}
Applying the Cauchy inequality,
\begin{align} \label{E:b9c}
& \tfrac{a_1^2}{a_0} \iint\limits_{Q_T} h^{2m-n} h_x^2 \; dx dt \\
& \hspace{.2in}  \notag \leqslant \tfrac{a_0 \alpha(1-\alpha)}{6}
\iint\limits_{Q_T} h^{\alpha-2} h_x^4 \; dx dt + \tfrac{3 a_1^4}{
2a_0^3 \alpha(1-\alpha)} \iint\limits_{Q_T} h^{2(2m-n+1)-\alpha}
\; dx dt.
\end{align}
By  (\ref{LpH1}),
\begin{equation}\label{E:b8'}
\iint\limits_{Q_T} {h^{2(2m-n+1)-\alpha} \,dx dt} \leqslant b_4
\int \limits_0^T {\left( \int\limits_{\Omega} {h^2_{x} \,dx}
\right)^{2m -n + 1 -\tfrac{\alpha}{2}} dt} + b_5
M_\varepsilon^{2(2m-n+1)-\alpha} \; T.
\end{equation}
Using (\ref{E:b9c}) and (\ref{E:b8'}) in (\ref{E:b9b}) yields
\begin{align} \notag
& \int\limits_{\Omega} {h_x^2 \,dx} + a_0 \iint\limits_{Q_T}
{f_{\varepsilon}(h) h^2_{xxx}\,dx dt} \leqslant
\int\limits_{\Omega} {h_{0\varepsilon, x}^2\,dx} + \tfrac{a_0
\alpha(1-\alpha)}{6}
\iint\limits_{Q_T} { h^{\alpha-2} h^4_{x} \,dx dt}\\
&  \hspace{.2in} \label{E:b9d}  + d_4 \int\limits_0^T \left(
\int\limits_\Omega h_x^2 \; dx \right)^{2m -n + 1
-\tfrac{\alpha}{2}} \; dt+ d_5 \; T
\leqslant  \int\limits_{\Omega} {h_{0\varepsilon, x}^2\,dx} \\
& \hspace{.2in} \notag + \tfrac{a_0 \alpha(1-\alpha)}{6}
\iint\limits_{Q_T} { h^{\alpha-2} h^4_{x} \,dx dt} + d_6
\int\limits_0^T \max\left\{ 1, \left( \int\limits_\Omega h_x^2 \;
dx \right)^{2m-n + 1 -\tfrac{\alpha}{2}} \right\} \; dt
\end{align}
where
$$
d_4 = \tfrac{3 a_1^4}{2 a_0^3 \alpha(1-\alpha)}  \; b_4, \qquad
d_5 = b_5 \tfrac{3 a_1^4}{2a_0^3 \alpha(1-\alpha)} \;
M_\varepsilon^{2(2m-n+1) -\alpha} , \qquad d_6 = d_4 + d_5.
$$
Add
$$
\int\limits_{\Omega} {G^{(\alpha)}_{\varepsilon}(h(x,T))\,dx}
$$
to both sides of (\ref{E:b9d}) and add
$$
a_0 \iint\limits_{Q_T} h^\alpha h_{xx}^2 \; dx dt
$$
to the right--hand side of the resulting inequality.    Using
(\ref{E:b9a}) yields
\begin{align} \label{E:b9e}
& \int\limits_{\Omega} {h_x^2(x,T) \,dx} + \int\limits_{\Omega}
{G^{(\alpha)}_{\varepsilon}(h(x,T))\,dx} + a_0
\iint\limits_{Q_T} {f_{\varepsilon}(h) h^2_{xxx}\,dx dt} \\
& \hspace{.2in} \notag \leq
 \int\limits_{\Omega} {h_{0\varepsilon, x}^2\,dx}
+ \int\limits_\Omega G^{(\alpha)}_{\varepsilon}(h_{0\varepsilon})
\,dx
+ d_3 \int\limits_0^T \max\left\{ 1, \left( \int\limits_\Omega h_x^2 \; dx \right)^{\tfrac{\alpha}{2}+m-n +1} \right\} \\
& \hspace{.4in} \notag + d_6 \int\limits_0^T \max\left\{ 1, \left(
\int\limits_\Omega h_x^2 \; dx \right)^{2m- n +
1-\tfrac{\alpha}{2}}
\right\} \; dt \\
& \hspace{.2in} \leqslant \notag
 \int\limits_{\Omega} {h_{0\varepsilon, x}^2\,dx}
+ \int\limits_\Omega G^{(\alpha)}_{\varepsilon}(h_{0\varepsilon})
\,dx + d_7 \int\limits_0^T \max\left\{ 1, \left(
\int\limits_\Omega h_x^2 \; dx \right)^{\gamma_3} \right\}
\end{align}
where $d_7 = d_3 + d_6$, $\gamma_3 = \max \{\alpha/2 + m - n +1,
2m- n + 1 - \alpha/2 \}$.

Applying the nonlinear Gr\"onwall lemma \cite{Bihari} to
$$
v(T) \leqslant v(0) + d_7 \int\limits_0^T \max\{1,v^{\gamma_3}(t)\} \;
dt$$ with $v(T) = \int h_x^2(x,T) + \:
G_{\varepsilon}^{(\alpha)}(h(x,T)) \; dx$ yields
$$
v(t) \leq
\begin{cases}
\begin{cases}
v(0) + d_7 t & \mbox{if} \quad t < t_0 := \tfrac{1-v(0)}{d_9} \\
\left( 1 -  d_7(\gamma_3 - 1) (t-t_0) \right)^{-\tfrac{1}{\gamma_3
- 1}} & \mbox{if} \quad t \geqslant t_0
\end{cases}
& \mbox{if} \quad v(0) < 1 \\
\left( v(0)^{1 - \gamma_3} - d_7(\gamma_3 - 1) t
\right)^{-\tfrac{1}{\gamma_3 - 1}} & \mbox{if} \quad v(0) \geqslant 1
\end{cases}
$$
From this,
\begin{align} \label{bound1}
& \int\limits_\Omega (h_x^2(x,T) + G_{\varepsilon}^{(\alpha)}(h(x,T))) \; dx \\
& \hspace{.4in} \leqslant 2^{\tfrac{1}{\gamma_3 - 1}} \max\left\{ 1,
\int\limits_\Omega ({h_{0,\varepsilon}}_x^2(x) +
G_{\varepsilon}^{(\alpha)}(h_{0,\varepsilon}(x))) \; dx \right\} =
K_\varepsilon < \infty \notag
\end{align}
for all
$$
0 \leqslant T \leqslant T_{\varepsilon,loc}^{(\alpha)} :=
\tfrac{1}{2d_7(\gamma_3 - 1)} \min\left\{ 1,   \left(
\int\limits_\Omega ({h_{0,\varepsilon}}_x^2(x) +
G_{\varepsilon}^{(\alpha)}(h_{0,\varepsilon}(x))) \;
dx\right)^{-(\gamma_3 - 1)} \right\}.
$$
The bound (\ref{bound1}) holds for all $0 < \varepsilon <
\varepsilon_0$ where $\varepsilon_0$ is from Lemma \ref{MainAE}
and for all $t \leq
\min\{T_{loc},T_{\varepsilon,loc}^{(\alpha)}\}$ where $T_{loc}$ is
from Lemma \ref{MainAE}.

Using the $\delta \to 0, \varepsilon \to 0$ convergence of the
initial data and the choice of $\theta \in (0,2/5)$ (see
(\ref{D:inreg})) as well as the assumption that the initial data
$h_0$ has finite $\alpha$-entropy (\ref{finite_alpha_ent}), the
times $T_{\varepsilon, loc}^{(\alpha)}$ converge to a positive
limit and the upper bound $K_\varepsilon$ in (\ref{bound1}) can be
taken finite and independent of $\varepsilon$. (We refer the
reader to the end of the proof of Lemma \ref{F:local_BF} in this
Appendix for a fuller explanation of a similar case.) Therefore
there exists $\varepsilon_0^{(\alpha)}$ and $K$ such that the
bound (\ref{bound1}) holds for all $0 < \varepsilon <
\varepsilon_0^{(\alpha)}$ with $K$ replacing $K_\varepsilon$ and
for all
\begin{equation} \label{Tloc_alpha_is}
0 \leqslant t \leqslant T_{loc}^{(\alpha)} := \min\left\{T_{loc},
\tfrac{9}{10} \lim_{\varepsilon \to 0}
T_{\varepsilon,loc}^{(\alpha)} \right\}
\end{equation}
where $T_{loc}$ is the time from Lemma \ref{MainAE}.  Also,
without loss of generality, $\varepsilon_0^{(\alpha)}$ can be
taken to be less than or equal to the $\varepsilon_0$ from Lemma
\ref{MainAE}.

Using the uniform bound on $\int h_x^2$ that (\ref{bound1})
provides, one can find a uniform-in-$\varepsilon$ bound for the
right-hand-side of (\ref{E:b9a}) yielding the desired bound
\begin{equation} \label{alpha_bound1}
\int\limits_{\Omega} {G^{(\alpha)}_{\varepsilon}(h(x,T))\,dx} +
a_0 \iint\limits_{Q_T} {h^{\alpha}h^2_{xx}}\,dx dt
 +
a_0 \tfrac{\alpha(1 - \alpha)}{6} \iint\limits_{Q_T} {h^{\alpha -2
}h^4_{x}\,dx dt}  \leqslant K_1
\end{equation}
which holds for all $0 <\varepsilon < \varepsilon_0^{(\alpha)}$
and all $0 \leqslant T \leqslant T_{loc}^{(\alpha)}$.

It remains to argue that (\ref{alpha_bound1}) implies that for all
$0 < \varepsilon < \varepsilon_0^{(\alpha)}$ that
$h_\varepsilon^{\alpha/2 + 1}$ and $h_\varepsilon^{\alpha/4 +
1/2}$ are contained in balls in $L^{2}(0, T; H^2(\Omega))$ and
$L^{2}(0, T; W^1_4(\Omega))$ respectively. It suffices to show
that
$$
\iint\limits_{Q_T} \left( h_\varepsilon^{\alpha/2 +
1}\right)^2_{xx} \; dx dt \leqslant K, \qquad \iint\limits_{Q_T} \left(
h_\varepsilon^{\alpha/4 + 1/2}\right)^4_x \; dx dt \leqslant K
$$
for some $K$ that is independent of $\varepsilon$ and $T$. The
integral $\iint (h_\varepsilon^{\alpha/2+1})_{xx}^2$ is a linear
combination of $\iint h^{\alpha-2} h_x^4$, $\iint h^{\alpha-1}
h_x^2 h_{xx}$, and $\iint h^{\alpha} h_{xx}^2$. Integration by
parts and the periodic boundary conditions imply
\begin{equation} \label{calc_ident}
\tfrac{1 - \alpha}{3} \iint\limits_{Q_T} {h^{\alpha -2
}h^4_{x}\,dx dt} = \iint\limits_{Q_T} {h^{\alpha -1 }h^2_{x}
h_{xx}\,dx dt}
\end{equation}
Hence $\iint (h_\varepsilon^{\alpha/2+1})_{xx}^2$ is a linear
combination of $\iint h^{\alpha-2} h_x^4$, and $\iint h^{\alpha}
h_{xx}^2$. By (\ref{alpha_bound1}), the two integrals are
uniformly bounded independent of $\varepsilon$ and $T$ hence
$\iint (h_\varepsilon^{\alpha/2+1})_{xx}^2$ is as well, yielding
the first part of (\ref{E:b13}).

The uniform bound of $\iint (h_\varepsilon^{\alpha/4 + 1/2})_x^4$
follows immediately from the uniform bound of $\iint h^{\alpha-2}
h_x^4$, yielding the second part of (\ref{E:b13}).

\textbf{Case 2: $-\tfrac{1}{2} < \alpha < 0$.}  For $\alpha < 0$
the coefficient multiplying $\iint h^{\alpha-2} h_x^4$ in
(\ref{E:b1}) is negative.  However, we will show that if $\alpha >
-1/2$ then one can replace this coefficient with a positive
coefficient while also controlling the term $\iint h^\alpha h_x^2$
on the right-hand side of (\ref{E:b1}).

Applying the Cauchy-Schwartz inequality to the right--hand side of
(\ref{calc_ident}), dividing by $\sqrt{\iint h^{\alpha-2} h_x^4}$,
and squaring both sides of the resulting inequality yields
\begin{equation}\label{E:b4}
\iint\limits_{Q_T} {h^{\alpha -2 }h^4_{x}\,dx dt} \leq
\tfrac{9}{(1 - \alpha)^2} \iint\limits_{Q_T} {h^{\alpha}
h^2_{xx}\,dx dt} \qquad \forall \alpha < 1.
\end{equation}
Using (\ref{E:b4}) in (\ref{E:b1}) yields
\begin{align}
& \label{E:b5} \int\limits_{\Omega}
{G^{(\alpha)}_{\varepsilon}(h(x,T))\,dx} + a_0
\tfrac{1+2\alpha}{1-\alpha}\iint\limits_{Q_T} {h^{\alpha}h^2_{xx}\,dx dt}  \\
& \hspace{.2in} \notag \leqslant \int\limits_{\Omega}
{G^{(\alpha)}_{\varepsilon}(h_{0\varepsilon})\,dx} + a_1
\iint\limits_{Q_T} {h^{\alpha + m - n} h^2_{x}\,dx dt}. \notag
\end{align}
Note that if $\alpha > -1/2$ then all the terms on the left--hand
side of (\ref{E:b5}) are positive.  We now control the term $\iint
h^\alpha h_x^2$ on the right-hand side of (\ref{E:b5}).

By integration by parts and the periodic boundary conditions
\begin{equation}  \label{calc_ident2}
\iint\limits_{Q_T} h^{\alpha + m - n} h_x^2 \; dx dt = -
\tfrac{1}{\alpha + m - n + 1} \iint\limits_{Q_T} h^{\alpha+ m - n
+ 1} h_{xx} \; dx dt
\end{equation}
Applying the Cauchy-Schwartz inequality and the Cauchy inequality
to (\ref{calc_ident2}) yields
\begin{equation} \label{E:b6}
a_1 \iint\limits_{Q_T} h^{\alpha + m - n} h_x^2\; dx dt \leq
\tfrac{a_0 (1+2 \alpha)}{2(1-\alpha)}  \iint\limits_{Q_T} h^\alpha
h_{xx}^2 \ \; dx dt + \tfrac{a_1^2 (1-\alpha)}{2 a_0 (1+2 \alpha)
(\alpha + m - n + 1)^2} \iint\limits_{Q_T} h^{\alpha+2(m-n+1)}\;
dx dt
\end{equation}
Using inequality (\ref{E:b6}) in (\ref{E:b5}) yields
\begin{align}
& \label{E:b6a} \int\limits_{\Omega}
{G^{(\alpha)}_{\varepsilon}(h(x,T))\,dx} + a_0
\tfrac{1+2\alpha}{2(1-\alpha)}\iint\limits_{Q_T} {h^{\alpha}h^2_{xx}\,dx dt}  \\
& \notag \leqslant \int\limits_{\Omega}
{G^{(\alpha)}_{\varepsilon}(h_{0\varepsilon})\,dx} + \tfrac{a_1^2
(1-\alpha)}{2 a_0 (1+2 \alpha) (\alpha + m - n + 1)^2}
\iint\limits_{Q_T} h^{\alpha +2(m-n+1)}\; dx dt . \notag
\end{align}
Adding
$$
\tfrac{a_0 (1+2 \alpha) (1-\alpha)}{36}  \iint\limits_{Q_T}
h^{\alpha-2} h_x^4 \; dx dt
$$
to both sides of (\ref{E:b6a}) and using the inequality
(\ref{E:b4}) yields
\begin{align}
& \label{E:b6b} \int\limits_{\Omega}
{G^{(\alpha)}_{\varepsilon}(h(x,T))\,dx} + a_0
\tfrac{(1+2\alpha)}{4(1-\alpha)}\iint\limits_{Q_T} {h^{\alpha}h^2_{xx}\,dx dt}  \\
& \hspace{.2in} \notag + \tfrac{a_0(1+2\alpha)(1-\alpha)}{36}
\iint\limits_{Q_T} h^{\alpha-2} h_x^4 \; dx dt \leq
\int\limits_{\Omega}
{G^{(\alpha)}_{\varepsilon}(h_{0\varepsilon})\,dx} \\
& \hspace{.2in} \notag + \tfrac{a_1^2 (1-\alpha)}{2 a_0 (1+2
\alpha) (\alpha + m - n + 1)^2} \iint\limits_{Q_T} h^{\alpha
+2(m-n+1)}\; dx dt. \notag
\end{align}
Using (\ref{E:b9}) in (\ref{E:b6b}) yields
\begin{align}
&\notag \int\limits_{\Omega}
{G^{(\alpha)}_{\varepsilon}(h(x,T))\,dx} +
\tfrac{a_0(1+2\alpha)}{4(1-\alpha)}
\iint\limits_{Q_T} {h^{\alpha}h^2_{xx}\,dx dt}  \\
& \notag \hspace{.2in} + \tfrac{a_0(1+2\alpha)(1-\alpha)}{36}
\iint\limits_{Q_T} {h^{\alpha -2 }h^4_{x}\,dx dt}  \leq
\int\limits_{\Omega}
{G^{(\alpha)}_{\varepsilon}(h_{0\varepsilon})\,dx} \\
& \hspace{.2in} \notag + e_1 \int\limits_0^T \left(
\int\limits_\Omega h_x^2 \; dx \right)^{\tfrac{\alpha}{2} + m - n
+ 1} \; dt + e_2 \; T \\
& \hspace{.2in} \label{E:b9a'} \leqslant \int\limits_{\Omega}
{G^{(\alpha)}_{\varepsilon}(h_{0\varepsilon})\,dx} + e_3
\int\limits_0^T \max\left\{ 1, \left( \int\limits_\Omega h_x^2 \;
dx \right)^{\tfrac{\alpha}{2}+ m - n +1} \right\} \; dt
\end{align}
where
$$
e_1 = \tfrac{a_1^2(1-\alpha)}{2a_0(1+2\alpha)(\alpha + m - n +
1)^2}\; b_4, \ \ e_2 = b_5
\tfrac{a_1^2(1-\alpha)}{2a_0(1+2\alpha)(\alpha + m - n + 1)^2}\;
M_\varepsilon^{\alpha+2(m - n + 1)},
$$
and $e_3 = e_1 + e_2 $.

Recall the bound (\ref{E:b9b}):
\begin{equation}\label{E:b9b'}
\int\limits_{\Omega} {h_x^2 \,dx} + a_0 \iint\limits_{Q_T}
{f_{\varepsilon}(h) h^2_{xxx}\,dx dt} \leqslant \int\limits_{\Omega}
{h_{0\varepsilon, x}^2\,dx} + \tfrac{a_1^2}{a_0}
\iint\limits_{Q_T} { h^{2m-n} h^2_{x} \,dx dt} .
\end{equation}
As before, by the Cauchy-Schwartz inequality and the Cauchy
inequality,
\begin{align} \label{E:b9c'}
& \tfrac{a_1^2}{a_0} \iint\limits_{Q_T} h^{2m - n} h_x^2 \; dx dt
\leqslant \tfrac{a_0(1+2\alpha)(1-\alpha)}{36}
\iint\limits_{Q_T} h^{\alpha-2} h_x^4 \; dx dt\\
& \hspace{1.4in}  \notag + \tfrac{9 a_1^4}{a_0^3
(1+2\alpha)(1-\alpha)} \iint\limits_{Q_T} h^{2(2m -n + 1)-\alpha}
\; dx dt.
\end{align}
Using (\ref{E:b9c'}), and  (\ref{E:b8'}) in (\ref{E:b9b'}) yields
\begin{align} \notag
& \int\limits_{\Omega} {h_x^2 \,dx} + a_0 \iint\limits_{Q_T}
{f_{\varepsilon}(h) h^2_{xxx}\,dx dt} \leqslant
\int\limits_{\Omega} {h_{0\varepsilon, x}^2\,dx} + \tfrac{a_0
(1+2\alpha)(1-\alpha)}{36}
\iint\limits_{Q_T} { h^{\alpha-2} h^4_{x} \,dx dt}\\
&  \hspace{.2in} \label{E:b9d'}  + e_4 \int\limits_0^T \left(
\int\limits_\Omega h_x^2 \; dx \right)^{2m-n+1-\tfrac{\alpha}{2}}
\; dt + e_5 \; T
\leqslant  \int\limits_{\Omega} {h_{0\varepsilon, x}^2\,dx} \\
& \hspace{.2in} \notag + \tfrac{a_0 (1+2\alpha)(1-\alpha)}{36}
\iint\limits_{Q_T} { h^{\alpha-2} h^4_{x} \,dx dt} + e_6
\int\limits_0^T \max\left\{ 1, \left( \int\limits_\Omega h_x^2 \;
dx \right)^{2m-n+1-\tfrac{\alpha}{2}} \right\} \; dt
\end{align}
where
$$
e_4 = \tfrac{9 a_1^4}{a_0^3 (1+2\alpha)(1-\alpha)} \; b_4, \ e_5 =
b_5 \tfrac{9 a_1^4}{a_0^3 (1+2\alpha)(1-\alpha)} \;
M_\varepsilon^{2(2m -n + 1)-\alpha} ,\ e_6 = e_4 + e_5.
$$
Add
$$
\int\limits_{\Omega} {G^{(\alpha)}_{\varepsilon}(h(x,T))\,dx}
$$
to both sides of (\ref{E:b9d'}) and add
$$
\tfrac{a_0(1+2\alpha)}{4(1-\alpha)} \iint\limits_{Q_T} h^\alpha
h_{xx}^2 \; dx dt
$$
to the right--hand side of the resulting inequality. Just as
(\ref{E:b9a}) and (\ref{E:b9b}) yielded (\ref{E:b9e}),
(\ref{E:b9a'}) combined with the above inequality yields
\begin{align} \label{E:b9e'}
& \int\limits_{\Omega} {h_x^2(x,T) \,dx} + \int\limits_{\Omega}
{G^{(\alpha)}_{\varepsilon}(h(x,T))\,dx} + a_0
\iint\limits_{Q_T} {f_{\varepsilon}(h) h^2_{xxx}\,dx dt} \\
& \hspace{.2in} \notag \leq
 \int\limits_{\Omega} {h_{0\varepsilon, x}^2\,dx}
+ \int\limits_\Omega G^{(\alpha)}_{\varepsilon}(h_{0\varepsilon})
\,dx
+ e_3 \int\limits_0^T \max\left\{ 1, \left( \int\limits_\Omega h_x^2 \; dx \right)^{\tfrac{\alpha}{2}+ m - n +1} \right\} \\
& \hspace{.4in} \notag + e_6 \int\limits_0^T \max\left\{ 1, \left(
\int\limits_\Omega h_x^2 \; dx \right)^{2m-n+1-\tfrac{\alpha}{2}}
\right\} \; dt \\
& \hspace{.2in} \leqslant \notag
 \int\limits_{\Omega} {h_{0\varepsilon, x}^2\,dx}
+ \int\limits_\Omega G^{(\alpha)}_{\varepsilon}(h_{0\varepsilon})
\,dx + e_7 \int\limits_0^T \max\left\{ 1, \left(
\int\limits_\Omega h_x^2 \; dx \right)^{\gamma_3} \right\}
\end{align}
where $e_7 = e_3 + e_6$, $\gamma_3 = \max \{ \alpha /2+ m - n +1,
2m-n+1 - \alpha/2\} $.

The rest of the proof now continues as in the $0 < \alpha < 1$
case. Specifically, one finds a bound
\begin{align} \label{bound2}
& \int\limits_\Omega (h_x^2(x,T) + G_{\varepsilon}^{(\alpha)}(h(x,T))) \; dx \\
& \hspace{.4in} \leqslant 2^{\tfrac{1}{\gamma_3 - 1}} \max\left\{ 1,
\int\limits_\Omega ({h_{0,\varepsilon,x}}^2(x) +
G_{\varepsilon}^{(\alpha)}(h_{0,\varepsilon}(x))) \; dx \right\} =
K_\varepsilon < \infty \notag
\end{align}
for all
$$
0 \leqslant T \leqslant T_{\varepsilon,loc}^{(\alpha)} :=
\tfrac{1}{2e_7(\gamma_3 - 1)}  \min\left\{ 1,   \left(
\int\limits_\Omega ({h_{0,\varepsilon,x}}^2(x) +
G_{\varepsilon}^{(\alpha)}(h_{0,\varepsilon}(x))) \;
dx\right)^{-(\gamma_3 - 1)} \right\}.
$$
The time $T_{loc}^{(\alpha)}$ is defined as in
(\ref{Tloc_alpha_is}) and the uniform bound (\ref{bound2}) used to
bound the right hand side of (\ref{E:b9a'}) yields the desired
bound
\begin{align} \notag
& \int\limits_{\Omega} {G^{(\alpha)}_{\varepsilon}(h(x,T))\,dx} +
\tfrac{a_0(1+2\alpha)}{4(1-\alpha)}
\iint\limits_{Q_T} {h^{\alpha}h^2_{xx}\,dx dt}  \\
& \hspace{1.2in} + \tfrac{a_0(1+2\alpha)(1-\alpha)}{36}
\iint\limits_{Q_T} {h^{\alpha -2 }h^4_{x}\,dx dt}  \leqslant K_2.
\label{alpha_bound2}
\end{align}
\end{proof}

\begin{proof}[Sketch of Proof of Lemma~\ref{F:LemPos}]

In the following, we denote the positive, classical solution
$h_\varepsilon$ constructed in Lemma \ref{MainAE2} by $h$
(whenever there is no chance of confusion).

Recall the entropy function $G_{0}^{(\alpha)}(z)$ defined by
(\ref{E:reg4}). Using the local entropy inequality (\ref{G:main1})
with $\zeta = \zeta(x)$ from Lemma~\ref{G:lem1}, we obtain
\begin{multline}\label{F:main1-2}
\int\limits_{\Omega} {\zeta^4(x)G^{(\alpha)}_{0}(h(x,T))\,dx}  +
C_1 \iint\limits_{Q_T} { (h^{\tfrac{\alpha + 2}{2}})^2_{xx}\zeta^4
\,dx dt} \leqslant \int\limits_{\Omega}
{\zeta^4(x)G^{(\alpha)}_{0}(h_{0})\,dx} +\\
 C_2 \iint\limits_{Q_T}
{ h^{\alpha +2} (\zeta_{x}^4 + \zeta^2 \zeta_{xx}^2) \,dx dt} +
C_3 \iint\limits_{Q_T} { h^{2(m - n + 1)+ \alpha}\zeta^4 \,dx dt}.
\end{multline}
Due to $h \in L^{\infty}(0,T_{loc}^{(\alpha)}; H^1(\Omega))$, we
deduce from (\ref{F:main1-2}) that
\begin{equation}\label{F:pp7}
\int\limits_{\Omega} {\zeta^4(x) G^{(\alpha)}_{0}(h(x,T))\,dx}
\leqslant \int\limits_{\Omega} {\zeta^4(x) G_{0}(h_{0}(x))\,dx} +
C_4 T \leqslant K < \infty
\end{equation}
for any $T \in [0, T_{loc}^{(\alpha)}]$.
\end{proof}

\section*{B Proof of second moment estimate}
\label{Bapp}
\renewcommand{\thesection}{B}\setcounter{equation}{0}

\begin{proof}[Sketch of Proof of Lemma~\ref{F:lem0}]
Let
\begin{equation}\label{F:g-delta}
\widetilde{G}_{\varepsilon}(z) =  \tfrac{z^{ 2-n}}{  2 - n} -
\tfrac{\varepsilon \,z^{ 2-s}}{s - 2},\
\widetilde{G}'_{\varepsilon}(z) = \tfrac{z }{f_{\varepsilon}(z)},
\end{equation}
$$
\widetilde{G}''_{\varepsilon}(z) = \tfrac{1}{f_{\varepsilon}(z)} -
\tfrac{z f'_{\varepsilon}(z)}{f^2_{\varepsilon}(z)} = (1 - n)
\tfrac{1}{f_{\varepsilon}(z)} - \varepsilon (s - n)z^{ - s} .
$$
Here we use the equality
$$
f'_{\varepsilon}(z) = n z^{-1} f_{\varepsilon}(z) + \varepsilon (s
- n)z^{-(s+1)}f_{\varepsilon}^2(z).
$$
Multiplying (\ref{D:1r'}) for $\delta = 0$ by $x^2
\widetilde{G}'_{\varepsilon}(h)$, and integrating on $Q_T$, yields
\begin{multline}\label{F:d3}
\int\limits_{\Omega} {x^2 \widetilde{G}_{\varepsilon}(h)\,dx} -
\int\limits_{\Omega} {x^2
\widetilde{G}_{\varepsilon}(h_{0\varepsilon})\,dx} = \\
\iint\limits_{Q_T} {f_{\varepsilon}(h) (a_0 h_{xxx} + a_1
D''_{\varepsilon}(h)h_x ) \bigl( 2x \tfrac{h}{f_{\varepsilon}(h)}
+ x^2 \widetilde{G}''_{\varepsilon}(h)
h_x \bigr) \,dx dt} = \\
2 \iint\limits_{Q_T} {x h (a_0 h_{xxx} + a_1
D''_{\varepsilon}(h) h_x)\,dx dt} + \\
(1 - n ) \iint\limits_{Q_T} {x^2   h_x (a_0 h_{xxx} + a_1
D''_{\varepsilon}(h)h_x ) \,dx dt} - \\
\varepsilon (s - n) \iint\limits_{Q_T} {x^2 h^{ - s
}f_{\varepsilon}(h) h_x (a_0 h_{xxx} + a_1 D''_{\varepsilon}(h)h_x
) \,dx dt}  =: I_{1} + I_{2} + I_3.
\end{multline}
We now bound the terms $I_1$, $I_2$ and $I_3$.  First,
\begin{multline}\label{F:d4}
I_{1} = 2a_0\int\limits_{0}^T { x h  h_{xx} \biggl
|_{\partial\Omega}\,dx} - 2a_0\iint\limits_{Q_T} {\{h h_{xx} +
  x\,
h_x h_{xx} \}\,dx dt} + \\
2a_1 \iint\limits_{Q_T} {x h  D''_{\varepsilon}(h) h_x \,dx dt} =
\int\limits_{0}^T { x \bigl( 2a_0 h  h_{xx} - a_0 h_{x}^2 + 2a_1
L_{\varepsilon}(h)\bigr) \biggl |_{\partial\Omega}\,dt} + \\
3a_0 \iint\limits_{Q_T} { h_{x}^2 \,dx dt}
 - 2 a_1 \iint\limits_{Q_T} {L_{\varepsilon}(h) \,dx dt} , \
L_{\varepsilon}(z) := \int\limits_{0}^{z} {\tau \,
D''_{\varepsilon}(\tau)\,d\tau},
\end{multline}
and
\begin{multline}\label{F:d5}
I_{2} = a_0 (1 - n) \iint\limits_{Q_T} {x^2 h_x h_{xxx} \,dx dt} +
a_1( 1 - n) \iint\limits_{Q_T} {x^2 D''_{\varepsilon}(h)
 h_x^2 \,dx dt} = \\
-2a_0( 1-n) \iint\limits_{Q_T} {x  h_x h_{xx} \,dx dt} - a_0( 1 -
n) \iint\limits_{Q_T} {x^2
 h^2_{xx} \,dx dt} + \\
a_1( 1 - n) \iint\limits_{Q_T} {x^2 D''_{\varepsilon}(h) h_x^2
\,dx dt} = - a_0 (1-n)\int\limits_{0}^T { x h_{x}^2 \biggl
|_{\partial\Omega}\,dt} + a_0(1 - n) \iint\limits_{Q_T} {h_x^2
\,dx dt} - \\
a_0(1 - n) \iint\limits_{Q_T} {x^2 h^2_{xx} \,dx dt}
+ a_1(1 - n) \iint\limits_{Q_T} {x^2 D''_{\varepsilon}(h) h_x^2
\,dx dt}.
\end{multline}
Using the H\"{o}lder inequality and (\ref{D:d2}) with $\delta =0$,
we find
\begin{multline*}
I_{3} \leqslant \varepsilon (s - n) \Bigl( \iint\limits_{Q_T}
{f_{\varepsilon}(h)(a_0 h_{xxx} + a_1 D''_{\varepsilon}(h)h_x )^2
\,dx dt} \Bigr)^{\tfrac{1}{2}} \times \\
 \Bigl( \iint\limits_{Q_T}
{x^4 h^{-2s}f_{\varepsilon}(h) h_x^2 \,dx dt} \Bigr)^{\tfrac{1}{2}}
\leqslant \varepsilon \, C \Bigl( \iint\limits_{Q_T} {
\tfrac{h_x^2}{(h^{s-n} + \varepsilon )^2}  \,dx dt}
\Bigr)^{\tfrac{1}{2}} .
\end{multline*}
Using the Young's inequality $a b \leqslant \tfrac{a^p}{p} +
\tfrac{b^q}{q}, \ \tfrac{1}{p} + \tfrac{1}{q} =1$ ($ \Rightarrow p
\,a b \leqslant a^p + \tfrac{p}{q} b^q = a^p + (p - 1) b^q$) with
$a = z^{\tfrac{s - n}{p}}$ and $b = \bigl( \tfrac{\varepsilon}{p -
1} \bigr)^{\tfrac{1}{q}}$, we deduce
$$
I_{3} \leqslant \varepsilon^{\tfrac{q-1}{q} } \, C \Bigl(
\iint\limits_{Q_T} { h^{- \tfrac{2(s-n)}{p}} h_x^2 \,dx dt}
\Bigr)^{\tfrac{1}{2}} = \varepsilon^{\tfrac{q-1}{q} } \, \widetilde{C}
\Bigl( \iint\limits_{Q_T} { (h^{\tfrac{p - s + n}{p}} )_x^2 \,dx
dt} \Bigr)^{\tfrac{1}{2}},
$$
choosing $p = - \tfrac{s - n}{\alpha} > 1$ and $q = \tfrac{2(s - n)
}{2(s -n) + \alpha}  > 1$ ($\Rightarrow 0 > \alpha > - 2( s -
n)$), due to (\ref{E:b13}), we find
\begin{multline}\label{F:d6}
I_{3} \leqslant \varepsilon^{\tfrac{q-1}{q} } \, \widetilde{C} \Bigl(
\iint\limits_{Q_T} { (h^{\tfrac{p - s + n}{p}} )_x^2 \,dx dt}
\Bigr)^{\tfrac{1}{2}} = \\
\varepsilon^{- \tfrac{\alpha}{s-n} } \,
\widetilde{C} \Bigl( \iint\limits_{Q_T} { (h^{\tfrac{\alpha + 2}{2}}
)_x^2 \,dx dt} \Bigr)^{\tfrac{1}{2}} \leqslant C_1 \varepsilon^{-
\tfrac{\alpha}{s-n} },
\end{multline}
where the constant $C_1 > 0$ is independent of $\varepsilon$.
Hence,
\begin{equation}\label{F:d6'}
\lim_{\varepsilon \to 0}{I_3} \leqslant 0.
\end{equation}

From (\ref{F:d3})--(\ref{F:d5}), we deduce that
\begin{multline}\label{F:d7}
\int\limits_{\Omega} {x^2 \widetilde{G}_{\varepsilon}(h)\,dx} =
\int\limits_{\Omega} {x^2 \widetilde{G}_{\varepsilon}
(h_{0\varepsilon})\,dx} + \\
\int\limits_{0}^T { x \bigl( 2a_0 h h_{xx} - a_0(2  - n)  h_{x}^2
+ 2a_1 L_{\varepsilon}(h)\bigr) \biggl |_{\partial\Omega}\,dt} +
a_0(4  - n) \iint\limits_{Q_T} { h_x^2 \,dx dt}   -\\
a_0( 1 - n) \iint\limits_{Q_T} {x^2   h^2_{xx} \,dx dt}   + a_1( 1
- n) \iint\limits_{Q_T} {x^2 D''_{\varepsilon}(h)
 h_x^2 \,dx dt}  - \\
2 a_1 \iint\limits_{Q_T} {L_{\varepsilon}(h) \,dx dt} + I_3.
\end{multline}
Letting $\varepsilon \to 0$ in (\ref{F:d7}), due to (\ref{F:d6'}),
we find
\begin{multline}\label{F:d9}
\int\limits_{\Omega} {x^2 \widetilde{G}_{0} (h)\,dx} \leqslant
\int\limits_{\Omega} {x^2 \widetilde{G}_{0} (h_{0})\,dx} + a_0(4 - n)
\iint\limits_{Q_T} { h_x^2 \,dx
dt} - \\
2a_1 (  m - n + 1) \iint\limits_{Q_T} {D_{0}(h)\,dx dt} +
\int\limits_{0}^T { x \bigl( 2a_0 h  h_{xx} - a_0(2  - n) h_{x}^2
+ 2a_1 L_{0}(h)\bigr) \biggl |_{\partial\Omega}\,dt} + \\
a_0(n - 1 ) \iint\limits_{Q_T} {x^2 h^2_{xx} \,dx dt} + a_1(1 - n
) \iint\limits_{Q_T} {x^2
h^{ m - n} h_x^2 \,dx dt}   =\\
\int\limits_{\Omega} {x^2 \widetilde{G}_{0} (h_{0})\,dx} +  2(4
  -n ) \int\limits_{0}^T {\mathcal{E}_{0}(t)
\,dt} - 2\,a_1(m - 3 )
\iint\limits_{Q_T} {D_{0}(h)\,dx dt} + \\
\int\limits_{0}^T { x \bigl( 2a_0 h h_{xx} - a_0(2  - n) h_{x}^2 +
2a_1 L_{0}(h)\bigr) \biggl |_{\partial\Omega}\,dt} + \\
a_0(n - 1 ) \iint\limits_{Q_T} {x^2  h^2_{xx} \,dx dt} + a_1(1 - n
 ) \iint\limits_{Q_T} {x^2
h^{ m - n} h_x^2 \,dx dt}.
\end{multline}
Due to
\begin{multline*}
\iint\limits_{Q_T} {x^2 h^{ m -n} h_x^2 \,dx dt} = -2
\int\limits_{0}^T { x D_{0}(h) \biggl |_{\partial\Omega}\,dt} + 2
\iint\limits_{Q_T} {D_{0}(h)\,dx dt} - \\
\iint\limits_{Q_T} {x^2 D'_{0}(h) h_{xx} \,dx dt},
\end{multline*}
from (\ref{F:d9}), in view of (\ref{C:d2'}), we deduce that
\begin{multline}\label{F:d8}
\int\limits_{\Omega} {x^2 \widetilde{G}_{0} (h)\,dx} + 2a_1(m + n - 4)
\iint\limits_{Q_T} {D_{0}(h)\,dx dt} \leqslant
\int\limits_{\Omega} {x^2 \widetilde{G}_{0} (h_{0})\,dx} + \\
k_1 \mathcal{E}_{0}(0)T  - (1 - n) \iint\limits_{Q_T} {x^2 h_{xx}(
a_0
h_{xx} + a_1  D'_0(h)) \,dx dt} + W(T) = \\
 \int\limits_{\Omega} {x^2 \widetilde{G}_{0} (h_{0})\,dx} + k_1 \mathcal{E}_{0}(0)T  -
\tfrac{1 - n}{a_0} \iint\limits_{Q_T} {x^2  (a_0 h_{xx} + a_1
D'_0(h))^2 \,dx
dt} + \\
\tfrac{a_1(1 - n)}{a_0} \iint\limits_{Q_T} {x^2 D'_0(h)(a_0 h_{xx}
+ a_1 D'_0(h)) \,dx dt} + W(T),
\end{multline}
where $k_1 = 2(4 -n)$, and $W(T)$ is from (\ref{F:W}). From
(\ref{F:d8}) we find that
\begin{multline}\label{F:d8-1}
\int\limits_{\Omega} {x^2 \widetilde{G}_{0} (h)\,dx} + 2a_1(m + n - 4)
\iint\limits_{Q_T} {D_{0}(h)\,dx dt} + \\
\tfrac{ 1 - n}{2a_0} \iint\limits_{Q_T} {x^2  ( a_0 h_{xx} + a_1
D'_0(h))^2 \,dx dt} \leqslant \int\limits_{\Omega} {x^2
\widetilde{G}_{0} (h_{0})\,dx} + k_1 \mathcal{E}_{0}(0)T + \\
\tfrac{a_1^2( 1 - n)}{2a_0} \iint\limits_{Q_T} {x^2 (D'_0(h))^2
\,dx dt} + W(T) \leqslant \int\limits_{\Omega} {x^2 \widetilde{G}_{0}
(h_{0})\,dx} + k_1 \mathcal{E}_{0}(0)T + \\
 \int \limits_{0}^T {
\Bigl(\widetilde{A} (t) \int \limits_{\Omega} {x^2 \widetilde{G}_{0} (h)
\,dx} \Bigr) dt} + W(T),
\end{multline}
where $m \geqslant 4 - n$, and
$$
\widetilde{A} (t) := \tfrac{a_1^2 |1 - n|(  2- n)}{2a_0( m
-n+1)^2}\|h(.,t)\|_{L^{\infty}(\Omega)}^{ 2m -n}.
$$
Applying the nonlinear Gr\"onwall lemma \cite{Bihari} to
$$
v(T) \leqslant v(0) + k_1 \mathcal{E}_{0}(0)T + \int \limits_{0}^T
{ \widetilde{A}(s) v(s) \, dt} + W(T)
$$
with $v(T) = \int\limits_{\Omega} {x^2 \widetilde{G}_{0}
(h(x,T))\,dx}$ yields
$$
v(T) \leqslant e^{\widetilde{B}(T)}\Bigl( v(0) + \int \limits_{0}^T {
(k_1 \mathcal{E}_{0}(0) + W'(t)) e^{-\widetilde{B}(t)} dt}
 \Bigr)
$$
with $\widetilde{B}(T):= \int\limits_0^T {\widetilde{A}(\tau)\,d\tau}$.
From this we obtain (\ref{F:d9-01}) for $0 < n \leqslant 1$.

In the case of $1 < n < 2$, from (\ref{F:d8}) we find that
\begin{multline}\label{F:d8-3}
\int\limits_{\Omega} {x^2 \widetilde{G}_{0} (h)\,dx} + 2a_1(m + n - 4)
\iint\limits_{Q_T} {D_{0}(h)\,dx dt} \leqslant
\int\limits_{\Omega} {x^2 \widetilde{G}_{0} (h_{0})\,dx} + k_1
\mathcal{E}_{0}(0)T + \\
\tfrac{3a_0(n - 1)}{2} \iint\limits_{Q_T} {x^2 h_{xx}^2 \,dx dt} +
\tfrac{a_1^2(n - 1)}{2a_0} \iint\limits_{Q_T} {x^2 (D'_0(h))^2
\,dx dt} + W(T) \leqslant \int\limits_{\Omega} {x^2
\widetilde{G}_{0} (h_{0})\,dx} + \\
k_1 \mathcal{E}_{0}(0)T + k_2 \iint\limits_{Q_T} {x^2 h_{xx}^2
\,dx dt} + \int \limits_{0}^T { \Bigl( \widetilde{A}(t) \int
\limits_{\Omega} {x^2 \widetilde{G}_{0} (h) \,dx} \Bigr) dt}+ W(T) ,
\end{multline}
where $m \geqslant  4 - n$, $k_2 = \tfrac{3a_0(n - 1)}{2}$, and
$W(t)$ is from (\ref{F:W}). Applying the Gr\"{o}nwall lemma
\cite{Bihari} to (\ref{F:d8-3}), we obtain the estimate
(\ref{F:d9-01}) for $1 < n < 2$.

\end{proof}

\section*{Appendix C}

\renewcommand{\thesection}{D}\setcounter{lemma}{0}

\begin{lemma}\label{A.1}
(\cite{Lions}) Suppose that $X,\ Y,$  and  $Z$ are Banach spaces,
$X \!\Subset \!Y \subset \! Z$, and $X$ and $Z$ are reflexive.
Then the imbedding $ \{ u \in \! L^{p_0 } (0,T;$ $X):$ $\partial
_t u \in L^{p_1 } (0,T;Z),1 < p_i < \infty ,i = 0,1 \} \Subset
L^{p_0 } (0,T;Y)$ is compact.
\end{lemma}

\begin{lemma}\label{A.2}
(\cite{Sim}) Suppose that $X,\ Y,$  and  $Z$ are Banach spaces and
$X \!\Subset \!Y\hspace{-0.2cm} \subset \! Z$. Then the imbedding
$ \{u \in L^\infty (0,T;X):\partial _t u \in L^p (0,T;Z)$, $p > 1
\} \Subset C(0,T;Y)$ is compact.
\end{lemma}

\begin{lemma}\label{A.3}
(\cite{G4,BSS}) Let $\Omega \subset \mathbb{R}^N,\ N < 6$, be a
bounded convex domain with smooth boundary, and let $n \in \bigl(2
- \sqrt{1 -\tfrac{N}{N + 8}}, 3 \bigr)$ for $N > 1$, and
$\tfrac{1}{2} < n < 3$ for $N = 1$. Then the following estimates
hold for any strictly positive functions $v \in H^2(\Omega)$ such
that $\nabla v \cdot \vec{n} = 0$ on $\partial\Omega$ and
$\int\limits_{\Omega}{v^n |\nabla\Delta v|^2} < \infty$:
\begin{equation*}
\int\limits_{\Omega} {\varphi ^6 \{ v^{n - 4}|\nabla v|^6 + v^{n -
2} |D^2 v|^2 |\nabla v|^2 \}} \leqslant c
\Bigl\{\int\limits_{\Omega} {\varphi ^6 v^n |\nabla \Delta v|^2} +
\int\limits_{\{ \varphi > 0\} } {v^{n + 2} |\nabla \varphi|^6}
\Bigr\},
\end{equation*}
\begin{multline*}
\int\limits_{\Omega} {\varphi ^6 |\nabla \Delta v^{\tfrac{n +
2}{2}}|^2 } \leqslant c \Bigl\{\int\limits_{\Omega} {\varphi ^6 v^n |\nabla \Delta v|^2} + \\
+ \int\limits_{\{ \varphi
> 0\} } {v^{n + 2} \{|\nabla \varphi|^6 + \varphi^2 |D^2 \varphi|^2
|\nabla \varphi|^2 + \varphi^3 | \Delta \varphi|^3 \}} \Bigr\},
\end{multline*}
where $\varphi \in C^2 (\Omega)$ is an arbitrary nonnegative
function such that the tangential component of $\nabla \varphi$ is
equal to zero on $\partial\Omega$, and the constant $c > 0$ is
independent of $v$.
\end{lemma}

\begin{lemma}\label{A.4}
(\cite{N1})  If $\Omega  \subset \mathbb{R}^N $ is a bounded
domain with piecewise-smooth boundary, $a > 1$, $b \in (0, a),\ d
> 1,$ and $0 \leqslant i < j,\ i,j \in \mathbb{N}$, then there
exist positive constants $d_1$ and $d_2$ $(d_2 = 0 \text{ if }
\Omega$ is unbounded$)$ depending only on $\Omega ,\ d,\ j,\ b,$
and $N$ such that the following inequality is valid for every
$v(x) \in W^{j,d} (\Omega ) \cap L^b (\Omega )$:
$$
\left\| {D^i v} \right\|_{L^a (\Omega )}  \leqslant d_1 \left\|
{D^j v} \right\|_{L^d (\Omega )}^\theta  \left\| v \right\|_{L^b
(\Omega )}^{1 - \theta }  + d_2 \left\| v \right\|_{L^b (\Omega )}
,\ \theta  = \tfrac{{\tfrac{1} {b} + \tfrac{i} {N} - \tfrac{1}
{a}}} {{\tfrac{1} {b} + \tfrac{j} {N} - \tfrac{1} {d}}} \in \left[
{\tfrac{i} {j},1} \right)\!\!.
$$
\end{lemma}

\begin{lemma}\label{A.5}
(\cite{G2}) Let $(\beta_1,\ldots,\beta_m) \in \mathbb{R}^m, m
\geqslant 1 $, and let $\beta = \prod \limits_{j=1}^m \beta_j,\
\overline{\beta_i} = \tfrac{\beta}{\beta_i}= \prod \limits_{j=1, j
\neq i}^m \beta_j$. Assume that $G_i(s)$ are nonnegative
nonincreasing functions satisfying the conditions
$$
G_i(s + \delta ) \leqslant c_i \Bigl({\sum \limits_{i=1}^m
\tfrac{G_i(s)}{\delta^{\alpha_i}} }\Bigr)^{\beta_i}\ \forall\, s >
0,\ \delta > 0,\ i= \overline{1,m}
$$
with real constants $c_i > 0,\ \beta_i >1$, and $\alpha_i
\geqslant 0$ for $i=\overline{1,m}$, and $\alpha_i > 0$ for $i =
\overline{1,\ell}$. Let $ G(s)\! = \!\sum \limits_{i=1}^m
\!{\overline{(c_i^{\overline{\beta}_i})}
\left({G_i(s)}\right)^{\overline{\beta}_i}}$, and let the function
$H(s) = m^{\beta} \sum \limits_{i=\ell + 1}^m \!{
c_i^{\overline{\beta}_i} \overline{(c_i^{\overline{\beta}_i})}^{1
- \beta_i} \!\!\!\!\!\!\!\!\! \left({G_i(s)}\right)^{\beta_i - 1}}
$ be such that $H(s_1) < 1$ at a some $s_1 \geqslant 0$. Then
there exists a positive constant $c > 1$ depending on $m,\
\alpha_i,\ \beta_i,\ \ell $, and $H(s_1)$ such that $G_i(s_0)
\equiv 0$ for all $i=\overline{1,\ell}$, where $ s_0 = s_1 + c
\sum \limits_{i= 1}^\ell {\bigl( {
c_i^{\overline{\beta}_i}{\overline{(c_i^{\overline{\beta}_i})}}^{1-\beta_i}
\left({G(s_1)}\right)^{\beta_i - 1}}\bigr) ^{\tfrac {1} {\alpha_i
\beta}}}$.
\end{lemma}


\def\cprime{$'$} \def\cprime{$'$} \def\cprime{$'$} \def\cprime{$'$}
  \def\cprime{$'$}

\end{document}